\documentclass[11pt]{amsart}
\usepackage[english]{babel}
\usepackage{graphicx,amsmath}
\usepackage{enumitem}
\pdfoutput=1
\theoremstyle{plain}
\newtheorem{theorem}{Theorem}[section]
\newtheorem{remark}{Remark}[section]
\newtheorem{proposition}{Proposition}[section]
\newtheorem{lemma}{Lemma}[section]
\newtheorem{definition}{Definition}[section]
\newtheorem{example}{Example}[section]
\newtheorem{corollary}{Corollary}[section]

\newcommand\mycom[2]{\genfrac{}{}{0pt}{}{#1}{#2}}
\title[On a family of KP multi--line solitons]{On a family of KP multi--line solitons associated to rational degenerations of real hyperelliptic curves and to the finite non--periodic Toda hierarchy}
\author{Simonetta Abenda}
\address{Dipartimento di Matematica, Universit\`a di Bologna, P.zza di Porta San Donato 5, I-40126 Bologna BO, ITALY
}
\email{simonetta.abenda@unibo.it
}
\thanks{This work has been partially supported by PRIN2010-11 ``Teorie geometriche e analitiche dei sistemi Hamiltoniani in dimensioni finite e infinite''.}

\begin{document}

\begin{abstract}
{We continue the program started in \cite{AG} of associating rational degenerations of $\mathtt M$--curves to points in $Gr^{\mbox{\tiny TNN}}(k,n)$ using KP theory for real finite gap solutions. More precisely, we focus on the inverse problem of characterizing the soliton data which produce Krichever divisors compatible with the KP reality condition when $\Gamma$ is a   certain rational degeneration of a hyperelliptic $\mathtt M$--curve. Such choice is motivated by the fact that $\Gamma$  is related to the curves associated to points in $Gr^{\mbox{\tiny TP}}(1,n)$ and in $Gr^{\mbox{\tiny TP}}(n-1,n)$ in \cite{AG}. We prove that the reality condition on the Krichever divisor on $\Gamma$ singles out a special family of KP multi--line solitons ($T$--hyperelliptic solitons) in $Gr^{\mbox{\tiny TP}} (k,n)$, $k\in [n-1]$, naturally connected to the finite non-periodic Toda hierarchy. We discuss the relations between the algebraic-geometric description of KP $T$--hyperelliptic solitons and of the open Toda system.
Finally, we also explain the effect of the space--time transformation which conjugates soliton data in $Gr^{\mbox{\tiny TP}} (k,n)$ to soliton data in $Gr^{\mbox{\tiny TP}} (n-k,n)$ on the Krichever divisor for such KP solitons. }


\end{abstract}
\maketitle

\tableofcontents
\section{
Introduction}

Regular bounded KP $(n-k,k)$-line solitons are associated to soliton data $({\mathcal K}, [A])$ with ${\mathcal K}= \{\kappa_1<\cdots < \kappa_n\}$ and  $[A]\in Gr^{\mbox{\tiny TNN}} (k,n)$, the totally non--negative part of the real Grassmannian, which is a reduction of the infinite dimensional Sato Grassmannian \cite{S}. The asymptotic properties of such solitons have been successfully related to the combinatorial structure of $Gr^{\mbox{\tiny TNN}} (k,n)$ in \cite{CK,KW1, KW2}.

On the other side, in principle, such soliton solutions may be obtained assigning Krichever data, which satisfy the KP reality condition, on rational degenerations of regular $\mathtt M$--curves\cite{DKN,DN,Kr1,Kr2}.

In \cite{AG}, we have started the program of connecting such two areas of mathematics - the theory of totally positive
Grassmannians and the rational degenerations of regular $\tt M$-curves -
using the real finite--gap theory for regular bounded KP $(n-k,k)$-line solitons: 
we have associated to any soliton data $({\mathcal K}, [A])$ with $[A]\in Gr^{\mbox{\tiny TP}} (k,n)$, 
a curve $\Gamma$ which is the rational degeneration of a regular $\mathtt M$-curve of genus $g=k(n-k)$ and a Krichever divisor ${\mathcal D}$ compatible with the reality conditions settled in \cite{DN}.

In the present paper, we focus on the inverse problem: we choose $\Gamma$ a given rational degeneration of an $\mathtt M$--curve, we fix the point $P_+\in \Gamma$ where the KP wavefunction has its essential singularity and a local coordinate $\zeta$ on $\Gamma$ such that $\zeta^{-1} (P_+)=0$, and we use the reality condition for the Krichever divisor to classify the regular bounded KP $(n-k, k)$--line solitons compatible with such algebraic--geometric setting, $(\Gamma, P_+, \zeta)$. 

\begin{figure}
\includegraphics[scale=0.25]{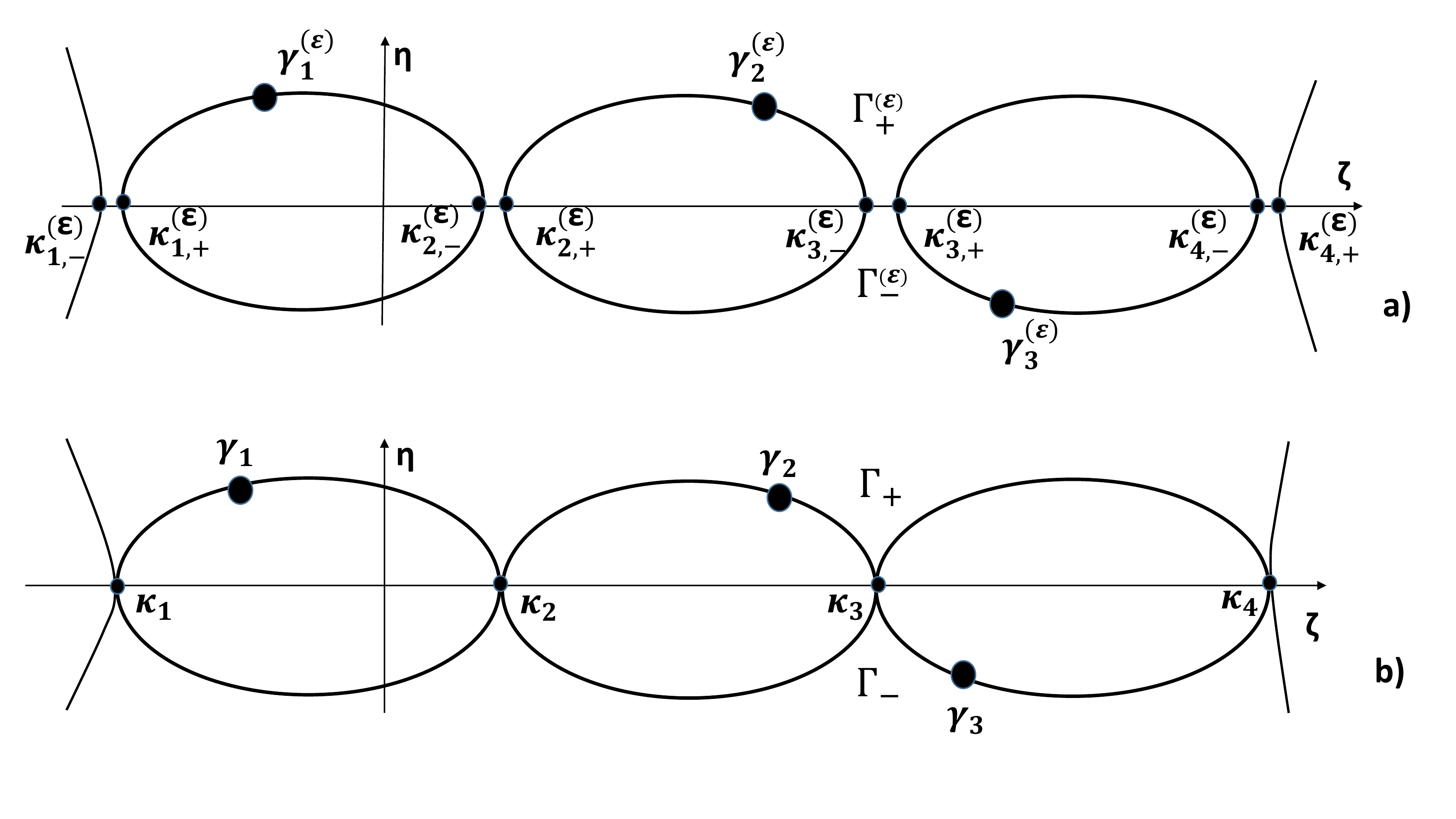}
\vspace{-0.6 truecm}
\caption{{\small{\sl Figure 1.a): KP regular real quasi--periodic solutions on the genus 3 hyperelliptic curve $\Gamma^{(\epsilon)}=\Gamma^{(\epsilon)}_+\cup \Gamma^{(\epsilon)}_-$ are parametrized by 3--point divisors such that there is exactly one divisor point $\gamma^{(\epsilon)}_j$, $j\in [3]$ in each finite oval. Figure 1.b): In the limit $\epsilon\to 0$, $\Gamma^{(\epsilon)}$ rationally degenerates to $\Gamma=\Gamma_+\cup \Gamma_-$ and the limiting KP solution is a real bounded regular KP multi--soliton solution.}}}
\label{fig:epsilon}
\end{figure}

More precisely, we successfully investigate the case where $\Gamma$ is obtained in the limit $\epsilon\to 0$ from regular real hyperelliptic curves with affine part 
$\Gamma^{(\epsilon)} =\{ (\zeta, \eta) \; : \;
\eta^2 =  -\epsilon^2 +\prod_{j=1}^n (\zeta - \kappa_j)^2 \}$, 
{\sl i.e.}
\begin{equation}\label{eq:G}
\Gamma = \Gamma_+ \sqcup \Gamma_- = \{ (\zeta, \eta) \; : \; \eta^2 = \prod\limits_{m=1}^n (\zeta - \kappa_m)^2 \}.
\end{equation}
We choose $\Gamma$ as in (\ref{eq:G}) because it is a desingularization of the curve associated in \cite{AG} to soliton data $({\mathcal K}, [A])$, with $[A]\in Gr^{\mbox{\tiny TP}}(n-1,n)$. Moreover, by a straightforward modification of the construction presented in \cite{AG}, $\Gamma$ is also associated to soliton data $({\mathcal K}, [A])$, with $[A]\in Gr^{\mbox{\tiny TP}}(1,n)$.
We make the following ansatz: 
\begin{enumerate}
\item the number of phases $n$ and the number $k$ of divisor points belonging to the intersection of the finite ovals with $\Gamma_+$, the copy of ${\mathbb CP}^1$ containing the essential singularity of KP wavefunction, identifies the Sato finite dimensional reduction $Gr (k,n)$ corresponding to the KP solutions;
\item the arithmetic genus of $\Gamma$, $n-1$, equals the dimension of the divisor and the dimension of the subvariety in $Gr (k,n)$ described by such KP soliton solutions;
\item the real boundedness and regularity of the KP soliton solution due to the algebraic geometric data implies that the soliton data are realizable in $Gr^{\mbox{\tiny TNN}} (k,n)$.
\end{enumerate}
The above ansatz is compatible both with real KP finite--gap theory, with Sato dressing of the vacuum with finite dimensional operators and the characterization of real regular bounded $(n-k,k)$--line solitons. 

\smallskip

In the first part of the paper, we review some known facts about $(n-k,k)$--line soliton KP solutions and the finite--gap setting (Section  \ref{sec:introKP}), we justify the above ansatz by explicitly characterizing the KP--soliton data producing such divisor structure and call such KP--solitons $T$--hyperelliptic (Sections \ref{sec:vacuumKP}  and \ref{sec:Tsol}) and explain the relation with the results in \cite{AG} in the case $Gr^{\mbox{\tiny TP}} (n-1,n)$ (Section \ref{sec:desing}).  
The main results of this part are: 
\begin{enumerate}
\item We identify the points in $Gr^{\mbox{\tiny TNN}} (k,n)$ which correspond to real regular bounded KP $(n-k,k)$--line solitons with algebraic geometric data on $\Gamma$;
\item We prove that the divisor structure on $\Gamma$ is compatible with the KP reality condition if and only if the soliton data $({\mathcal K}, [A])$ correspond to points $[A]\in Gr^{\mbox{\tiny TP}}(k,n)$, with
\[
A^i_j = a_j k_j^{i-1}, \quad i\in [k], \, j\in [n], \quad\quad  [a]\in Gr^{\mbox{\tiny TP}}(1,n).
\]  
We call such KP solitons $T$--hyperelliptic and explicitly construct the KP--wavefunction on $\Gamma$. $k$ divisor points belong to the intersection of the real ovals with $\Gamma_+$ and the remaining $(n-k-1)$ to the intersection of the real ovals with $\Gamma_-$. For instance the divisor structure shown in picture \ref{fig:epsilon}.b) corresponds to soliton data ${\mathcal K} = \{ \kappa_1 < \cdots < \kappa_4\}$ and $[A]\in Gr^{\mbox{\tiny TP}}(2,4)$ as above.
\end{enumerate}

\smallskip

The special form of the KP $\tau$--function associated to $T$--hyperelliptic solitons relates such class of solutions to the finite non--periodic Toda system\cite{Mos}. Since $\Gamma$ is related to the algebraic--geometric description of the finite non--periodic Toda hierarchy\cite{Mar,BM,KV}, it is then natural to investigate the relations between the algebraic--geometric approach for the two systems. The asymptotics of the KP wavefunction and of the Toda Baker--Akhiezer functions are different at the essential singularity $P_+$ since they are modeled respectively on regular finite gap KP theory and on the periodic Toda system. 

In the second part of the paper we thouroughly investigate such relations and we prove that the divisor structure of the two systems are connected. 
In section \ref{sec:introToda} we review known results on the finite non--periodic Toda system to settle notations. In section \ref{sec:TodaDarboux}, we introduce two sequences of Darboux transformations and express the Toda Baker--Akhiezer functions of \cite{KV} using the the Toda resolvent. In section \ref{sec:TodaKP}, we explain the relations between the divisor structure of KP $T$--hyperelliptic solitons and the spectral problem for the Toda system. In section \ref{sec:invKP} we solve the inverse problem of reconstructing the KP soliton data from a $k$--compatible divisor and express the Toda hierarchy solutions in function of the zero--divisor dynamics of the Toda Baker--Akhiezer functions. 
The main results for this part are:
\begin{enumerate}
\item The KP vacuum divisor associated here to $T$--hyperelliptic solitons is the Toda divisor found in \cite{KV};
\item The Darboux transformations generating $T$--hyperelliptic $(n-k,k)$--line solitons are associated to well--known recurrencies for the Toda system;
\item The pole divisor of the normalized KP $(n-k,k)$--line $T$--hyperelliptic soliton wavefunction coincides with the zero divisor at times $\vec t\equiv \vec 0$ of the $k$--th component of the Toda Baker--Akhiezer function defined in \cite{KV}.
\end{enumerate}

\smallskip

In the last part of the paper we discuss the duality correspondence between KP $T$--hyperelliptic soliton data in $Gr(k,n)$ and in $Gr(n-k,n)$ induced by space--time inversion from the algebraic--geometric point of view. For the Toda system, such transformation corresponds to the composition of space--time inversion with the reflection of the entries of the Toda Lax matrix ${\mathfrak A}$ with respect to the antidiagonal.  
In section \ref{sec:dual}, we discuss the relation between space time--inversion, duality in Grassmann cells and divisors of dual $T$--hyperelliptic solitons and its relation to Toda. In particular, we give the explicit formula to compute the divisor of the dual soliton in $Gr^{\mbox{\tiny TP}}(n-k,n)$ from the soliton data in $Gr^{\mbox{\tiny TP}}(k,n)$.

In section \ref{sec:summary}, we summarize the results of the paper.
We are convinced that the results presented in this paper may be generalized in many directions and open the way to novel interpretations of the KP wavefunctions associated to generic points in $Gr^{\mbox{\tiny TP}} (k,n)$\cite{AG} and its generalization to the whole $Gr^{\mbox{\tiny TNN}} (k,n)$\cite{AG1}. 

\section{$(n-k,k)$-line solitons via Darboux transformation, in the Sato Grassmannian and in finite--gap theory}\label{sec:introKP}

The KP-II equation\cite{KP}
\begin{equation}\label{eq:KP}
(-4u_t+6uu_x+u_{xxx})_x+3u_{yy}=0,
\end{equation}
is the first non--trivial flow of an integrable hierarchy \cite{D,DKN,H,MJD,S} and 
in the following we denote $\vec t = (t_1=x,t_2=y,t_3=t, t_4,\dots)$. 
In this section, we characterize the real bounded regular $(n-k,k)$-line soliton solutions in the general class of KP--soliton
solutions via Darboux transformations, Sato's dressing transformations and finite gap--theory.

For any $k,n \in \mathbb N$ with $k<n$, we denote $[k,n] = \{ k, k+1, \dots, n-1, n\}$ and $[n] =[1,n]$. Following Postnikov \cite{Pos}, a $k\times n$ real matrix $A\in Mat^{\mbox{\tiny TNN}}_{k,n} $ if all the maximal ($k\times k$) minors of $A$ are non--negative and at least
one of them is positive. The totally non--negative Grassmannian is
$Gr^{\mbox{\tiny TNN}} (k,n)= GL_k^+ \backslash Mat^{\mbox{\tiny TNN}}_{k,n}$, where $GL_k^+$ are the $k\times k$ real matrices with positive determinant.
The totally positive Grassmannian is $Gr^{\mbox{\tiny TP}}(k,n)= {\mathcal S} \cap Gr^{\mbox{\tiny TNN}} (k,n)$,  where $\mathcal S$ is the top cell in the
Gelfand--Serganova decomposition of $Gr(k,n)$, {\sl i.e.} $[A] \in  Gr^{\mbox{\tiny TP}}(k,n)$ if and only if all maximal ($k\times k$) minors of $A$ are positive.

The simplest way to construct KP solitons is via the Wronskian method \cite{FN,Mat1}:
suppose that $f^{(1)} (\vec t),\dots,f^{(k)} (\vec t)$ satisfy the heat hierarchy
$\partial_{t_l} f^{(r)} = \partial_x^l f^{(r)}$, with $l\ge 1$, $r\in [k]$,
and let
$\tau (\vec t) ={\rm Wr_x} (f^{(1)},\dots,f^{(k)}) $.
Then
$u (\vec t) = 2 \partial_x^2 \log({\tau(\vec t)})$,
is a solution to KP-II.
Let ${\mathcal K} =\{\kappa_1 < \kappa_2 < \cdots <\kappa_n\}$ and $A =( A^r_j)$ be a real $k\times n$ matrix. The $(n - k,k)$--line soliton solutions are obtained choosing
\begin{equation}\label{eq:f} 
f^{(r)} ( \vec t ) = \sum_{j=1}^n
A^r_j E_j (\vec t),
 \quad r=1,\dots,k,
\quad\quad 
E_j (\vec t)=e^{\theta( \kappa_j; \vec t)},  \; \mbox{ with } \; \theta(\zeta ; \vec t) =\sum\limits_{n \ge 1} \zeta^n t_n.
\end{equation}
In such a case $\tau (\vec t) = \sum_{1\le i_1
<\cdots <i_k \le n} {\Delta (i_1,\dots,i_k)} E_{[i_1,
\dots,i_k]} (\vec t)$, where $\Delta (i_1,\dots,i_k)$ are the
Pl\"ucker coordinates of the corresponding point in the real Grassmannian,
$[A]\in Gr(k,n)$, and $E_{[i_1,
\dots,i_k]} (\vec t) ={\rm Wr}_x
(E_{i_1},\dots,E_{i_k})$.
Then, following \cite{KW1}, the $(n - k,k)$--line soliton $u(\vec t)$ is regular and bounded for all  $\vec t=(x,y,t,0,\dots)$  if and only if
$[A] \in Gr^{\mbox{\tiny TNN}} (k,n)$,  {\it i.e.} all $k\times k$ minors
$\Delta (i_1,\dots,i_k) \ge 0$.

The KP solitons are also realized as special solutions in the Sato theory of the KP hierarchy \cite{S,MJD} using the Dressing transformation.
Indeed let the vacuum hierarchy be
\[
{\partial_x} {\Psi^{(0)}} =
\lambda{\Psi^{(0)}}, \quad\quad
\displaystyle \partial_{t_n} {\Psi^{(0)}} = {\partial_x^n}
{\Psi^{(0)}} = \lambda^n {\Psi^{(0)}}, \quad n\ge 1,
\]
and suppose that the dressing operator $W =1-w_1\partial_x^{-1} -w_2\partial_x^{-2} - \cdots $
satisfies the Sato equations $\partial_{t_n} W = (W\partial_x^n W^{-1})_+ W -
W \partial_x^n$, $n\ge 1$, where the symbol $(\cdot)_+$ denotes the differential part of the given operator. Then the KP hierarchy is generated
by the inverse gauge (dressing) transformation $L= W
\partial_x W^{-1}$
\[
\displaystyle {L}{\tilde\Psi}^{(0)} = \lambda
{\tilde\Psi}^{(0)},\quad\quad
\displaystyle \partial_{t_n} {\tilde\Psi}^{(0)}  =
B_n {\tilde\Psi}^{(0)},  \quad B_n =  (W\partial_x^n W^{-1})_+,\quad n\ge 1,
\]
with ${\tilde\Psi}^{(0)} =W \Psi^{(0)}$ and $x=t_1$, $y=t_2$, $t=t_3$. In such a case the Lax operator takes the form $\displaystyle
L =\partial_x + u_2 \partial_x^{-1} +u_3 \partial_x^{-2} + \cdots
$, and $u_2 = \partial_x w_1$ satisfies the KP equation.
$u(x,y,t)$ is the $(n - k,k)$--line soliton associated to soliton data $({\mathcal K}, [A])$, if and only if the dressing operator
takes the form
$W = 1-w_1\partial_x^{-1} -w_2\partial_x^{-2} - \cdots - w_k \partial_x^{-k}$,
and $Df^{(r)}=0$, $r=1,\dots,k$, where $D$ is the Darboux transformation\cite{Mat}
\begin{equation}\label{eq:Darboux}
D\equiv W\partial_x^k=\partial_x^{k}-w_1(\vec t) \partial_x^{k-1}- \ldots - w_k(\vec t).
\end{equation}

Regular finite--gap solutions are the complex  periodic or quasi--periodic meromorphic solutions to the KP equation (\ref{eq:KP}). Krichever \cite{Kr1,Kr2}
has classified this class of solutions: for any non-singular genus $g$ complex algebraic curve $\Gamma$ with a marked point $P_0$ and a local parameter
$\lambda$ such that $\lambda^{-1} (P_0)=0$, there exists a family of regular complex finite--gap solutions $u(\vec t)$ to  (\ref{eq:KP}) parametrized by
non special divisors $\mathcal{D}=(P_1,\dots, P_g)$ on $\Gamma \backslash \{P_0	\}$. More precisely, the Baker--Akhiezer function ${\tilde \Psi} (P; \vec t )$ meromorphic on $\Gamma
\backslash \{ P_0 \}$ with poles on $\mathcal{D}$ and an essential singularity at $P_0$ with the following asymptotics
$
\displaystyle {\tilde \Psi}(\zeta; \vec t) = \big( 1 + \frac{\chi_1(\vec t)}{\zeta} + O(\zeta^{-2}) \big) e^{\zeta x + \zeta^2 y + \zeta^3 t + \cdots}$, as $(\zeta\to\infty)$,
is a solution to
$\displaystyle
\frac{\partial {\tilde \Psi} }{\partial y} = B_2 {\tilde \Psi}$, $\displaystyle \frac{\partial {\tilde \Psi} }{\partial t} = B_3 {\tilde \Psi}$, where
$B_2 \equiv (L^2)_+ = \partial_x^2 + u$, $B_3 = (L^3)_+ = \partial_x^3 +\frac{3}{4} (u\partial_x +\partial_x u) + u_3$ satisfy the compatibility
conditions $[ -\partial_y + B_2, -\partial_t +B_3] =0$.
If the divisor $\mathcal D$ is non--special, then $\tilde \Psi$ is uniquely identified by its normalization for $P\to P_0$. Finally,
$\partial_x u_3 =\frac{3}{4} \partial_y u$, and the KP regular finite--gap solution is
\[
u(x,y,t) = 2 \partial_x \chi_1 (x,y,t,0,\dots) = 2 \partial_x^2 \log( \Theta ( Ux+Vy+Zt+z_0)) + c,
\]
where $c\in \mathbb C$, $\Theta (z)$, $z\in \mathbb C^g$,  is the Riemann theta--function associated to $\Gamma$, $z_0\in {\mathbb C}^g$ is a constant vector
which depends on the divisor $\mathcal D$, and $U$, $V$, $Z\in {\mathbb C}^g$ are the periods of certain normalized meromorphic differentials on $\Gamma$.

According to \cite{DN}, a regular finite--gap KP--solution $u$ is real (quasi)--periodic if and only if it corresponds to
Krichever data on a regular $\mathtt M$--curve $\Gamma$. More precisely $\Gamma$ must possess an anti--holomorphic involution which fixes the maximum number
of ovals, $\Omega_0,\dots, \Omega_g$ such that $P_0\in \Omega_0$ and there is exactly one divisor point in each other oval, $P_j \in \Omega_j$, $j\in [g]$. We recall that the ovals are topologically circles and, by a theorem of Harnack  \cite{Har}, the maximal number of
components (ovals) of a real algebraic curve in the projective plane is equal to $(d-1)(d-2)/2+1$, where $d$ denotes the degree of the
curve. 
According to finite--gap theory \cite{D,DKN}, soliton solutions are obtained from finite--gap regular solutions in the limit in which some of the cycles
of $\Gamma$ become singular. In particular,
the real smooth bounded $(n-k,k)$-line solitons may be obtained from regular real quasi--periodic solutions in the rational limit of $\mathtt M$--curves
where some cycles shrink to double points. 

\begin{example}
The hyperelliptic involution $\sigma$ fixes the $n$ real ovals of the genus $(n-1)$ real hyperelliptic curve $\Gamma^{(\epsilon)} =\{ (\zeta, \eta) \; : \;
\eta^2 =  -\epsilon^2 +\prod_{j=1}^n (\zeta - \kappa_j)^2 \}$ and the divisor $(\gamma^{(\epsilon)}_1,\dots,\gamma^{(\epsilon)}_{n-1})$ such that $\zeta (\gamma^{(\epsilon)}_j) \in [\kappa_j +\epsilon, \kappa_{j+1}-\epsilon]$, for any $j\in [n-1]$ satisfies the reality condition in \cite{DN} for any $\epsilon^2>0$ sufficiently small.
\end{example}

\section{Vacuum KP--wavefunction on $\Gamma$}\label{sec:vacuumKP}

In this section, we define a family of vacuum wavefunctions $\Psi(P,\vec t)$ on $\Gamma=\Gamma_+\sqcup\Gamma_-$ as in (\ref{eq:hyprat}), which coincide with Sato vacuum wavefunction on $\Gamma_+$ and are parametrized by non special divisors ${\mathcal B} = \{ b_1 < \cdots < b_{n-1} \}\subset \Gamma_-$, such that $b_r \in ]\kappa_r, \kappa_{r+1}[$, $r\in [n-1]$. It is straightforward to verify that there is a bijection between such non special divisors and points $[a] \in Gr^{\mbox{\tiny TP}} (1,n)$.

We model such vacuum wavefunctions on those constructed in \cite{AG} for KP soliton data in $Gr^{\mbox{\tiny TP}} (1,n)$: in this case the effect of the Darboux transformation is to move the vacuum divisor points inside the finite ovals in such a way that exactly one pole belongs to $\Gamma_+$.

In the next sections, we investigate which Darboux transformations move the vacuum divisor points inside the finite ovals in such a way that exactly $k$ poles belong to $\Gamma_+$, for $k\in [n-1]$.

In the following, the algebraic setting is $(\Gamma, P_+, \zeta)$, where ${\Gamma}$ is the rational degeneration of a real hyperelliptic curve of genus $g=n-1$, with affine part 
\begin{equation}\label{eq:hyprat}
{\Gamma} \; : \; \; \{ (\zeta; \eta) \in \mathbb C^2 \; : \; \; \eta^2 = \prod\limits_{j=1}^n (\zeta-\kappa_j)^2\}.
\end{equation}
$\Gamma=\Gamma_+ \sqcup \Gamma_-$, with
$\Gamma_+ = \{ (\zeta; \eta(\zeta) )\, ; \, \zeta \in \mathbb C \}$, $\Gamma_- = \sigma (\Gamma_+)$, where $\sigma$ is the hyperelliptic involution, {\em i.e.} $\sigma (\zeta; \eta(\zeta) ) = (\zeta; -\eta(\zeta))$. The marked point is $P_+ \in \Gamma_+$ such that $\zeta^{-1} (P_+)=0$ and $P_- =\sigma (P_+) \in \Gamma_-$.
To simplify the notations, $\zeta$ denotes the local coordinate in both copies $\Gamma_{\pm}$. 

$\Gamma$ possesses $n$ ovals $\Omega_j$, $j\in [0,n-1]$. $\Omega_0$ is the oval containing the 
points $P_{\pm}$ and we call it infinite oval. We enumerate the other (finite) ovals according to the double points belonging to them, {\it i.e.} $\Omega_j$ is the oval containing $\kappa_j$ and
$\kappa_{j+1}$, $j\in [n-1]$ (see Figure \ref{fig:vacuum} for the example $n=4$).

\begin{definition}\label{def:psi1}
Let $\Gamma =\Gamma_+\sqcup \Gamma_-$ as in (\ref{eq:hyprat}). On $\Gamma_-$ we take $n-1$ real ordered points ${\mathcal B} = \{ b_1, \dots, b_{n-1} \}$ such that $b_r \in ]\kappa_r, \kappa_{r+1}[$, for any $r\in [n-1]$.
To such data we associate a vacuum wave function 
\begin{equation}\label{eq:psim}
\Psi(P,\vec t) = \left\{ \begin{array}{ll}
\Psi^{(+)}(\zeta;\vec t)=e^{\theta (\zeta;\vec t)}, &\quad \zeta \in \Gamma_+,\\
\displaystyle \Psi^{(-)} (\zeta; \vec t) = \frac{ \sum\limits_{l=1}^{n}a_l E_l (\vec t) \prod_{j\not = l}^n
(\zeta -\kappa_j) }{\prod_{r=1}^{n-1} (\zeta-b_r)}, &\quad \zeta \in \Gamma_-
\end{array}
\right.
\end{equation}
where 
\[
\theta \equiv \theta (\zeta; \vec t) = \sum\limits_{i\ge 1} \zeta^i t_i,
\quad\quad\quad\quad
a_l = \frac{\prod_{r=1}^{n-1} (\kappa_l - b_r)}{ \prod_{s\not =l}^n (\kappa_l -\kappa_s)}, \quad\quad l\in [n].
\]
\end{definition}

In Figure \ref{fig:vacuum}, we show the real part of the curve ${\Gamma}$ in the case $n=4$. The divisor points $b_j$ are in the intersection of the finite ovals $\Omega_j$ with $\Gamma_-$, for any $j\in [3]$. 

The wavefunction as in Definition \ref{def:psi1} has the following properties

\begin{lemma}\label{prop:0}
Let $\Gamma$, ${\mathcal B}$ and $\Psi(\zeta;\vec t)$ as in Definition \ref{def:psi1}. Then $\Psi(\zeta,\vec t)$
is a regular function of the variables $\vec t=\{t_1=x,t_2=y,t_3=t,t_4,\ldots\}$ and, as a function of $\zeta$, it is defined on 
$\Gamma\backslash \{P_+\}$. Moreover
\begin{enumerate}
\item\label{def:psi1:prop0} {$\Psi^{(\pm)}(\zeta; \vec t)$ is real for real $\zeta$ and real $\vec t$};
\item $\Psi(P, \vec 0)\equiv 1 $, for all $P\in \Gamma$;
\item\label{def:psi1:prop1} $\Psi^{(+)}(\zeta;\vec t)$ is the Sato vacuum KP wave function;
\item $\Psi$ has an essential singularity at the marked infinity point $P_+\in \Omega_0$;
\item The coefficients $a_j$, $j\in [n]$, are positive and $\sum_{j=1}^n a_j=1$;
\item\label{def:psi1:prop2} \textbf{Divisor of poles of $\Psi(P,\vec t)$:}
for any $\vec t$, $\Psi^{(-)}(\zeta;\vec t)$ is meromorphic in $\zeta$ on $\Gamma_-$  with simple poles at the points $b_r$, $r\in[n-1]$, whose position is independent of
$\vec t$;
\item \textbf{Divisor of zeros of $\Psi(P,\vec t)$:}
In each finite oval $\Omega_{r}$, ($r\in [n-1]$), $\Psi(\zeta, \vec t)$ possesses exactly  one simple pole $b_r$ and exactly one simple zero $b_r (\vec t)$, such that
$b_r (\vec 0) =b_r$ and $b_r (\vec t) \in ]\kappa_{r}, \kappa_{r+1} [ \subset\Gamma_-\cap \Omega_{r}$,  for all $\vec t$;
\item\label{def:psi1:prop3} \textbf{Gluing rules between $\Gamma_{\pm}$:} 
For any $j\in [n]$, the values of $\Psi^{(\pm)}$ coincide at the marked points  $\kappa_{j}$ for all $\vec t$:
\begin{equation}\label{eq:PP1}
\Psi^{(+)} ( \kappa_j , \vec t) = \Psi^{(-)} (\kappa_j, \vec t) \equiv E_j (\vec t), {\quad \forall \vec t},
\end{equation}
with $E_j(\vec t)$ as in (\ref{eq:f}).
\end{enumerate}
\end{lemma}

\begin{figure}
\includegraphics[scale=0.3]{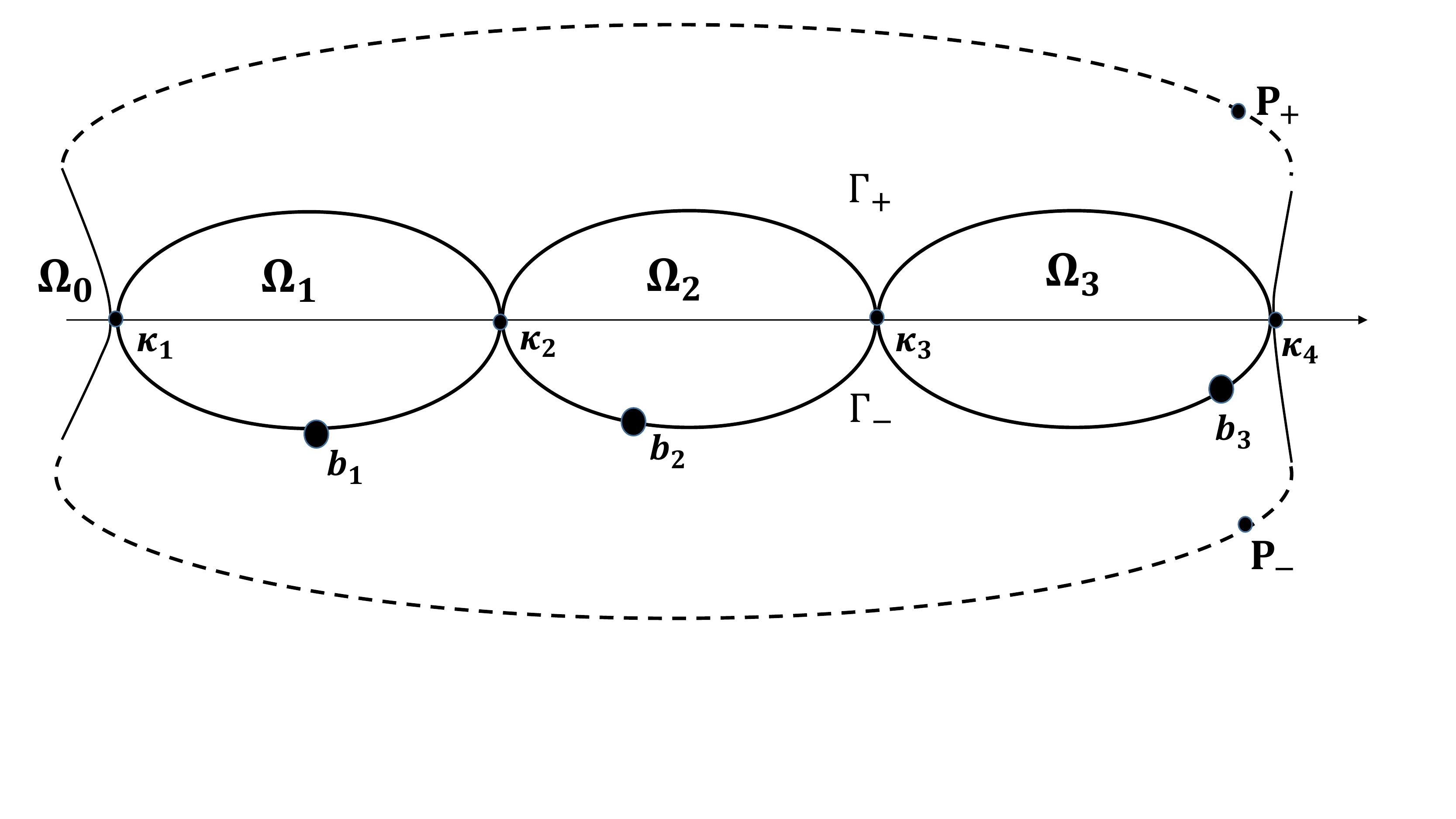}
\vspace{-1.5 truecm}
\caption{{\small{\sl A vacuum divisor $\{b_1< b_2 <b_3 \}$ on $\Gamma$ of arithmetic genus 3. The essential singularity of the Sato vacuum wavefunction is at $P_+$. The hyperelliptic involution is reflection w.r.t. the horizontal axis and it leaves invariant the ovals $\Omega_0, \dots,\Omega_3$. The values of $\Psi^{(+)}(\zeta, \vec t)$, $\Psi^{(-)}(\zeta, \vec t)$ coincide at the double points $\kappa_j$, $j\in [4]$, for all $\vec t$.}}}
\label{fig:vacuum}
\end{figure}

We remark that the condition that each zero of $\Psi(\zeta,\vec t)$ lies in a well-defined  open interval $]\kappa_{j},\kappa_{j+1}[$ for all $\vec t$, is natural since
$\Psi(\zeta,\vec t)$ represents a vacuum wave function: no collision is possible for the zero divisor.

To each vacuum divisor ${\mathcal B}$ as in  Definition \ref{def:psi1}, we associate the point in $\left[ a\right] \equiv \left[a_1, \dots, a_n \right] \in Gr^{\mbox{\tiny TP}} (1,n)$. Viceversa, to each point $[a] \in Gr^{\mbox{\tiny TP}} (1,n)$, we uniquely associate a divisor ${\mathcal B}$ and vacuum wave--function $\Psi(\zeta; \vec t)$, where ${\mathcal B} = \{ \zeta \, :  Q(\zeta)=0 \} $, with
\begin{equation}\label{eq:Q}
Q(\zeta)= \prod\limits_{j=1}^n (\zeta - \kappa_j) \left( \sum\limits_{l=1}^n \frac{a_l}{\zeta-\kappa_l}\right), 
\end{equation}
Indeed we have the following

\begin{lemma}\label{lemma:inv}
Let ${\mathcal K} = \{\kappa_1< \kappa_2 < \cdots < \kappa_n\}$ and $[a] \in Gr^{\mbox{\tiny TP}} (1,n)$ be fixed, with normalization $\sum\limits_{j=1}^n  a_j=1$. Let $b_r$, $r\in [n-1]$, be the solutions to $Q(\zeta) =0$ as in (\ref{eq:Q}). Then
$b_r\in ]\kappa_r, \kappa_{r+1}[$, $\forall r\in [n-1]$ and
\begin{equation}\label{eq:1M}
\Psi (\zeta; \vec t ) = \left\{ \begin{array}{ll}
\displaystyle \Psi^{(+)} (\zeta; \vec t) \equiv e^{\theta (\zeta;\vec t)}, & \quad \mbox{ if } \zeta\in \Gamma_+,\\
\displaystyle \Psi^{(-)} (\zeta; \vec t) \equiv
\sum\limits_{l=1}^{n}  a_j E_j( \vec t)\,\;\frac{\prod_{s\not =l} (\zeta -\kappa_s)}{\prod_{r=1}^{n-1} (\zeta-b_r)} , & \quad \mbox{ if } \zeta\in\Gamma_-,
\end{array}
\right.
\end{equation}
satisfies all the properties in Lemma \ref{prop:0}.
\end{lemma}

\begin{remark}\label{rem:Psi}
Lemmata \ref{prop:0} and \ref{lemma:inv} imply that the condition $b_r \in ]\kappa_r, \kappa_{r+1}[$, $r\in [n-1]$, is equivalent to $a_j >0$, for all $j\in [n]$.
It follows that assigning the algebraic geometric data $(\Gamma, P_+, {\mathcal B})$ as in Definition \ref{def:psi1} is equivalent to assigning the vacuum soliton data ${\mathcal K}, [a]$, with ${\mathcal K} =\{ \kappa_1< \cdots \kappa_n\}$, $[a]\in Gr^{\mbox{\tiny TP}}(1,n)$. 
\end{remark}

\begin{remark}
In \cite{AG}, we have introduced a parameter $\xi$ which governs the position of the marked points at which we glue different copies of ${\mathbb CP}^1$ to construct a rational degeneration of an $\mathtt M$--curve of arithmetic genus $k(n-k)$. At such marked points we control the asymptotics of the vacuum wave--function. 
In the case $k=1$, in \cite{AG} we glue exactly two copies $\Gamma_0$ and $\Gamma_1$ of ${\mathbb CP}^1$ and the introduction of the scaling parameter $\xi$ is unnecessary since any set of $n$ ordered points on $\Gamma_1$ would work. In particular, for the choice $\lambda^{(1)}_j=\kappa_j$, $j\in [n]$, the vacuum wavefunction constructed in \cite{AG} coincides with
 (\ref{eq:psim}) for soliton data $({\mathcal K}, [a])$, $[a] \in Gr^{\mbox{\tiny TP}} (1,n)$.
\end{remark}

\section{$T$-hyperelliptic solitons and $k$--compatible divisors}\label{sec:Tsol}

\begin{figure}
\includegraphics[scale=0.3]{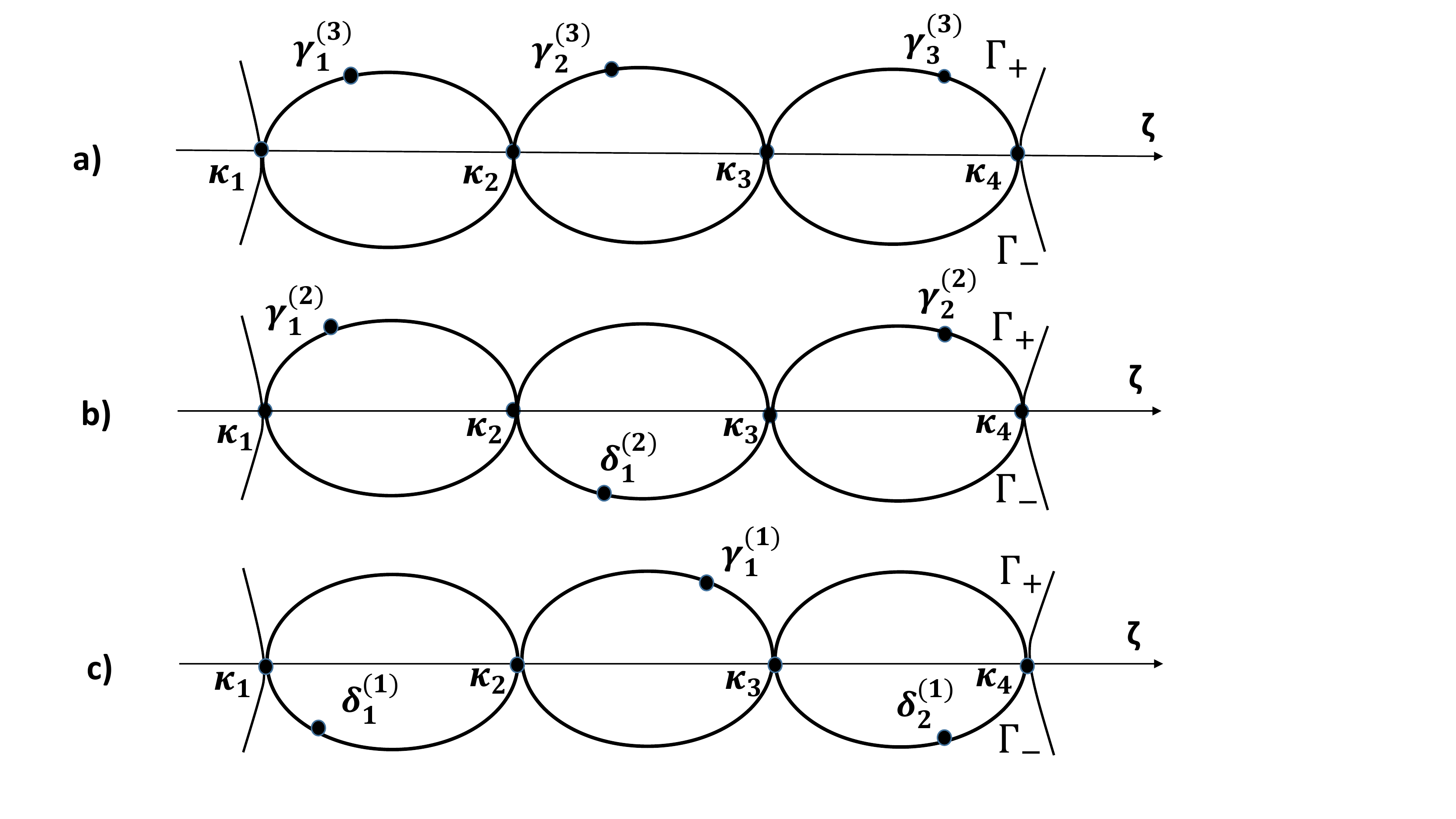}
\vspace{-0.4 truecm}
\caption{{\small{\sl According to Corollary \ref{cor:mat}, $k$--compatible divisors are realized in $(n-1)$-- dimensional varieties of KP soliton data in $Gr^{\mbox{\tiny TP}}(k,n)$. Here $\Gamma$ has arithmetic genus 3, vacuum divisor as in Figure \ref{fig:vacuum} and $k=3$ (Figure \ref{fig:gen3}.a), $k=2$ (Figure \ref{fig:gen3}.b)  and $k=1$(Figure \ref{fig:gen3}.c). The case $k=0$ (3 divisor points on $\Gamma_-$) corresponds to the vacuum KP wavefunction in Figure \ref{fig:vacuum} and the trivial KP solution $u(\vec t)\equiv 0$.}}}
\label{fig:gen3}
\end{figure}

Any $(n-k,k)$--line soliton regular bounded KP-solution is uniquely identified by the soliton data $({\mathcal K}, [A])$, where ${\mathcal K} = \{\kappa_1 < \dots < \kappa_n\}$ and $[A] \in Gr^{\mbox{\tiny TNN}} (k,n)$ define the heat hierarchy solutions $f^{(1)} (\vec t),\dots f^{(k)} (	\vec t)$, as in (\ref{eq:f}) and the Darboux transformation 
$D^{(k)}\equiv W\partial_x^k=\partial_x^{k}-w_1(\vec t) \partial_x^{k-1}- \ldots - w_k(\vec t),$
as in (\ref{eq:Darboux}) where $W$ is the Dressing operator of the vacuum and $D^{(k)}f^{(i)}\equiv 0$, for any $i\in [k]$.

Let $({\Gamma}, P_+, \zeta)$, with ${\Gamma}$ as in (\ref{eq:hyprat}). We make the ansatz that such real boundend soliton solution may be obtained from a real regular finite gap KP--solution on $\{ (\zeta;\eta) \; : \; \eta^2 =-\epsilon^2+\prod\limits_{j=}^n (\zeta - \kappa_j)^2 \}$, in the limit $\epsilon\to 0$. Such ansatz 
means that the Darboux transformation $D^{(k)}$, associated to the soliton data $({\mathcal K}, [A])$, acts on the vacuum divisor moving the poles inside the finite ovals in such a way that exactly  $k$ of them belong to $\Gamma_+$ and $n-k-1$ to $\Gamma_-$ and that $D^{(k)}$ creates exactly $k$ fixed zeros at $P_-$. Then the normalized KP--wavefunction $\frac{D^{(k)}\Psi(\zeta, \vec t)}{D^{(k)}\Psi(\zeta, \vec 0)}$ has a Krichever divisor ${\mathcal D}$ which satisfies the reality condition \cite{DN} and is realized from the vacuum dressing via $W =D^{(k)}\partial_x^{-k}$, that is:
\begin{enumerate}
\item the Krichever divisor is ${\mathcal D} = \{ \gamma_1, \dots ,\gamma_{k}, \delta_{1},\dots, \delta_{n-k-1} \} \subset {\Gamma}\backslash \{P_+\}$;
\item there is exactly one divisor point in each finite oval $\Omega_{j}$, $j\in [n-1]$ according to the counting rule below; 
\item ${\mathcal D} \cap \Gamma_+ = \{ \gamma_1, \dots ,\gamma_{k} \}$, ${\mathcal D} \cap \Gamma_- = \{ \delta_{1},\dots, \delta_{n-k-1} \}$.
\end{enumerate}
We call such a divisor $k$--compatible and $T$--hyperelliptic the corresponding KP soliton data $({\mathcal K}, [A])$.
In this section we prove that the soliton data $({\mathcal K}, [A])$ are $T$--hyperelliptic if and only if there is $[a]\in Gr^{\mbox{\tiny TP}}(1,n)$ and
\[
A^i_j = a_j \kappa_j^{i-1}, \quad\quad i\in [k],\; j\in [n].
\]
In such a case the vacuum divisor $\mathcal B$ is the one associated to $[a]\in Gr^{\mbox{\tiny TP}}(1,n)$ in the previous section.
In particular for $\mathcal K$ and $k\in [2,n-2]$ fixed, such soliton data parametrize an $(n-1)$--dimensional variety in $Gr^{\mbox{\tiny TP}}(k,n)$. 

\begin{remark}\label{def:counting} \textbf{(The counting rule)} We call $\mathcal{ D}$ generic, if no points of $\mathcal{ D}$ lie at the double points of $\Gamma$, otherwise we call it non generic.
In the non generic case, both $X=\kappa_m$ and $\sigma(X) =\kappa_m$ belong to $\mathcal{D}$, for some $m\in [n]$, and the wavefunction  has simple zeroes (resp. simple poles) at $\kappa_m$ at both the components $\Gamma_{-}$ and $\Gamma_{+}$, i.e. we have a collision of 2 divisor points 
$\gamma_{s}\in\Gamma_{+}$ and $\delta_{l}\in\Gamma_{-}$. 
Then we use the following counting rule: if we have a pair of divisor points at a double point, one of them is assigned to the left oval and the other is assigned 
to the right oval.
\end{remark}

To characterize on $({\Gamma}, P_+, \zeta)$ the admissible soliton data $({\mathcal K}, [A])$, where $[A]\in Gr^{\mbox{\tiny TP}} (k,n)$, we
introduce the following definitions of $k$--compatible divisor and of $T$-hyperelliptic soliton.

\begin{definition}\label{def:compdiv}\textbf{($k$--compatible divisor)}
Let $(\Gamma, P_+, \zeta)$ be as above.
Let $\Omega_j^{\pm} = \Gamma_{\pm} \cap [\kappa_{j},\kappa_{j+1}]$, $j\in [n-1]$ so that the finite ovals are $\Omega_{j} =\Omega_{j}^+ \cup \Omega_{j}^-$. We call a divisor $\mathcal D = \{ \gamma_1, \dots, \gamma_k, \delta_{1}, \dots, \delta_{n-k-1} \}\subset {\Gamma}\backslash \{P_+\}$ $k$--compatible if:
\begin{enumerate}
\item $\gamma_j \in \Gamma_+$, $j\in [k]$ are pairwise distinct and
$\delta_{l} \in \Gamma_-$, $l\in [n-k-1]$ are pairwise distinct points;
\item $\mathcal D \cap \Omega_j \not =\emptyset$, $j=1,\dots, n-1$, and each finite oval contains exacly one divisor point according to the counting rule;
\item $\mathcal D_{\pm} \cap \{\kappa_1, \kappa_n\} = \emptyset$, so that in particular no divisor point is in the infinite oval $\Omega_0$;
\item if $\kappa_m \in \mathcal D_{+}$ for some $m\in [2,n-1]$, then $\kappa_m \in \mathcal D_{-}$
and ${\mathcal D}_{\pm} \cap \left( [\kappa_{m-1}, \kappa_{m+1} ] \backslash \{\kappa_m\} \right) =\emptyset$.
\end{enumerate}
We define the $k$--compatible divisor generic if $\mathcal D_{\pm} \cap \{\kappa_1,\dots, \kappa_n\} = \emptyset$.
\end{definition}

In figure Figure \ref{fig:gen3}, we show $k$--compatible divisors for the case $n=4$ and $k\in [3]$.

\begin{definition}\label{def:hypsol}\textbf{($T$-hyperelliptic soliton)}
Let $([A], \mathcal K)$, $[A] \in Gr^{\mbox{\tiny TNN}} (k,n)$ be the soliton data of a regular bounded $(n-k, k)$--soliton solution to the KP equation and let
$D^{(k)} = \partial_x^k -w_1 (\vec t) \partial_x^{k-1} -\cdots - w_k (\vec t)$ be the associated Darboux transformation. Let $({\Gamma}, P_+, \zeta)$,
with $\Gamma$ as in (\ref{eq:hyprat}), $P_+\in \Gamma_+$ such that $\zeta^{-1} (P_+)=0$.

We call the soliton $([A], \mathcal K)$ $T$-hyperelliptic, if there exists a $0$--compatible vacuum divisor ${\mathcal B} = \{ b_1 < \cdots < b_{n-1} \} \subset \Gamma_-$
{\sl i.e.}  $\kappa_1 <b_1 < \kappa_2 < \cdots < b_{n-1} < \kappa_n$, and the corresponding vacuum wavefunction $\Psi (\zeta; \vec t)$ as in Definition \ref{def:psi1} has the following property: after the Darboux transformation
$D^{(k)}$  associated to the soliton data $([A], \mathcal K)$, the zero--divisor of $\Psi^{(k)} (\zeta; \vec t) \equiv\displaystyle D^{(k)}\Psi (\zeta; \vec t)$  
is ${\mathcal Z}^{(k)}_{{\mathcal B}} (\vec t) \equiv {\mathcal D}_{{\mathcal B}} (\vec t) \cup \{ k\, P_-\} \subset \Gamma \backslash \{ P_+\}$, with  ${\mathcal D}_{{\mathcal B}} (\vec t)$ $k$--compatible for any $\vec t$.
\end{definition}

\begin{remark}
We remark that the $T$-hyperelliptic solitons are a class of KP--soliton solutions whose algebraic geometric data are associated to rational degenerations of hyperelliptic curves, but they do not exhaust the whole class of KP--soliton solutions associated to algebraic geometric data on rational degenerations of real hyperelliptic curves. 
\end{remark}

The vacuum wave--function defined in the previous section on $\Gamma$ possesses a compatible 0--divisor according to
the above definition and the corresponding $T$-hyperelliptic soliton is the trivial solution $u(\vec t )\equiv 0$.

Let $\mathcal K = \{ \kappa_1 <\cdots <\kappa_n\}$, $\Gamma$ as in (\ref{eq:hyprat}), ${\mathcal B} = \{ b_1 < \cdots < b_{n-1} \} \subset \Gamma_-$, and the vacuum wave--function as in (\ref{eq:psim})
\begin{equation}\label{eq:psiB}
\Psi_{{\mathcal B}} (\zeta; \vec t ) = \left\{ \begin{array}{ll}
\displaystyle e^{\theta (\zeta;\vec t)}, & \quad \mbox{ if } \zeta\in \Gamma_+,\\
\displaystyle \Psi^{(-)}_{{\mathcal B}} (\zeta; \vec t) \equiv
\sum\limits_{l=1}^{n} a_j({\mathcal B}) E_j( \vec t)\; \frac{\prod_{s\not =l} (\zeta -\kappa_s)}{\prod_{r=1}^{n-1} (\zeta-b_r)}, & \quad \mbox{ if }
\zeta\in\Gamma_-,\\
\end{array}
\right.
\end{equation}
where $ a_j ({\mathcal B}) = \operatorname*{Res}_{\zeta =\kappa_j}
\frac{\prod_{r=1}^{n-1} (\zeta - b_r)}{ \prod_{s=1}^n
(\zeta -\kappa_s)} >0$, $j\in [n]$, $\sum\limits_{j=1}^n a_j({\mathcal B}) =1$.

Let $[A] \in Gr^{\mbox{\tiny TNN}} (k,n)$ and choose a basis of heat hierarchy solutions $f^{(i)} (\vec t) = \sum\limits_{j=1}^n A^i_j E_j (\vec t)$, $i\in [k]$, representing $({\mathcal K},[A])$.
The Darboux transformation
$D^{(k)} = \partial_x^k -w_1 (\vec t) \partial_x^{k-1} - \cdots - w_k(\vec t)$ 
is obtained solving the linear system  $D^{(k)} f^{(i)} (\vec t)\equiv 0$, $i\in [k]$.
The soliton data $({\mathcal K}, [A])$ are $T$--hyperelliptic if and only if the Darboux transformed wave--function takes the
form
\begin{equation}\label{eq:psiN}
\Psi^{(k)}_{{\mathcal B}} (\zeta; \vec t) \equiv\displaystyle D^{(k)}\Psi_{{\mathcal B}}(\zeta; \vec t)=\left\{ \begin{array}{ll} \displaystyle \left(\zeta^k  -w_1 (\vec t) \zeta^{k-1} -\cdots - w_k (\vec t)\right) e^{\theta (\zeta;\vec t)},
 &\quad \zeta\in \Gamma_+,\\
\displaystyle {\tilde A}^{(k)} (\vec t) \frac{\prod_{l=1}^{n-k-1} (\zeta-\delta_l^{(k)} (\vec t))}{\prod_{r=1}^{n-1} (\zeta-b_r)}, &\quad \zeta\in
\Gamma_-,
\end{array}
\right.
\end{equation}
with ${\tilde A}^{(k)} (\vec t)\not =0$ for almost all $\vec t$ (we anticipate that, from our construction, necessarily ${\tilde A}^{(k)} (\vec t)>0$ for all $\vec t$). 

For any fixed 0--compatible (vacuum) divisor ${\mathcal B}$, the restriction of $D^{(k)}\Psi_{{\mathcal B}} (\zeta; \vec t)$ to $\Gamma_+$ satisfies condition (1) in Definition \ref{def:compdiv}, since  ${\mathcal
Z}^{(k)}_{{\mathcal B}} ( \vec t) \cap \Gamma_+ = \{ \gamma^{(k)}_1 (\vec t) , \dots ,
\gamma^{(k)}_k (\vec t) \}$, with
\begin{equation}\label{eq:gammak}
\zeta^k - w_1(\vec t) \zeta^{k-1} -\cdots -w_k(\vec t) =
\prod\limits_{n=1}^k (\zeta -\gamma^{(k)} (\vec t)).
\end{equation}
Moreover, by a Theorem in Malanyuk  \cite{Mal}, for any real $\vec t$, 
$
\kappa_1 \le\gamma^{(k)}_1 (\vec t) < \cdots < \gamma^{(k)}_k (\vec t)
\le \kappa_n$.
On $\Gamma_-$ for generic choice of $\mathcal B$,  
\begin{equation}\label{eq:Dknonhyp}
D^{(k)} \Psi^{(-)}_{{\mathcal B}} (\zeta; \vec t ) = \sum\limits_{j=1}^n
{a}_j ({\mathcal B}) \, E_j (\vec t) \, \prod\limits_{l=1}^k (\kappa_j
-\gamma^{(k)}_l (\vec t)) \, \frac{\prod_{s\not =
j}^n (\zeta - \kappa_s)}{\prod_{r =1}^{n-1} (\zeta - b_r)} ={\tilde A}^{(k)} (\vec t) \frac{\prod_{l=1}^{n-1} (\zeta-\delta_l^{(k)} (\vec t))}{\prod_{r=1}^{n-1} (\zeta-b_r)},
\end{equation}
has non--special complex zero divisor ${\mathcal Z}^{(k)}_{{\mathcal B},-} ( \vec t) \equiv {\mathcal Z}^{(k)}_{{\mathcal B}}
(\vec t) \cap \Gamma_-=
\left\{ \delta^{(k)}_1 (\vec t) , \dots , \delta^{(k)}_{n-1} (\vec t)
\right\}$, such that, for all $\vec t$, $\# \left({\mathcal Z}^{(k)}_{{\mathcal B},-}(\vec t)\cap [k_1, k_n] \right) \ge n-k-1$.

\begin{remark}
(\ref{eq:gammak}) and (\ref{eq:Dknonhyp}) imply that generic soliton data $({\mathcal K}, [A])$, with $[A]\in Gr^{\mbox{\tiny TNN}} (k,n)$, may be realized assigning on $\Gamma$, a non special divisor ${\mathcal D} \equiv {\mathcal D} (\vec 0)=\{ \gamma^{(k)}_1,\dots, \gamma^{(k)}_k, \delta^{(k)}_1,\dots, \delta^{(k)}_{n-1}\} $ such that:
\begin{enumerate}
\item
for all $j\in [k]$, $\gamma^{(k)}_j \in \Gamma_+\cap [\kappa_1,\kappa_n]$ is a root of (\ref{eq:gammak}) for $\vec t=\vec 0$; 
\item $\delta^{(k)}_s\in \Gamma_-$, for all $s\in [n-1]$ are such that $\# \left(\{ \delta^{(k)}_s\} \cap [k_1, k_n] \right) \ge n-k-1$;
\item ${\mathcal D}$ satisfies the counting rule;
\item in each finite oval there is an odd number of divisor points.
\end{enumerate}
\end{remark}

In the following Lemma we establish the necessary and sufficient
conditions so that $\{ k \, P_- \} \subset {\mathcal Z}^{(k)}_{{\mathcal B},-}
( \vec t)$.

\begin{lemma}\label{lemma:pos}
Let $({\mathcal K},[A])$, with $[A] \in Gr^{\mbox{\tiny TNN}} (k,n)$
and $\mathcal K = \{ \kappa_1 < \cdots < \kappa_n\}$ be given. Let $D^{(k)}$ be the Darboux transformation for $({\mathcal K},[A])$. Let ${\mathcal B} = \{
b_1 < \cdots < b_{n-1} \}$ be a 0--compatible divisor on
$\Gamma$ as in (\ref{eq:hyprat}) and $\Psi_{{\mathcal B}}
(\zeta; \vec t)$ as in (\ref{eq:psiB}).
Let $s\in [k]$ be fixed. Then the following assertions are
equivalent
\begin{enumerate}[label={(\roman*)}]
\item $\{s \, P_- \} \subset {\mathcal Z}^{(k)}_{{\mathcal B},-} (\vec t)$ and $\left( {\mathcal Z}^{(k)}_{{\mathcal B},-} ( \vec t)
\backslash \{ s \, P_- \} \right)\, \subset \, \Gamma_- \backslash \{P_-\}$;
\item For all $\vec t$, $\displaystyle
\sum\limits_{j=1}^n \kappa_j^{i-1} a_j ({\mathcal B})
\prod\limits_{l=1}^k (\kappa_j -\gamma^{(k)}_l (\vec t)) E_j (\vec t)
\equiv 0$, $\forall  i\in [s]$ and $\displaystyle
\sum\limits_{j=1}^n \kappa_j^s a_j ({\mathcal B})
\prod\limits_{l=1}^k (\kappa_j -\gamma^{(k)}_l (\vec t)) E_j (\vec t) \not
= 0$;
\item The heat hierarchy solutions 
$\mu_{i} (\vec t) =\sum\limits_{j=1}^n k^{i} a_j
({\mathcal B})E_j (\vec t)$, $i\ge 0$,
satisfy $D^{(k)} \mu_i (\vec t)=0$, for all $i\in [0,s-1]$, and  $D^{(k)} \mu_s (\vec t) \not =0$.
\end{enumerate}
\end{lemma}

The proof is trivial and it is omitted. 

\begin{remark}
Let $k\in[n-1]$ be fixed and let $({\mathcal K}, [A])$, $[A]\in Gr^{\mbox{\tiny TNN}} (k,n)$ be the soliton data. For $s=1$, the condition $D^{(k)} \mu_0 (\vec t) =0$ in Lemma \ref{lemma:pos} is realized taking $\displaystyle \mu_0 (\vec t) =\sum\limits_{j=1}^n A^i_j E_j (\vec t)$, for some fixed $i\in[k]$ and choosing ${\mathcal B}$ as in Lemma \ref{lemma:inv}, for $a_j\equiv	\frac{ A^i_j}{\sum_{l=1}^n A^i_l}$, for all $j\in [n]$.

For $1<s< k$, generically, there does not exist a 0--divisor $\mathcal B\subset \Gamma_-$ such that the heat hierarchy solutions $\mu_0 (\vec t),\dots, \mu_{s-1} (\vec t)$ satisfy Lemma \ref{lemma:pos}. 

For $s=k$, if such a divisor $\mathcal B\subset \Gamma_-$ exists, it is unique. Moreover in such case $\mu_0 (\vec t), \dots, \mu_{k-1} (\vec t)$ generate the Darboux transformation $D^{(k)}$.
\end{remark}

In the following, we restrict ourselves to the case $s=k$.

\begin{corollary}\label{cor:mat}
Let the soliton data $(\mathcal K, [A])$ be given, with $[A]\in Gr^{\mbox{\tiny TNN}} (k,n)$, and suppose that there exists a vacuum divisor ${\mathcal B}$ such that Lemma \ref{lemma:pos} holds for $s=k$. Then 
$[A]\in Gr^{\mbox{\tiny TP}} (k,n)$ and $[A] = [B]$, where
\begin{equation}\label{eq:Bmat}
B^i_j = \kappa_j^{i-1} a_j ({\mathcal B}), \quad\quad j\in [n], \quad i \in [k]. 
\end{equation}
\end{corollary}

Let ${\mathcal K} = \{\kappa_1,\dots\kappa_4\}$. Then  
the divisors in Figure \ref{fig:gen3} correspond to KP--soliton data $({\mathcal K}, [A])$, respectively with $[A] \in Gr^{\mbox{\tiny TP}}(3,4)$ (Figure \ref{fig:gen3}.a)), $[A] \in Gr^{\mbox{\tiny TP}}(2,4)$ (Figure \ref{fig:gen3}.b)) and $[A] \in Gr^{\mbox{\tiny TP}}(1,4)$ (Figure \ref{fig:gen3}.c)) where $[A]$ satisfies Corollary \ref{cor:mat}.

\begin{corollary}\label{cor:grn-1}
Let $\mathcal K$ be given. Then for any $[A]\in
Gr^{\mbox{\tiny TP}} (n-1,n)$, there exists $[a] \in Gr^{\mbox{\tiny TP}} (1,n)$ such that $[A] =[B]$, with $B$ as in (\ref{eq:Bmat}).
\end{corollary}

\begin{proof}
Indeed, let $x_i>0$, $i=1\dots, n-1$ and 
{\small\begin{equation}\label{eq:An-1}
A = \left( \begin{array}{ccccc}
1      & 0      & \cdots & 0 & (-1)^{n-2} x_{1}\\
\vdots & \ddots & \ddots & \vdots & \vdots          \\
0      & \cdots & 1      & 0 & - x_{n-2}       \\
0      & \cdots & 0      & 1 & x_{n-1} 
\end{array}
\right)
\end{equation}}
be the representative matrix in the reduced row echelon form of the given point $[A] \in Gr^{\mbox{\tiny TP}} (n-1,n)$ . For $[a] \in Gr^{\mbox{\tiny TP}} (1,n)$, let $B$ the matrix as in (\ref{eq:Bmat}) and denote
\[
y_i \equiv \Delta_{[1,\dots,{\hat i}, \dots, n]} (B) = \left(\prod\limits_{s\not = i} a_s \right) \mathop{\prod\limits_{1\le j < l \le n}}_{j,l \not =i} (\kappa_l-\kappa_j), \quad\quad i\in [n],
\]
its $(n-1)$ minors with the $i$--th column omitted.
Then $[A] = [B]$ if and only if
\[
x_i = \frac{y_i}{y_n} = (-1)^{n-i-1}\frac{a_n}{ a_i} \mathop{\prod\limits_{s=1}}_{s \not =i}^{n-1}\frac{(\kappa_n-\kappa_s)}{(\kappa_i-\kappa_s)}, \quad\quad i\in [n-1].
\]
\end{proof}

For fixed $k$, $1<k< n$, (\ref{eq:Bmat}) identify a $(n-1)$--dimensional variety in $Gr^{\mbox{\tiny TP}} (k,n)$. Indeed, let the matrix in reduced row echelon form 
{\small \begin{equation}\label{eq:RREFmat}
A = \left( \begin{array}{ccccccc}
1      & 0      & \cdots & 0      & (-1)^{k-1} x^k_1 & \cdots &  (-1)^{k-1} x^k_{n-k}\\
\vdots & \ddots & \ddots & \vdots & \vdots           & \vdots & \vdots \\
0      & \cdots & 0      & 1      & x^{1}_1          & \cdots &  x^{1}_{n-k} 
\end{array}
\right),
\end{equation}}
represent a point $[A]\in Gr^{\mbox{\tiny TP}} (k,n)$ for which Lemma \ref{lemma:pos} holds for $n=k$. The $k\times (n-k)$ matrix
{\small\[
X = \left( \begin{array}{ccc}
x^1_1 &\cdots  & x^1_{n-k}\\
\vdots & \vdots & \vdots \\
x^{k}_1 &\cdots  &  x^{k}_{n-k} 
\end{array}
\right),
\]}
is totally positive in classical sense and the explicit relations between $[ a] =[ a_1,\dots, a_n]\in Gr^{\mbox{\tiny TP}} (1,n)$ and the coefficients of $X$ are as follows:

\begin{lemma}
Let $B$ the $k\times n$ matrix defined in (\ref{eq:Bmat}) and associated to $[a] =[a_1,\dots, a_n]\in Gr^{\mbox{\tiny TP}} (1,n)$, such that $\sum\limits_{j=1}^n a_n =1$. Then the coefficients of the associated reduced row echelon form matrix $A$ as in (\ref{eq:RREFmat}) take the form
\[
x^i_j\;  = \; (-1)^{k-i} \; \frac{a_j}{a_i} \; \mathop{\prod\limits_{l=1}^k}_{l\not = i}  \frac{(\kappa_j - \kappa_l) }{(\kappa_i - \kappa_l) }, \quad \quad i\in [k], \quad j \in [n-k+1,k].
\]
\end{lemma}

We end this section summarizing the divisor properties of $T$--hyperelliptic $(n-k,k)$--solitons.

\begin{theorem}\label{theo:divisor} Let $k\in [n-1]$, $\mathcal K = \{ \kappa_1 < \cdots < \kappa_n\}$, $[a] =[ a_1,\dots, a_n]\in Gr^{\mbox{\tiny TP}} (1,n)$ and $[A] \in Gr^{\mbox{\tiny TNN}} (k,n)$ satisfy Corollary \ref{cor:mat}. Let $D^{(k)}$ and ${\tilde \Psi}^{(k)} (\zeta; \vec t) = \frac{D^{(k)}\Psi(\zeta; \vec t)}{D^{(k)}\Psi(\zeta; \vec 0)}$ respectively be the Darboux transformation and the KP--wavefunction associated to the soliton data $({\mathcal K}, [A])$, with $\Psi(\zeta; \vec t)$ as in (\ref{eq:1M}).
Then, for all $\vec t$, the zero divisor $\mathcal{ D}^{(k)}(\vec t)$ of ${\tilde \Psi}^{(k)}$ has the following properties:
\begin{enumerate}
\item $\Gamma_+$ contains exactly $k$ points of 	$\mathcal{ D}^{(k)} (\vec t)$: $\mathcal{ D}_+^{(k)} (\vec t) = \{\gamma_1^{(k)}(\vec t), \dots, \gamma_k^{(k)}(\vec t)\}$;
\item $\Gamma_-$ contains exactly $n-k-1$ points of $\mathcal{ D}^{(k)}(\vec t)$: $\mathcal{ D}_-^{(k)} (\vec t) = \{\delta_1^{(k)}(\vec t), \dots, \delta_{n-k-1}^{(k)}(\vec t)\}$;
\item All points $ \gamma_l^{(k)}(\vec t)$, $l\in [k]$, lying in  $\Gamma_+$ are pairwise different;
\item All points $ \delta_s^{(k)}(\vec t)$, $s\in [n-k-1]$, lying in  $\Gamma_-$ are pairwise different;
\item $\mathcal{D} (	\vec t)\cap \Omega_0 =\emptyset$;
\item $\mathcal{D} \subset \bigcup\limits_{n} \Omega_{n}$, that is each divisor point is real and lies in some finite oval;
\item Each finite oval  $\Omega_{n}$ contains exactly one point of $\mathcal{ D} (\vec t)$ both for the generic and the non generic case, according to the counting rule.
\end{enumerate}
\end{theorem}

\begin{remark} For any fixed $\vec t$, no zero or pole of   $\tilde \Psi^{(k)}(\zeta,\vec t)$ lies at the double points $\kappa_1$ or $\kappa_n$, and, thanks to the counting rule, $\# \left( \mathcal D_+^{(k)} (\vec t) \cap \{ \kappa_2,\dots, \kappa_{n-1} \}\right) =
\# \left( \mathcal D_-^{(k)} (\vec t) \cap \{ \kappa_2,\dots, \kappa_{n-1} \}\right) \le {\mbox min } \{ k  , n-k-1 \}$.
\end{remark}

The proof of Theorem~\ref {theo:divisor} follows the same lines as for Theorem 7 in \cite{AG} and is omitted. Notice that the
pole divisor of ${\tilde \Psi}^{(k)}$ is just ${\mathcal D}^{(k)}(\vec 0)$.

\section{Comparison with the construction in \cite{AG}}\label{sec:desing}

In this section we show that $\Gamma$ as in (\ref{eq:hyprat})  is a desingularization of the curve $\Gamma_{\xi}$ constructed in \cite{AG} for points in $Gr^{\mbox{\tiny TP}} (n-1,n)$ and that the respective KP wavefunctions coincide.

In \cite{AG}, for any fixed $\xi>>1$ we have associated to any soliton data $({\mathcal K}, [A])$, with ${\mathcal K }=\{\kappa_1<\kappa_2 <\cdots <\kappa_n\}$ and $[A]\in Gr^{\mbox{\tiny TP}}(k,n)$,
\begin{enumerate}
\item a connected curve $\Gamma_{\xi}=\Gamma_0 \sqcup \Gamma_{\xi,1} \sqcup \cdots \sqcup \Gamma_{\xi,k}$, which is the rational degeneration of a regular $\mathtt M$--curve of genus $(n-k)k$ with $1+(n-k)k$ ovals,  $\Omega_0$, $\Omega_{i,j}$, $i\in [k]$, $j\in [n-k]$;
\item a vacuum wavefunction $\Psi_{\xi} (\lambda;\vec t)$ on $\Gamma(\xi)$ with the following properties:
\begin{enumerate}
\item $\Psi_{\xi} (\lambda;\vec t)$ is real for $\lambda\in \Omega_0 \mathop{\cup}_{i,j} \Omega_{i,j}$ and real $\vec t$;
\item on $\Gamma_0$, $\Psi_{\xi}(\lambda;\vec t)$ coincides with Sato vacuum wavefunction and has essential singularity at $P_0 \in \Gamma_0\cap \Omega_0$;
\item on each $\Gamma_{\xi,i}$, $i\in [k]$, $\Psi_{\xi}(\lambda;\vec t)$ is meromorphic and possesses $n-k$ divisor points $b_{\xi,1}^{(i)}, \dots, b_{\xi,n-k}^{(i)}$,
whose position is independent of $\vec t$;
\item in each finite oval $\Omega_{i,j}$, $i\in [k]$, $j\in [n-k]$, there is exactly one such divisor point according to the counting rule;
\end{enumerate}
\end{enumerate}
and we have proven that, after the Darboux transformation $D^{(k)}$ associated with the given soliton data $({\mathcal K}, [A])$, the normalized wavefunction ${\tilde \Psi}_{\xi} (\lambda;\vec t)=\frac{D^{(k)}\Psi_{\xi}(\lambda;\vec t)}{D^{(k)}\Psi_{\xi}(\lambda;\vec 0)}$ 
satisfies Dubrovin--Natanzon conditions, {\sl i.e.}
\begin{enumerate}
\item ${\tilde\Psi}_{\xi} (\lambda;\vec t)$ is real for $\lambda\in \Omega_0 \mathop{\cup}_{i,j} \Omega_{i,j}$ and real $\vec t$;
\item on $\Gamma_0$ it has an essential singularity at $P_0$ and it possesses $k$ divisor points $\gamma_{\xi,1}^{(0)},\dots,\gamma_{\xi,k}^{(0)}$ whose position is independent of time;
\item in each $\Gamma_{\xi,i}$, $i\in [k]$, ${\tilde \Psi}_{\xi}(\lambda;\vec t)$ is meromorphic and possesses $n-k-1$ divisor points $\delta_{\xi,1}^{(i)},\dots,\delta_{\xi,n-k-1}^{(i)}$ whose position is independent of $\vec t$;
\item in each finite oval $\Omega_{i,j}$, $i\in [k]$, $j\in [n-k]$, there is exactly one such divisor point according to the counting rule.
\end{enumerate}

\begin{figure}
\includegraphics[scale=0.4]{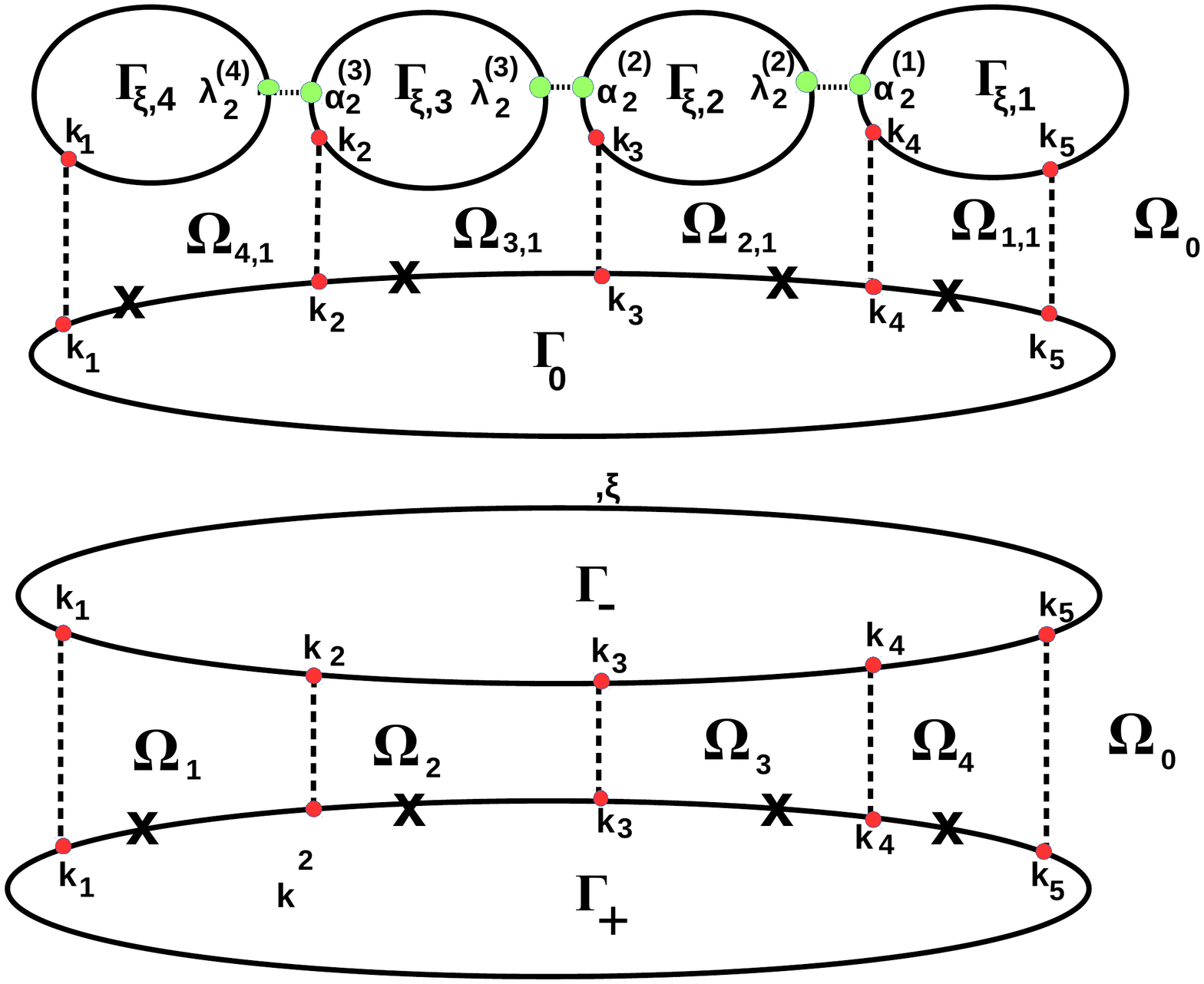}
\vspace{-.7 truecm}
\caption{{\small{\sl The desingularization of $\Gamma_{\xi}$ to $\Gamma$ in the case $n=5$. $\Gamma=\Gamma_+\sqcup\Gamma_-$ (below) is a
desingularization of $\Gamma_{\xi} = \Gamma_0 \sqcup \hat \Gamma_{\xi,-}$ (above), with $\Gamma_-=\Gamma_0$ and ${\hat \Gamma}_{\xi,-} = \Gamma_{\xi,1} \sqcup \cdots \sqcup \Gamma_{\xi,4}$. The divisor points on $\Gamma$ and $\Gamma_{\xi}$ are represented by crosses and are left unchanged by the desingularization.}}}
\label{figure:grn-1}
\end{figure}

Let ${\mathcal K }=\{\kappa_1<\kappa_2 <\cdots <\kappa_n\}$ be fixed and $k=n-1$. According to Corollary \ref{cor:grn-1}, any $[A]\in Gr^{\mbox{\tiny TP}}(n-1,n)$ contains a representative matrix $B$ as in (\ref{eq:Bmat}). 
Let $\xi>>1$ be fixed and $\Gamma_{\xi}=\Gamma_0 \sqcup \Gamma_{\xi,1} \sqcup \cdots \sqcup \Gamma_{\xi,n-1}$ as in Theorem 6 in \cite{AG}. Then $\Gamma_0 = \Gamma_+$ and the double points on $\Gamma_{\xi,r}$, $r\in [n-1]$, in the local coordinate 
$\lambda$ are 
\begin{equation}\label{eq:double}
\lambda^{(r)}_1 =0, \quad\quad \lambda^{(r)}_2 =-1, \quad\quad \alpha^{(r)}_2 = \xi^{-1}.
\end{equation}
On each $\Gamma_{\xi,r}$, $r=1,\dots, n-1$, let us perform the linear substitution
\begin{equation}\label{eq:Mr}
\zeta = M_{\xi}^{(r)} (\lambda)\equiv c^{(r)}_{\xi,0} \lambda + c^{(r)}_{\xi,1},
\end{equation}
where $c^{(r)}_{\xi,1} = \kappa_{n-r}$, and $c^{(r)}_{\xi,0}$ are recursively defined
\[
c^{(1)}_{\xi,0} =\kappa_{n-1}-\kappa_n, \quad\quad c^{(r)}_{\xi,0} =\kappa_{n-r}- M_{\xi}^{(r-1)} (\frac{1}{\xi}) =\sum_{j=0}^{r-1} (-1)^r\frac{\kappa_{n-r+j}-\kappa_{n-r+j+1}}{\xi^j}  , \; r\in [2,n-1].
\]
In the local coordinate $\zeta$ the marked points on $\Gamma_{\xi,r}$ are $\lambda^{(r)}_1 = \kappa_{n-r}$, $r\in [n-1]$, 
\[
\lambda^{(r)}_2 = \left\{ 
\begin{array}{ll}
\kappa_n,                                                    &\mbox{ if } r=1,\\
\frac{c^{(r-1)}_{\xi,0}}{\xi} + \kappa_{n-r+1} , &\mbox{ if } r\in [2,n-1],
\end{array}
\right.
\]
and $\alpha^{(r)}_2= \frac{c^{(r)}_{\xi,0}}{\xi} + \kappa_{n-r}$, $r\in [n-1]$.
If $\xi$ is sufficiently big, $\lambda^{(r)}_2 = \alpha^{(r-1)}_2\in ] \kappa_{n-r}, \kappa_{n-r+1}[$, since $c^{(r-1)}_{\xi,0}
< \kappa_{n-r+1}-\kappa_{n-r+2}<0$; moreover
$\lim\limits_{\xi \to \infty} \alpha^{(r)}_2 = \kappa_{n-r}$, for any $r\in [n-1]$.

\begin{remark}
For any fixed $\xi>>1$, $\Gamma_{\xi}$ itself is a desingularization of $\Gamma_{\infty}= \Gamma_0 \sqcup \Gamma_{\infty,1} \sqcup \cdots \sqcup \Gamma_{\infty,n-1}$. On $\Gamma_{\infty}$, $\Gamma_0$ is glued at $\kappa_n$ with $\Gamma_{\infty,1}$, at $\kappa_1$ with $\Gamma_{\infty,n-1}$, and,
for $r\in [2,n-1]$, at $\kappa_{n-r+1}$ with $\Gamma_{\infty,r-1}$ and $\Gamma_{\infty,r}$.
\end{remark}

Let us denote $\Psi_{\xi, r}$, ${\tilde \Psi}_{r}$ respectively  the vacuum and the normalized Darboux transformed wavefunctions on $\Gamma_{\xi,r}$, $r\in [0,n-1]$.
Following \cite{AG}, on $\Gamma_0= \Gamma_+$, the vacuum wave--function is $\Psi_{\xi, 0} (\zeta;\vec t)  = e^{\theta(\zeta;\vec t)}$ and
$D^{(n-1)} \Psi_{\xi, 0} (\zeta;\vec t)= \prod\limits_{j=1}^{n-1} (\zeta - \gamma^{(n-1)}_j (\vec t)) e^{\theta(\zeta;\vec t)} = \Psi^{(n-1)}_+ (\zeta; \vec t) $ (see (\ref{eq:Psik})).

On each $\Gamma_{\xi,r}$, $r\in [n-1]$, applying the inverse of (\ref{eq:Mr}), $\lambda = M^{-1}_{\xi} (\zeta)$, the vacuum wave--function is $\Psi_{\xi, r}(\zeta;\vec t) = C_r(\xi,\vec t) \frac{\zeta - \chi^{(r)}_1 (\vec t)}{\zeta - b^{(r)}_1 }$, with $C_r(\xi,\vec t)>0$ for all $\vec t$, $\chi^{(r)}_1 (\vec t), b^{(r)}_1 \in ] \kappa_{n-r}, \lambda^{(r)}_2 [$, for all $\vec t$, and $\chi^{(r)}_1 (\vec 0)= b^{(r)}_1$.
After the Darboux transformation, $D^{(n-1)} \Psi_{\xi, r} (\zeta;\vec t)= \frac{{\tilde C}_r(\xi, \vec t) }{\zeta - b^{(r)}_1}$ and the normalized wavefunction  restricted to $\Gamma_{\xi,r}$ is
\begin{equation}\label{eq:Psin-1}
{\tilde \Psi}_{r} (\zeta; \vec t) =\frac{{\tilde C}_r(\xi, \vec t) }{{\tilde C}_r(\xi, \vec 0)} = \phi (\vec t).
\end{equation}
$\phi$ is constant in $\zeta$ and also in $\xi$, since the glueing condition between $\Gamma_0$ and $\Gamma_{\xi,r}$ implies ${\tilde \Psi}_{r} (\kappa_{n-r}; \vec t) = {\tilde \Psi}^{(n-1)}_+ (\kappa_{n-r}; \vec t)$, where the right hand side is independent of $\xi$. Finally $\phi(\vec t)$ does not depend on $r\in [n-1]$ as well, because of the glueing condition between $\Gamma_{\xi,r+1}$ and $\Gamma_{\xi,r}$,
${\tilde \Psi}_{ r+1} (\lambda^{(r+1)}_2 ; \vec t) = {\tilde \Psi}_{ r} (\alpha^{(r)}_2 ; \vec t)$ , $r\in[n-2]$.
In conclusion, the double points $\lambda^{(r)}_2\in \Gamma_{\xi,r}$, $\alpha^{(r-1)}_2\in \Gamma_{\xi,r-1}$, $r=2,\dots, n-1$, are due to the technical conditions posed in \cite{AG}, but they play no role since the normalized KP wave--function
${\tilde \Psi}(\zeta; \vec t)$ takes the same constant value on $\hat\Gamma_{\xi,-} =\Gamma_{\xi,1} \sqcup \cdots \sqcup \Gamma_{\xi,n-1}$  for any $\xi>1$. So we may desingularize $\hat \Gamma_{\xi,-}$ to $\Gamma_-$ without modifying $\tilde \Psi$ for any $\xi>1$.

Finally the normalized wavefunction as in (\ref{eq:Psik}) also takes the constant value $\phi(\vec t)$ for any $\zeta \in\Gamma_-$, ${\tilde \Psi}^{(n-1)}_- (\zeta, \vec t)\equiv\phi(\vec t)$, since it is constant w.r.t. the spectral parameter $\zeta$ and ${\tilde \Psi}^{(n-1)}_- (\kappa_j, \vec t)={\tilde \Psi}^{(n-1)}_+ (\kappa_j, \vec t) = {\tilde \Psi} (\kappa_j,\vec t)$, for all $j\in [n]$ and $\vec t$. We thus have proven

\begin{theorem}\label{theo:n-1}
Let $\xi>>1$ be fixed. Let the soliton data $(\mathcal K, [A])$ be fixed with ${\mathcal K} = \{ \kappa_1 < \cdots < \kappa_n \}$ and $[A] \in Gr^{\mbox {\tiny TP}}(n-1,n)$. Let $D^{(n-1)} = \partial_x^{n-1} -w_1^{(n-1)} (\vec t) \partial_x^{n-2} -\cdots -w_{n-1}^{(n-1)} (\vec t)$ be the Darboux transformation associated to $(\mathcal K, [A])$, ${\tilde \Psi} (\zeta; \vec t)$ be the normalized wave-function on  $\Gamma_{\xi} = \Gamma_0 \sqcup \Gamma_{\xi,1} \sqcup \cdots \sqcup \Gamma_{\xi,n-1}$ constructed in \cite{AG}, and ${\tilde \Psi}^{(n-1)} (\zeta; \vec t)$ be the normalized wave-function on  ${\Gamma} = \Gamma_+ \sqcup \Gamma_-$ as in Theorem \ref{theo:KPToda} for $k=n-1$.  Then
\begin{enumerate}
\item The curve $\Gamma_{\xi} =\Gamma_0\sqcup \hat \Gamma_-$, with $\hat \Gamma_{\xi,-} 	\equiv \Gamma_{\xi,1}\sqcup\ldots\sqcup\Gamma_{\xi,n-1}$, is the rational degeneration of a regular hyperelliptic curve of genus
$g=n-1$.  $\Gamma_{\xi}$ may be desingularized to $\Gamma$,  where $\Gamma_0 =\Gamma_+$ and
$\hat \Gamma_-$ is $\Gamma_-$ with
the extra double points $\lambda^{(r)}_2	\in \Gamma_{\xi,r}$ and  $\alpha^{(r-1)}_2\in \Gamma_{\xi,r-1}$, at which we connect $\Gamma_{\xi,r}$ to $\Gamma_{\xi,r-1}$, for $r=2,\dots,n-1$;
\item The wavefunctions associated to $\Gamma_{\xi}$ and to its desingularization $\Gamma$ are the same. More precisely, for any $\zeta\in \Gamma_0\equiv \Gamma_+$ and for any $\vec t$, 
\[
{\tilde \Psi}_{\xi,0} (\zeta; \vec t) = {\tilde \Psi}^{(n-1)}_+ (\zeta; \vec t) = \left(\frac{\zeta^{n-1} -w_1^{(n-1)} (\vec t) \zeta^{n-2} -\cdots -w_{n-1}^{(n-1)} (\vec t)}{\zeta^{n-1} -w_1^{(n-1)} (\vec 0) \zeta^{n-2} -\cdots -w_{n-1}^{(n-1)} (\vec 0)}\right) e^{\theta(\zeta;\vec t)},
\]
and there exists a regular function $\phi (\vec t)$ which is the common value respectively of ${\tilde \Psi} (\zeta; \vec t)$ on $\hat\Gamma_{\xi,-}$ and of ${\tilde \Psi}^{(n-1)}_- (\zeta; \vec t)$ on $\Gamma_-$:  
\[
\phi(\vec t)= {\tilde \Psi}^{(n-1)}_+ (\kappa_j; \vec t), \quad\quad \forall j\in [n], \;\forall \vec t.
\]
\end{enumerate}
\end{theorem}

In Figure \ref{figure:grn-1}, we show the desingularization of $\Gamma_{\xi}$ to $\Gamma$ when $n=5$.

Let us fix the soliton data $({\mathcal K}, [A])$, with $[A]\in Gr^{\mbox{\tiny TP}}(n-1,n)$.  
The representative matrix 
{\small\begin{equation}\label{eq:Ahat}
A =\left( \begin{array}{cccccc}
1      & \frac{x_1}{x_2} & 0               & \cdots & \cdots               & 0\\
0      & 1               & \frac{x_2}{x_3} & 0      & \cdots               & 0 \\
\vdots & \ddots          & \ddots          & \ddots & \ddots               & \vdots \\
0      & \cdots          & 0               & 1      & \frac{x_{n-2}}{x_{n-1}} & 0   \\
0      & \cdots          & \cdots          & 0      &  1                      & x_{n-1} 
\end{array}
\right),
\end{equation}
}
is equivalent to $A$ in (\ref{eq:An-1}), and the 
Darboux transformation for $([A], \mathcal K)$, $D^{(n-1)} = \partial_x^{n-1}-w_1^{(n-1)} (\vec t) \partial_x^{n-2} -\cdots - w_{n-1}^{(n-1)} (\vec t)$, has the following kernel $f^{(r)}
(\vec t) = x_{n-r+1} E_{n-r} (\vec t) + x_{n-r}E_{n-r+1} (\vec t)$, $r\in [n-1]$.

\begin{corollary}\label{cor:divM1} The pole divisor ${\mathcal D}^{(n-1)} = \{ \gamma^{(n-1)}_1, \dots ,\gamma^{(n-1)}_{n-1} \}\subset \Gamma_+$ associated to the soliton data $([A], \mathcal K)$ satisfies
\begin{equation}\label{eq:divM1}
x_{i+1} \prod\limits_{l=1}^n ( \kappa_i - \gamma^{(n-1)}_l ) +x_{i} \prod\limits_{l=1}^n ( \kappa_{i+1} - \gamma^{(n-1)}_l ) =0, \quad \quad i=1,\dots,n-1,
\end{equation}
with $\gamma^{(n-1)}_l \in ] \kappa_l, \kappa_{l+1}[\cap \Gamma_+$, for any $l\in [n-1]$.
\end{corollary}

Identities (\ref{eq:divM1}) are easily deduced from the equations $D^{(n-1)} f^{(r)} =0$, $r=1,\dots,n$.

Soliton data in $Gr^{\mbox{\tiny TP}} (n-1,n)$, are rather special since the divisor satisfies ${\mathcal D}^{(n-1)} (\vec t)\subset \Gamma_+$. 

\begin{proposition}
Let us fix the soliton data $({\mathcal K}, [A])$, with $[A]\in Gr^{\mbox{\tiny TP}} (n-1,n)$.
Let $D^{(n-1)}$ be the Darboux transformation associated to $({\mathcal K}, [A])$ and $\Gamma$ as in (\ref{eq:hyprat}). For any given $[c]\in Gr^{\mbox{\tiny TP}}(1,n)$, let the vacuum wavefunction be
\[
\Psi_{[ c]} (\zeta; \vec t) = \left\{\begin{array}{ll}  e^{\theta(\zeta; \vec t)}&\quad \zeta \in \Gamma_+,\\
\displaystyle \sum_{j=1}^n \frac{c_j E_j(\vec t)}{\sum_{l=1}^n c_l } \frac{\prod_{s\not = j}^n (\zeta -\kappa_s)}{\prod_{r = 1}^{n-1} (\zeta -b_r(c))}. & \quad \zeta \in \Gamma_-.
\end{array}
\right.
\]
Then,
${\tilde \Psi}^{(n-1)}_{[ c]} (\zeta; \vec t) \equiv \frac{D^{(n-1)}\Psi_{[c]} (\zeta; \vec t)}{D^{(n-1)}\Psi_{[c]} (\zeta; \vec 0)} =\phi (\vec t)$, for all $\zeta \in \Gamma_-$ and for all $\vec t$, 
with $\phi(\vec t)$ as in Theorem \ref{theo:n-1}.
\end{proposition}

The proof is trivial and is omitted.
The above proposition means that, for any  $[c]\in Gr^{\mbox{\tiny TP}}(1,n)$, the zero divisor of the {\sl un--normalized} wave--function, $D^{(n-1)}\Psi_{[c]} (\zeta; \vec t)$ is ${\mathcal D}(\vec t) \cup \{ P_{[c],1},\dots,P_{[c],n-1}\}$, where the points $P_{[c],j} \in \Gamma_-$ are independent of $\vec t$, for all $j\in [n-1]$.
However, there is a unique point $[a]\in Gr^{\mbox{\tiny TP}}(1,n)$ such that the zero divisor of the {\sl un--normalized} wavefunction $D^{(n-1)} \Psi_{[a]} (\zeta, \vec t)$ is ${\mathcal D}^{(n-1)}\cup \{ (n-1) P_-\}$. This is one of the reasons why we have defined $k$--compatibility for the un-normalized wave--function $D^{(k)}\Psi$ instead that for the normalized wavefunction ${\tilde \Psi}^{(k)}$.

\section{The finite non--periodic Toda lattice hierarchy}\label{sec:introToda}

The special form of the KP $\tau$--functions associated to $T$--hyperelliptic solitons relates such class of solutions to those of the finite non--periodic Toda system\cite{T,Mos}. Also $\Gamma$ is related to the algebraic--geometric description of the finite non--periodic Toda hierarchy\cite{Mar,BM,KV}. For these reasons, it is then natural to investigate the relations between the algebraic--geometric approach for the two systems. In this section we review known results on the solutions to the Toda hierarchy and on the Toda Baker--Akhiezer function and then, in the following sections, we discuss the relations between the two systems.

Toda \cite{T} proposed a model of a chain of $n$ mass points moving on the real axis, with position ${\mathfrak q}_l$, $l\in [n]$, Hamiltonian
$\displaystyle
H = \frac{1}{2} \sum\limits_{l=1}^n {\mathfrak p}_l^2 +  \sum\limits_{l=1}^{n-1} e^{({\mathfrak q}_l -{\mathfrak q}_{l+1})},$
which is integrable both in the periodic and non--periodic case \cite{Fla,Mos}. The Toda system is the first flow of an integrable hierarchy and it has been generalized in many ways \cite{DLNT,EFS, JM, Kos, LS, N,UT}.

The finite non--periodic Toda lattice system corresponds to 
formal boundary conditions ${\mathfrak q}_0=-\infty$, ${\mathfrak q}_{n+1}=+\infty$, and, under the transformation 
${\mathfrak a}_k =  e^{{\mathfrak q}_k -{\mathfrak q}_{k+1}}$, $k\in [n-1]$, ${\mathfrak b}_k = -{\mathfrak p}_k$, $k\in [n]$,
it is equivalent to
\begin{equation}\label{eq:t1}
\displaystyle \frac{d{\mathfrak a}_k}{dt_1} = {\mathfrak a}_k \left( {\mathfrak b}_{k+1}-{\mathfrak b}_k\right), \quad k\in [n-1],
\quad\quad
\frac{d{\mathfrak b}_k}{dt_1} =  {\mathfrak a}_{k}-{\mathfrak a}_{k-1}, \quad k\in [n],\\
\end{equation}
with boundary conditions ${\mathfrak a}_0={\mathfrak a}_{n}=0$. The space of configurations in the new variables is
\begin{equation}\label{eq:confsp}
D = \left\{ ({\mathfrak a}, {\mathfrak b}) \in {\mathbb R}^{n-1} \times {\mathbb R}^n \; : \; {\mathfrak a}_k >0, \, k\in [n-1] \right\}.
\end{equation}
The system (\ref{eq:t1}) may be put in Lax form
$\displaystyle \frac{d{\mathfrak A}}{dt_1} = [ {\mathfrak B}_1, {\mathfrak A}]$,
with
{\small
\begin{equation}\label{eq:A}
{\mathfrak A}=\left( \begin{array}{ccccc}
\mathfrak{b}_1 & \mathfrak{a}_1 & 0              & \cdots             & 0                  \\
1              & \mathfrak{b}_2 & \mathfrak{a}_2 &\ddots              & \vdots             \\
0              & \ddots         & \ddots         &\ddots              & 0                  \\
\vdots         & \ddots         & 1              & \mathfrak{b}_{n-1} & \mathfrak{a}_{n-1} \\ 
0              & \cdots         & 0              & 1                  & \mathfrak{b}_{n}
\end{array}\right),
\end{equation}
}
${\mathfrak B}_1= \left( {\mathfrak A} \right)_+$, where $(P)_+$ denotes the strict upper triangular part of the matrix $P$. (\ref{eq:t1}) is the first flow of an integrable hierarchy 
\begin{equation}\label{eq:Todafj}
\frac{d{\mathfrak A}}{dt_j} = [ {\mathfrak B}_j, {\mathfrak A}], \quad\quad j\ge 1,\quad\quad {\mathfrak B}_j = \left( {\mathfrak A}^j \right)_+.
\end{equation}
(\ref{eq:Todafj}) are 
the equations associated to the symmetries of the Toda lattice generated by $H_j = \frac{1}{j+1}\mbox{ Tr }{\mathfrak A}^{j+1}$. 
Since the 0-th flow is trivial, $\frac{d{\mathfrak A}}{dt_0} \equiv 0$, in the following we take $t_0=0$ and denote
$\vec t=(t_1, t_2, t_{3}, \dots)$.

In the configuration space $D$, the eigenvalues of ${\mathfrak A}$, $\kappa^{(T)}_j$ are real, distinct and independent of all $t_j$, {\sl i.e.} they are constants of the motion. We order them in increasing order, $\kappa^{(T)}_1< \kappa^{(T)}_2 <\cdots <\kappa^{(T)}_n$, and denote the characteristic polynomial and the resolvent
of $\mathfrak A(\vec t)$, respectively, 
\[
\Delta_n (\zeta) = \; {\rm det} \; \left( \zeta I - {\mathfrak A} \right) = \prod\limits_{j=1}^n (\zeta -\kappa^{(T)}_j).
\quad\quad
{\mathfrak R} (\zeta; \vec t) = \left(\zeta {\mathfrak I}_n - {\mathfrak A} (\vec t) \right)^{-1}.
\]
Let ${\Delta}_j(\zeta; \vec t)$, ${\hat  \Delta}_j(\zeta; \vec t)$, $j\in [n]$, respectively be the minors formed by the last $j$ rows and columns and by the first $j$ rows and columns of $\zeta {\mathfrak I}_n -{\mathfrak A}$, $j\in [n-1]$:
\begin{equation}\label{eq:delta}
\displaystyle \Delta_j (\zeta; \vec t)= \mbox{\rm det } \left( \zeta {\mathfrak I}_n - {\mathfrak A}(\vec t) \right)_{
\left[  n-j+1 , \dots , n \right]}, \quad\quad
\displaystyle {\hat \Delta}_j (\zeta; \vec t)= \mbox{\rm det } \left( \zeta {\mathfrak I}_n - {\mathfrak A} (\vec t)\right)_{ 
\left[ 1 , \dots , j  \right]}, 
\end{equation}
where $\Delta_0 (\zeta;\vec t)\equiv {\hat  \Delta}_0 (\zeta;\vec t)\equiv 1$. The vectors $(\Delta_{n-1},\dots,\Delta_{0})$ and $({\hat \Delta}_0,\dots,{\hat\Delta}_{n-1}) $ are respectively the first column and the last row of $\Delta_n (\zeta)  \, {\mathfrak R} (\zeta;\vec t)$ and the following identities hold for $j\in [n-1]$,
\begin{equation}\label{eq:idD}
\begin{array}{l}
\Delta_{j+1}(\zeta; \vec t) = (z- \mathfrak{b}_{n-j}(\vec t)) \Delta_j (\zeta; \vec t)-{\mathfrak a}_{n-j}(\vec t) \Delta_{j-1}(\zeta; \vec t),
\\
{\hat \Delta}_{j+1} (\zeta; \vec t)= (z- \mathfrak{b}_{j+1}(\vec t)) {\hat \Delta}_j (\zeta; \vec t)-{\mathfrak a}_j (\vec t){\hat \Delta}_{j-1}(\zeta; \vec t),
\\
\Delta_n(\zeta) ={\hat \Delta}_n (\zeta)= \Delta_{n-j} (\zeta; \vec t){\hat \Delta}_{j}(\zeta; \vec t) -{\mathfrak a}_j (\vec t)\Delta_{n-j-1}(\zeta; \vec t) {\hat \Delta}_{j-1}(\zeta; \vec t).
\end{array}
\end{equation}

\begin{remark}
We remark that reflection w.r.t. to the anti--diagonal transforms $\mathfrak A (\vec t)$ into ${\mathfrak A}^* (\vec t)$ where
{\small
\begin{equation}\label{eq:Astar}
{\mathfrak A}^{*} (\vec t)=\left( \begin{array}{ccccc}
\mathfrak{b}_n (\vec t)& \mathfrak{a}_{n-1} (\vec t) & 0              & \cdots             & 0                  \\
1              & \mathfrak{b}_{n-1} (\vec t) & \mathfrak{a}_{n-2} (\vec t)&\ddots              & \vdots             \\
0              & \ddots         & \ddots         &\ddots              & 0                  \\
\vdots         & \ddots         & 1              & \mathfrak{b}_{2}(\vec t) & \mathfrak{a}_{1}(\vec t) \\ 
0              & \cdots         & 0              & 1                  & \mathfrak{b}_{1}(\vec t)
\end{array}\right).
\end{equation}
}
Such transformation preserves the spectrum and inverts the role of the minors since
\begin{equation}\label{eq:deltastar}
\begin{array}{l}
\displaystyle \Delta_j^* (\zeta; \vec t)\equiv  \mbox{\rm det } \left( \zeta {\mathfrak I}_n - {\mathfrak A}^*(\vec t) \right)_{
\left[  n-j+1 , \dots , n \right]} ={\hat \Delta}_j (\zeta; \vec t), \\\
\displaystyle {\hat \Delta}_j^* (\zeta; \vec t)\equiv \mbox{\rm det } \left( \zeta {\mathfrak I}_n - {\mathfrak A}^* (\vec t)\right)_{ 
\left[ 1 , \dots , j  \right]}=\Delta_j (\zeta; \vec t),\\ 
\end{array}	\quad	\quad j\in[0,n].
\end{equation}
As a consequence we may equivalently use $\Delta_j$ or ${\hat \Delta}_j$ to represent IVP solutions of the finite Toda system (\ref{eq:Todafj}). It is well known that assigning the initial datum ${\mathfrak A} (\vec 0)$ in the configuration space $D$ as in (\ref{eq:confsp}) is equivalent to the Toda data $({\mathcal K}, [a])$, where ${\mathcal K} = \{\kappa^{(T)}_1 < \cdots <\kappa^{(T)}_n\}$ is the spectrum of ${\mathfrak A}_0$ and $[ a] \in Gr^{\mbox{\tiny TP}} (1,n)$, with   
$a_l =\mathop{\mbox{\rm Res}}_{\zeta=\kappa^{(T)}_l}   {\mathfrak R}_{11} (\zeta; \vec 0) $, $l\in [n]$. The reflection transformation induces a transformation of  $Gr^{\mbox{\tiny TP}} (1,n)$ into itself, and ${\mathfrak A} (\vec 0)$ is also associated to the reflected set $({\mathcal K}, [{\hat a}])$, with ${\mathcal K}$ as before and $[{\hat  a}] \in Gr^{\mbox{\tiny TP}} (1,n)$ uniquely identified by the residues of ${\mathfrak R}_{nn} (\zeta; \vec 0) $. 
\end{remark}

\subsection{IVP solutions to the Toda hierarchy}
We recall below different characterizations of the solutions to the IVP ${\mathfrak A} (\vec 0) ={\mathfrak A}_0$ for the Toda hierarchy flows (\ref{eq:Todafj}). Let us define the generating functions
\begin{equation}\label{eq:gen}
\begin{array}{l}
\displaystyle  {\mathfrak f} (\zeta;\vec t) \equiv \prec e_1 , {\mathfrak R}(\zeta; \vec t) e_1 \succ \equiv \frac{\Delta_{n-1} (\zeta;\vec t)}{\Delta_{n} (\zeta)} = \sum\limits_{j\ge 0} \frac{  h_j (\vec t)}{\zeta^{j+1}} = \sum\limits_{j=1}^n \frac{  {\mathfrak M}_j(\vec t)}{\zeta-\kappa^{(T)}_j},
\\
\displaystyle {\hat {\mathfrak f}} (\zeta;\vec t) \equiv \prec e_n , {\mathfrak R}(\zeta; \vec t) e_n \succ \equiv \frac{{\hat \Delta}_{n-1} (\zeta;\vec t)}{\Delta_{n} (\zeta)}  =\sum\limits_{j\ge 0} \frac{{\hat h}_j (\vec t)}{\zeta^{j+1}}= \sum\limits_{j=1}^n \frac{{\hat {\mathfrak M}}_j (\vec t)}{\zeta-\kappa^{(T)}_j}.
\end{array}
\end{equation}

For any initial datum ${\mathfrak A}_0$ in the configuration space $D$, the exponential matrix $\psi(\vec t)$, $\vec t= (t_1,\dots, t_{s}), $ $s\ge n- 1$, admits a Bruhat decomposition
\begin{equation}\label{eq:psi}
\psi (\vec t) = \exp \left( {\mathfrak A}_0 t_1 + \cdots +{\mathfrak A}_0^{s} t_{s} \right) =
{\mathfrak L} (\vec t) {\mathfrak U} (\vec t),
\end{equation}
where ${\mathfrak L} (\vec t) $ is lower triangular with positive entries on the diagonal and 
${\mathfrak U} (\vec t)$ is unit upper triangular, ${\mathfrak L} (\vec 0) = {\mathfrak U} (\vec 0) = {\mathfrak I}_n$. Such decomposition gives the following explicit characterization of the IVP solution to (\ref{eq:Todafj}).

\begin{proposition}\label{prop:psiA}
Let ${\mathfrak A}_0$ be a Jacobi matrix of the form (\ref{eq:A}) with $a_j(\vec 0)>0$, $j\in [n-1]$ and eigenvalues $\kappa^{(T)}_1 < \kappa^{(T)}_2 < \cdots < \kappa^{(T)}_n$. Define $\psi(\vec t)$, ${\mathfrak f} (\zeta;\vec t)$ and ${\hat  {\mathfrak f} } (\zeta; \vec t)$ as in (\ref{eq:psi}) and (\ref{eq:gen}). Then
\[
{\mathfrak A} (\vec t) = {\mathfrak L} (\vec t)^{-1} {\mathfrak A}_0 {\mathfrak L} (\vec t) = {\mathfrak U} (\vec t) {\mathfrak A}_0 {\mathfrak U} (\vec t)^{-1},
\]
Finally, in this representation of the solution the generating functions and the Hankel coefficients take the form
\[
\begin{array}{ll}
\displaystyle {\mathfrak f} (\zeta; \vec t) = \frac{\prec e_1, \psi(\vec t) (\zeta {\mathfrak I}_n - {\mathfrak A}_0)^{-1} e_1\succ}{\prec e_1, \psi(\vec t)e_1\succ}, &\displaystyle \quad\quad {\hat {\mathfrak f}}(\zeta;\vec t) = \frac{\prec e_n, (\zeta {\mathfrak I}_n - {\mathfrak A}_0)^{-1}\psi^{-1}(\vec t)  e_n\succ}{\prec e_n, \psi^{-1}(\vec t)e_n\succ},
\\
\displaystyle
h_j (\vec t)= \frac{\prec e_1, \psi(\vec t) {\mathfrak A}_0^{j} e_1\succ}{\prec e_1, \psi(\vec t)e_1\succ}, &\displaystyle \quad\quad {\hat h}_j(\vec t) = \frac{\prec e_n, \psi^{-1}(\vec t) {\mathfrak A}_0^{j}  e_n\succ}{\prec e_n, \psi^{-1}(\vec t)e_n\succ}, \quad\quad j\ge 0.
\end{array}
\]
\end{proposition}

\begin{proposition}\label{prop:mu}
Under the hypotheses of Proposition \ref{prop:psiA}, let us define, for $j\ge 0$,
\begin{equation}\label{eq:mu}
\mu_j (\vec t)\equiv \prec e_1, \psi(\vec t ) {\mathfrak A}_0^{j} e_1 \succ = \partial_{t_1}^j \mu_0 (\vec t), \quad\quad
{\hat \mu}_j (\vec t)\equiv\prec e_n, \psi^{-1}(\vec t) {\mathfrak A}_0^{j} e_n \succ = (-1)^j\partial_{t_1}^j {\hat \mu}_0 (\vec t).
\end{equation}
Then $\mu_{l+j} (\vec t) = \partial_{t_l} \mu_j (\vec t) = \partial_{t_1}^l \mu_j (\vec t) =\partial_{t_j} \mu_l (\vec t) = \partial_{t_1}^j \mu_l (\vec t)$, $\forall l,j \ge 0$, and
\begin{equation}\label{eq:antiheat}
{\hat \mu}_{j+l} (\vec t) = -\partial_{t_l} {\hat \mu}_j (\vec t) = -\partial_{t_j} {\hat \mu}_l (\vec t), \quad
\partial_{t_l} {\hat \mu}_j (\vec t) = (-1)^{l+1}\partial_{t_1}^l {\hat \mu}_j (\vec t), \quad\quad\forall l,j \ge0.
\end{equation}
\end{proposition}

The converse to Proposition \ref{prop:mu} also holds true.

\begin{proposition}\label{prop:mu0}
Let $\mu_0 (\vec t)$ be a solution to the heat hierarchy, $\mu_j(\vec t)\equiv \partial_{t_1}^j \mu_0 (\vec t)$, $\forall j\ge 1$, and define the Hankel matrix $H_{\mu} (\vec t) = \left[ \mu_{i+j-2}(\vec t) \right]_{i,j\ge 1}$. Then 
$\displaystyle  {\mathfrak f}_{\mu} (\zeta; \vec t) = \sum\limits_{j\ge 0} \frac{\mu_j (\vec t)}{\mu_0 (\vec t)}\zeta^{-(j+1)}$ generates a solution to the Toda hierarchy (\ref{eq:Todafj}) in the configuration space $D$ as in (\ref{eq:confsp}), if and only if $H_{\mu} (\vec t)$ has finite rank $n$ with principal minors $\mbox{\rm det } H_{\mu, j} >0$, for all $j\in [n]$, that is if and only if there exists $({\mathcal K}, [ a])$ with $\mathcal K = \{ \kappa^{(T)}_1 < \cdots < \kappa^{(T)}_n\}$ and $[ a]\in Gr^{\mbox{\tiny TP}} (1,n)$, such that $\mu_0 (\vec t)=\sum_{j=1}^n a_j E_j(\vec t)$. 

Similarly, let ${\hat \mu}_0 (\vec t)$ be a solution to (\ref{eq:antiheat}), ${\hat \mu}_j(\vec t)= (-1)^{j}\partial_{t_1}^j  {\hat \mu}_0 (\vec t)$, $\forall j \ge 1$, and define the Hankel matrix ${\hat H}_{{\hat \mu}} (\vec t) = \left[ {\hat \mu}_{i+j-2}(\vec t) \right]_{i,j\ge 1}$. Then 
$\displaystyle {\hat {\mathfrak f}}_{\mu} (\zeta; \vec t) = \sum\limits_{j\ge 0} \frac{{\hat \mu}_j (\vec t)}{{\hat \mu}_0 (\vec t)}\zeta^{-(j+1)}$ generates a solution to the Toda flows system (\ref{eq:Todafj}) in the configuration space $D$ if and only if ${\hat H}_{{\hat \mu}} (\vec t)$  has finite rank $n$ with principal minors $\mbox{\rm det } {\hat H}_{{\hat \mu}, j} >0$, for all $j\in [n]$, that is if and only if there exists $({\mathcal K}, [\hat a])$ with $\mathcal K = \{ \kappa^{(T)}_1 < \cdots < \kappa^{(T)}_n\}$ and $[\hat a]\in Gr^{\mbox{\tiny TP}} (1,n)$, such that ${\hat \mu}_0 (\vec t)=\sum_{j=1}^n {\hat a}_j E_j(-\vec t)$. 
\end{proposition}

The explicit form of the solution is then given in the following Proposition

\begin{proposition}\label{prop:aaa}
Let ${\mathfrak A}_0$ be a Jacobi matrix as in (\ref{eq:A}) in the configuration space $D$ (\ref{eq:confsp}) and let simple spectrum be $\kappa^{(T)}_1 < \kappa^{(T)}_2 < \cdots < \kappa^{(T)}_n$. Let 
${\hat {\mathfrak f}} (\zeta;\vec 0) = \prec e_n , (\zeta {\mathfrak I}_n - {\mathfrak A}_0)^{-1} e_n \succ$, 
${\mathfrak f} (\zeta;\vec 0) = \prec e_1 , (\zeta {\mathfrak I}_n - {\mathfrak A}_0)^{-1} e_1 \succ$. Let 
\[
\begin{array}{l}
\displaystyle a_l \equiv{\mathfrak M}_l (\vec 0)= \mathop{\mbox{\rm Res}}_{\zeta=\kappa^{(T)}_l}  {\mathfrak f} (\zeta;\vec 0) =\displaystyle \frac{  \Delta_{n-1} (\kappa^{(T)}_l, \vec 0)}{\prod\limits_{s\not = l} (\kappa^{(T)}_l-\kappa^{(T)}_s)}, 
\\
\displaystyle {\hat a}_l \equiv {\hat {\mathfrak M}}_l (\vec 0)=\mathop{\mbox{\rm Res}}_{\zeta=\kappa^{(T)}_l}{\hat {\mathfrak f}} (\zeta;\vec 0) =
\frac{{\hat\Delta}_{n-1} (\kappa^{(T)}_l, \vec 0)}{\prod\limits_{s\not = l} (\kappa^{(T)}_l-\kappa^{(T)}_s)},\quad
l\in [n],
\end{array}
\]
$\theta_l(\vec t) = \sum\limits_{j\ge 1} \left(\kappa^{(T)}_l\right)^j t_l$ and $E_l(\vec t)= \exp(\theta_l(\vec t))$.
Let $\tau_0 (\vec t)\equiv 1$, $\tau_1(\vec t) = \mu_0 (\vec t)= \sum\limits_{j=1}^n  a_j E_j (\vec t)$, 
$\tau_j (\vec t) = \mbox{{\rm Wr}}_{t_1} ( \mu_0(\vec t), \partial_{t_1} \mu_0(\vec t), \cdots ,\partial_{t_1}^{j-1} \mu_0(\vec t))$, $j\ge 2$.

Then 
$\mu_0 (\vec t) $  generates the solution to (\ref{eq:Todafj}) with initial condition ${\mathfrak A} (\vec 0)= {\mathfrak A}_0$ and
$\displaystyle {\mathfrak M}_l (\vec t)= \frac{a_l E_l(\vec t)}{\sum\limits_{j=1}^n a_j E_j(\vec t)}$,  $l\in [n]$,
\begin{equation}\label{eq:abtau}
\displaystyle {\mathfrak a}_k (\vec t) =\frac{\tau_{k-1}(\vec t) \tau_{k+1}(\vec t)}{\tau^2_{k}(\vec t)}, \;\; k\in [n-1], \quad\quad {\mathfrak b}_k (\vec t) = \frac{\partial_{t_1}\tau_{k}(\vec t)}{\tau_{k}(\vec t)} -\frac{\partial_{t_1}\tau_{k-1}(\vec t)}{\tau_{k-1}(\vec t)}, \;\; k\in [n].
\end{equation}
Similarly ${\hat \mu}_0(	\vec t)=\sum\limits_{j=1}^n {\hat a}_j \exp(-\theta_j(\vec t))$ generates the same solution and, up to a multiplicative constant $c>0$
\begin{equation}\label{eq:adual}
{\hat  a}_j a_j = c \mathop{\prod\limits_{1\le i < l\le n} }_{i,l\not = j}(\kappa^{(T)}_l-\kappa^{(T)}_i)^{-2}, \quad\quad j\in [n].
\end{equation}
\end{proposition}

\begin{remark}\label{rem:toda1}
The meaning of the above Proposition is the following: the solution of the Toda hierarchy (\ref{eq:Todafj}) with initial condition ${\mathfrak A} (\vec 0) = {\mathfrak A}_0$, is thus uniquely identified by the data  $({\mathcal K}, [a])$, where ${\mathcal K}= \{\kappa^{(T)}_1<\kappa^{(T)}_2 <\cdots <\kappa^{(T)}_n\}$ is the spectrum of ${\mathfrak A}_0$ and $[ a] \in Gr^{TP} (1,n)$, via the heat hierarchy solution $\mu_0 (\vec t)= \sum_{j=1}^n a_j E_j(\vec t)$ and $f (\zeta; \vec t)$. The same solution may be also associated to the reflected set of data $({\mathcal K}, [{\hat a}])$, with $[\hat a]$ related to $[a]$ via (\ref{eq:adual}), using ${\hat \mu}_0 (\vec t)= \sum_{j=1}^n {\hat a}_j E_j(-\vec t)$ and ${\hat f} (\zeta; \vec t)$.
\end{remark}

\subsection{The spectral problem and the Baker--Akhiezer function for the finite Toda system}

The idea of singularizing the smooth spectral curve of the periodic Toda to obtain a spectral curve for the finite Toda system goes back to Mc Kean \cite{McK}. More recently new interest in the problem \cite{Mar, BM} has come from the connections of the Toda lattice with Seiberg--Witten theory of supersymmetric $SU(n)$ gauge theory. The spectral curve proposed for the open Toda lattice in \cite{Mar, BM} is determined by the equation
\begin{equation}\label{eq:Todacurve}
{\hat \eta} = \prod\limits_{j=1}^n (\zeta -\kappa^{(T)}_j),
\end{equation} 
considered as the limit $\epsilon\to 0$ of the hyperelliptic spectral curve 
\[
{\hat \eta} + \frac{\epsilon^2}{4\hat \eta} = \prod\limits_{j=1}^n (\zeta -\kappa^{(T)}_j),
\]
of the periodic Toda system.
In \cite{KV}, the Baker-Akhiezer function approach is used to provide a solution to the inverse spectral problem for the singular curve (\ref{eq:Todacurve}) and action--angle variables are constructed following the approach in \cite{KP1,KP2}.
In \cite{KV} they use the self-adjoint representation of the finite non--periodic Toda system and introduce the following finite--dimensional operators
{\small
\begin{equation}\label{eq:Aw}
\begin{array}{l}
{\mathfrak A}^{(sa)}_w=\left( \begin{array}{ccccc}
\mathfrak{b}_1 & \sqrt{\mathfrak{a}_1} & 0              & \cdots             & 0                  \\
\sqrt{\mathfrak{a}_1}             & \mathfrak{b}_2 & \sqrt{\mathfrak{a}_2} &\ddots              & \vdots             \\
\vdots             & \ddots         & \ddots         &\ddots              & 0                  \\
0         & \ddots         & 1              & \mathfrak{b}_{n-1} & \sqrt{\mathfrak{a}_{n-1}} \\ 
\frac{w}{\sqrt{\prod_{s=1}^{n-1} \mathfrak{a}_s}}              & \cdots         & 0              & \sqrt{\mathfrak{a}_{n-1}}                  & \mathfrak{b}_{n}
\end{array}\right),\\
{\mathfrak B}^{(sa)}_w=\frac{1}{2}\left( \begin{array}{ccccc}
0                        & \sqrt{\mathfrak{a}_1} & 0              & \cdots             & 0                  \\
-\sqrt{\mathfrak{a}_1}   & 0                     & \sqrt{\mathfrak{a}_2} &\ddots              & \vdots             \\
\vdots             & \ddots         & \ddots         &\ddots              & 0                  \\
0                  & \ddots         & 1              & 0 & \sqrt{\mathfrak{a}_{n-1}} \\ 
\frac{w}{\sqrt{\prod_{s=1}^{n-1} \mathfrak{a}_s}}              & \cdots         & 0              & -\sqrt{\mathfrak{a}_{n-1}}                  & 0
\end{array}\right),
\end{array}
\end{equation}
}
The time--dependent Baker--Akhiezer functions  $ \Psi^{(T)}_j (\zeta,\vec t)$, $ \Psi^{(T)\sigma}_j (\zeta,\vec t)$ are, respectively, eigenvectors of  ${\mathfrak A}^{(sa)}_w \Psi^{(T)}  =\zeta \Psi^{(T)} $ and of ${\mathfrak A}^{(sa),\sigma}_w \Psi^{(T),\sigma}  =\zeta \Psi^{(T),\sigma} $, where ${\mathfrak A}^{(sa),\sigma}_w $ is the adjoint of ${\mathfrak A}^{(sa)}_w$. 
Their components take the form
\begin{equation}\label{eq:TodaBA}
\begin{array}{l}
\Psi^{(T)}_j (\zeta, \vec t) = e^{\frac{\theta (\zeta, \vec t)}{2}} \left(  \sum\limits_{l=0}^j c_j(\vec t, n)  \zeta^l  \right),\\
\displaystyle \Psi^{(T),\sigma}_j (\zeta, \vec t) = e^{-\frac{\theta (\zeta, \vec t)}{2}} \left( \frac{ \sum\limits_{i=0}^{n-j-1} c_j^{\sigma}(\vec t, n)  \zeta^i }{\prod_{s=1}^{n-1} (\zeta -b^{(T)}_s}) \right),
\end{array} \quad\quad j\in [0,n-1],
\end{equation}
where $\theta (\zeta, \vec t) = \sum_{j\ge 1} t_j \zeta^j$, and the coefficients in (\ref{eq:TodaBA}) are uniquely defined by the gluing conditions 
\[
\Psi^{(T)}_j (\kappa^{(T)}_l, \vec t) = \Psi^{(T),\sigma}_j (\kappa^{(T)}_l, \vec t),\quad\quad l\in[n],
\]
and the normalization $c_j(j) c_j^{\sigma}(n-i-j)=1$. 
In particular, the Toda divisor ${\mathcal D}^{(T)} = \{ b^{(T)}_1 < \cdots < b^{(T)}_{n-1}\}$ is the spectrum of the matrix obtained from ${\mathfrak A}^{(sa)}_0 (\vec 0)$ deleting the first row and the first column and
\[
\kappa^{(T)}_1 < b^{(T)}_1 < \kappa^{(T)}_2 < \cdots < \kappa^{(T)}_{n-1} < b^{(T)}_{n-1}< \kappa^{(T)}_n.
\]
Finally, the Toda Baker--Akhiezer components are uniquely recovered from the recurrences 
\begin{equation}\label{eq:rec1}
\begin{array}{l}
\zeta \Psi^{(T)}_0 = \sqrt{\mathfrak{a}_1} \Psi^{(T)}_1 +  \mathfrak{b}_1\Psi^{(T)}_0,\\
\zeta \Psi^{(T)}_j = \sqrt{\mathfrak{a}_{j}} \Psi^{(T)}_{j+1} +  \mathfrak{b}_{j+1}\Psi^{(T)}_j+\sqrt{\mathfrak{a}_j} \Psi^{(T)}_{j-1},\quad\quad j\in [N-2],
\end{array}
\end{equation}
\begin{equation}\label{eq:rec2}
\begin{array}{l}
\Psi^{(T),\sigma}_j = \frac{\Phi^{(T)}_j}{\Phi^{(T)}_0}, \quad\quad j\in [0,N-1],\\
\zeta \Phi^{(T)}_{n-1} = \sqrt{\mathfrak{a}_{n-1}} \Phi^{(T)}_{n-2}+  \mathfrak{b}_n\Phi^{(T)}_{n-1},\\
\zeta \Psi^{(T)}_j = \sqrt{\mathfrak{a}_{j+1}} \Phi^{(T)}_{j+1} +  \mathfrak{b}_{j+1}\Phi^{(T)}_j+\sqrt{\mathfrak{a}_j} \Psi^{(T)}_{j-1},\quad\quad j\in [N-2],
\end{array}
\end{equation}
where the last equation which determines $w$ implies that the zeroes of
\[
w(\zeta) = \prod_{j=1}^n (\zeta-\kappa^{(T)}_j),
\]
are the eigenvalues of the matrix ${\mathfrak A}^{(sa)}$. In \cite{KV}, the explicit form of the coefficients $c_j(\vec t, n)$, $c_j^{\sigma}(\vec t, n)$ is given solving the gluing conditions under the assumption $\sum_{l=1}^n \kappa^{(T)}_l=0$.

\section{Toda data, Darboux transformations and Toda Baker--Akhiezer functions}\label{sec:TodaDarboux}

In this section, we associate explicitly the Toda Baker--Akhiezer functions (\ref{eq:TodaBA}) to the Toda data $({\mathcal K}, [a])$, with $[a] \in Gr^{TP} (1,n)$ and ${\mathcal K}= \{\kappa^{(T)}_1<\kappa^{(T)}_2 <\cdots <\kappa^{(T)}_n\}$ and we introduce two finite sequences of Darboux transformations, $k\in [n-1]$, 
\begin{equation}\label{eq:Dar_b}
\begin{array}{ll}
\displaystyle D^{(1)} = \partial_{t_1} -{\mathfrak b}_1 (\vec t), &\quad  D^{(k)} = \left(\partial_{t_1} - {\mathfrak b}_k (\vec t) \right) D^{(k-1)},\\
\displaystyle {\hat  D}^{(1)} = \partial_{t_1} -{\mathfrak b}_n (-\vec t), &\quad {\hat  D}^{(k)} = \left(\partial_{t_1} - {\mathfrak b}_k (-\vec t) \right) {\hat  D}^{(k-1)}.
\end{array}
\end{equation}
In the next section, we show that for any given datum $({\mathcal K}, [a])$, the KP vacuum divisor coincides with the Toda divisor, the Darboux transformation $D^{(k)}$ in (\ref{eq:Dar_b}) generates a $T$--hyperelliptic soliton with $k$--compatible divisor, for $k\in [n-1]$, and that such $k$--compatible divisor may be reconstructed from the zeroes at $\vec t=\vec 0$ of the $k$--th entry of the Toda Baker--Akhiezer function. 

Finally, we shall use the second set of Darboux transformations in (\ref{eq:Dar_b}), ${\hat  D}^{(k)}$, when discussing the duality of Grassmann cells under space--time inversion in Section \ref{sec:dual}.

The following proposition contains the necessary information which will be used in this section to re--express the Toda Baker--Akhiezer function in our representation and will be used in the next section in connection with KP $T$--hyperelliptic solitons. 

\begin{proposition}\label{prop:Toda}
Let  $({\mathcal K}, [ a])$, with $[ a] \in Gr^{TP} (1,n)$ and ${\mathcal K}= \{\kappa^{(T)}_1<\kappa^{(T)}_2 <\cdots <\kappa^{(T)}_n\}$, be the initial Toda data, ${\mathfrak A} (\vec t)$ be the corresponding Toda hierarchy solution and ${\mathfrak R} (\zeta; \vec t) = (\zeta {\mathfrak I}_n-{\mathfrak A}(\vec t))^{-1}$.  
Let $\Delta_k (\zeta;\vec t) =\zeta^k - {\hat w}^{(k)}_1 (\vec t)\zeta^{k-1} -\cdots {\hat w}^{(k)}_k(\vec t)$, ${\hat  \Delta}_k (\zeta;\vec t) =\zeta^k -  w^{(k)}_1 (\vec t)\zeta^{k-1} -\cdots  w^{(k)}_k(\vec t)$, $k\in [n-1]$, as in (\ref{eq:delta}).
Define
\begin{equation}\label{eq:Phi}
\displaystyle \Phi (\zeta; \vec t) = \mu_0 (\vec t) \Delta_{n-1} (\zeta;\vec t),\quad\quad 
\displaystyle {\hat \Phi} (\zeta; \vec t) = {\hat \mu}_0 (-\vec t) {\hat \Delta}_{n-1} (\zeta;-\vec t).
\end{equation}
with $\mu_0(\vec t) =\sum_{j=1}^n a_j E_j (\vec t) $, ${\hat \mu}_0(\vec t) =\sum_{j=1}^n {\hat a}_j E_j (-\vec t) $ and $a_j$ related to ${\hat  a}_j$ as in Proposition \ref{prop:aaa}. Then, for any $k\in [n-1]$,
\begin{equation}\label{eq:DarD}
\begin{array}{l}
D^{(k)} = \partial_{t_1}^k - w^{(k)}_1 (\vec t)\partial_{t_1}^{k-1} -\cdots - w^{(k)}_k(\vec t),\\
\displaystyle D^{(k)} \Phi (\zeta;\vec t) = \mu_0 (\vec t)\prod\limits_{j=1}^k {\mathfrak a_j} (\vec t) \Delta_{n-1-k} (\zeta;\vec t) = \frac{\tau_{k+1} (\vec t)}{\tau_{k} (\vec t)} \Delta_{n-k-1} (\zeta;\vec t) = \mu_0 (\vec t){\mathfrak R}^1_{k+1} (\zeta; \vec t) \Delta_n(\zeta),\\
D^{(k)} \mu_0 (\vec t) = D^{(k)} \mu_1 (\vec t) = \cdots D^{(k)} \mu_{k-1} (\vec t)=0,
\end{array}
\end{equation}
\begin{equation}\label{eq:DarDual}
\begin{array}{l}
{\hat  D}^{(k)} = \partial_{t_1}^k - {\hat w}^{(k)}_1 (-\vec t)\partial_{t_1}^{k-1} -\cdots  -{\hat w}^{(k)}_k(-\vec t),\\
\displaystyle {\hat  D}^{(k)} {\hat \Phi} (\zeta;\vec t) = {\hat\mu}_0 (-\vec t)\prod\limits_{j=1}^k {\mathfrak a_{n-j}} (-\vec t)  {\hat \Delta}_{n-1-k} (\zeta;-\vec t) = {\hat\mu}_0 (-\vec t){\mathfrak R}^{n-k}_{n} (\zeta; -\vec t)\Delta_n(\zeta),\\
{\hat  D}^{(k)} {\hat \mu}_0 (-\vec t) = {\hat  D}^{(k)} {\hat \mu}_1 (-\vec t) = \cdots {\hat  D}^{(k)} {\hat\mu}_{k-1} (-\vec t)=0,
\end{array}
\end{equation}
where $\mu_{j} (\vec t)=\partial_{t_1}^j \mu_{0}(\vec t)$, ${\hat \mu}_{j} (\vec t)=(-1)^j\partial_{t_1}^j {\hat \mu}_{0}(\vec t)$, $\tau_0 (\vec t)\equiv 1$, 
$\tau_j (\vec t) = \mbox{{\rm Wr}}_{t_1} ( \mu_0, \partial_{t_1} \mu_0, \cdots ,\partial_{t_1}^{j-1} \mu_0)$, $j\ge 1$.
\end{proposition}

\begin{proof}
The proof is by induction in $k$. 
By definition $\displaystyle {\mathfrak b}_1 (\vec t) = h_1 (\vec t)= \frac{\partial_{t_1} \mu_0 (\vec t)}{\mu_0 (\vec t)}$ and $\displaystyle {\mathfrak b}_n (\vec t) = {\hat h}_1 (\vec t)= -\frac{\partial_{t_1} {\hat \mu}_0 (\vec t)}{{\hat \mu}_0 (\vec t)}$, so that $D^{(1)} \mu_0 (\vec t)=0$, ${\hat D}^{(1)} {\hat \mu}_0 (-\vec t)=0$ and we directly verify  that 
\[
D^{(1)} \Phi (\zeta; \vec t) = \mu_0 (\vec t){\mathfrak a}_1 (\vec t) \Delta_{n-2} (\zeta; \vec t),\quad\quad
{\hat  D}^{(1)} {\hat \Phi} (\zeta; \vec t) = {\hat \mu}_0 (-\vec t){\mathfrak a}_{n-1} (-\vec t) {\hat \Delta}_{n-2} (\zeta; -\vec t).
\]
We easily prove the second identity in (\ref{eq:DarD}) and (\ref{eq:DarDual}) by induction
using (\ref{eq:t1}) and  
\[
\partial_{t_1} \Delta_{n-j} (\zeta;\vec t) = {\mathfrak a}_{j}(\vec t) \Delta_{n-j-1} (\zeta;\vec t), \quad\quad  \partial_{t_1} {\hat \Delta}_{j} (\zeta;\vec t) = -{\mathfrak a}_j (\vec t){\hat \Delta}_{j-1} (\zeta;\vec t)  \quad\quad   j\in [n-1].
\]
Since
\[
D^{(k)} \Phi (\zeta;\vec t) = O (\zeta^{n-k-1}) = \Delta_n (\zeta) \left[ \frac{D^{(k)}\mu_0 (\vec t)}{\zeta}  + \cdots + \frac{D^{(k)}\mu_{k-1} (\vec t)}{\zeta^k}
\right] + O (\zeta^{n-k-1})
\]
$\mu_0, \dots,\partial_{t_1}^{k-1} \mu_0$ are a basis of solutions for the linear differential operator $D^{(k)}$. Finally, 
\[
\mu_{s+k} (\vec t) = w^{(k)}_1 (\vec t) \mu_{s+k-1} (\vec t) + \cdots +w^{(k)}_k  (\vec t)\mu_{s} (\vec t), \quad \quad s\in [k-1],
\]
are the explicit relations between the Hankel coefficients of $H_{\mu} (\vec t)$ defined in Proposition \ref{prop:mu0} and the minors $\Delta_k (\vec t)$, so that then the coefficients of $D^{(k)}$ satisfy (\ref{eq:DarD}). The proof of the remaining identities in (\ref{eq:DarDual}) is similar.
\end{proof}

Let ${\mathfrak C} = \mbox{diag} \left( 1, \sqrt{\mathfrak{a}_{1}}, \sqrt{\mathfrak{a}_{1}\mathfrak{a}_{2}},\dots, \sqrt{\prod_{s=1}^{n-1} \mathfrak{a}_s}\right)$. Then  ${\mathfrak A} (\vec t) \equiv {\mathfrak C}^{-1} {\mathfrak A}^{(sa)}_0 (\vec t){\mathfrak C}$ and  ${\mathfrak C}^{-1}\Psi^{(T)}$ and ${\mathfrak C}\Psi^{(T),\sigma}$ are the Toda Baker--Akhiezer functions in our representation.
It is straightforward to check that (\ref{eq:rec1}) and (\ref{eq:rec2}) are equivalent to the first two recurrences in  (\ref{eq:idD}). Then we may use the Toda data $(\mathcal K, [a])$ to obtain the following equivalent representation of Toda Baker--Akhiezer functions.

\begin{corollary}\label{cor:TodaBA}
Let $\Psi^{(T)} (\zeta; \vec t) $, $\Psi^{(T),\sigma} (\zeta; \vec t) $ the Toda Baker--Akhiezer functions associated to the Toda datum $(\mathcal K, [a])$, $[a]\in Gr^{\mbox{\tiny TP}}(1,n)$ and let $\Phi(\zeta;\vec t)$ as in Proposition \ref{prop:Toda}. Then
\begin{equation}\label{eq:BATodamio}
\Psi^{(T)} (\zeta; \vec t) = e^{\frac{\theta (\zeta, \vec t)}{2}} {\mathfrak C} \left(\begin{array}{c} 
{\hat \Delta}_0 (\zeta; \vec t)\\
{\hat \Delta}_1 (\zeta; \vec t)\\
\vdots\\
{\hat \Delta}_{n-1} (\zeta; \vec t)
\end{array}
\right),\quad\quad 
\Psi^{(T),\sigma} (\zeta; \vec t) = e^{-\frac{\theta (\zeta, \vec t)}{2}} {\mathfrak C}^{-1} \left(\begin{array}{c} 
\frac{\Phi(\zeta;\vec t)}{\Phi(\zeta, \vec 0)}\\
\frac{D^{(1)} \Phi(\zeta;\vec t)}{\Phi(\zeta, \vec 0)}\\
\vdots\\
\frac{D^{(n-1)} \Phi(\zeta;\vec t)}{\Phi(\zeta, \vec 0)}
\end{array}
\right).
\end{equation}
\end{corollary}

\begin{remark}\label{rem:toda2}
In section \ref{sec:TodaKP}, we prove that $\Phi(\zeta;\vec t)$ is the vacuum KP-wavefunction on $\Gamma$ as in (\ref{eq:hyprat}) for the soliton data $(\mathcal K, [ a])$, and show that $D^{(k)}\Phi(\zeta;\vec t)$ is the un--normalized KP--wavefunction associated to the $(n-k,k)$-line $T$--hyperelliptic soliton generated by the Darboux transformation $D^{(k)}$, $k\in [n-1]$. 
It then follows that the Toda divisor ${\mathcal D}^{(T)}$ and the KP vacuum divisor $D^{(0)}=\{ b_1 < \cdots < b_{n-1}\}$ coincide and (\ref{eq:BATodamio}) settle a natural correspondence between the pole divisor of the $(n-k,k)$-line $T$--hyperelliptic soliton and
the zero divisor at $\vec t\equiv 0$ of the $k$--th component of the Toda Baker--Akhiezer function.
 
In section \ref{sec:dual}, we associate (\ref{eq:DarDual}) to the dual Toda hierarchy solution and the dual KP  
line-soliton solutions generated by the space--time inversion and associated to heat hierarchy solution $\sum_{j=1}^n {\hat a}_j E_j (\vec t)$ and use the third recurrence in  (\ref{eq:idD}) for the Toda system to compute the dual KP divisor.
\end{remark}

\section{$T$--hyperelliptic KP solitons and solutions to the Toda hierarchy}\label{sec:TodaKP}

Propositions \ref{prop:mu0}, \ref{prop:Toda} and Corollary \ref{cor:mat} imply a strict connection between $T$--hyperelliptic solitons and solutions to the finite non--periodic Toda hierarchy. Indeed let ${\mathcal K} = \{ \kappa_1 < \cdots < \kappa_n \}$,  $[ a] \in Gr^{\mbox{\tiny TP}}(1,n)$ be given, so that Toda spectrum and the KP phases coincide, that is $\kappa^{(T)}_j \equiv \kappa_j$, $j\in [n]$. Then $\mu_0 (\vec t) = \sum_{l=1}^n a_j E_j (\vec t)$\footnote{We remark that in our setting Toda and KP times coincide, that is $(t_1,t_2,t_3, t_4,\cdots) =(x,y,z,t_4,\cdots)$. If one uses Flasckha original change of variables there is a scaling factor $2^j$ between each j-th Toda and j-th KP time.} generates the $\tau$--functions with $\tau_k (\vec t)=Wr (\mu_0 (\vec t), \dots, \mu_{k-1} (\vec t))$, $k\in [n]$, which are the building blocks of
\begin{enumerate}
\item the Toda hierarchy solution ${\mathfrak a}_k(\vec t) = \frac{\tau_{k-1}(\vec t) \tau_{k+1} (\vec t)}{\tau_k^2(\vec t)}$, $k\in [n-1]$, ${\mathfrak b}_k (\vec t)= \frac{\partial_{t_1} \tau_k (\vec t)}{\tau_k (\vec t)} - \frac{\partial_{t_1} \tau_{k-1} (\vec t)}{\tau_{k-1} (\vec t)}$, $k\in [n]$;
\item the set of KP $T$--hyperelliptic solitons, $u_k(\vec t)= 2\partial_x^2 \log \tau_k (\vec t)$, $k\in [n-1]$.
\end{enumerate}
The identities above suggest a relation between the spectral problems for the finite non periodic Toda system and for KP $T$--hyperelliptic solitons. Indeed, for any data $({\mathcal K}, [a])$, with ${\mathcal K} =\{ \kappa_1<\cdots<\kappa_n\}$ and $[a]\in Gr^{\mbox{\tiny TP}}(1,n)$, we prove:
\begin{enumerate}
\item upon identifying $\Gamma_-$ with the copy of ${\mathbb CP}^1$ containing the Toda divisor, the KP vacuum divisor $\{ b_1,\dots,b_{n-1}\}$ and the Toda divisor $\{b^{(T)}_1,\dots,b^{(T)}_{n-1}\}$ coincide;
\item The Darboux transformations which generate $T$--hyperelliptic solitons with $k$--compatible divisors, coincide with the Darboux transformations recursively defined in (\ref{eq:Dar_b}) for the Toda system;
\item the divisor ${\mathcal D}^{(k)} = \{ \gamma^{(k)}_1,\dots, \gamma^{(k)}_k,\delta^{(k)}_1,\dots, \delta^{(k)}_{n-k-1}\}$ of the $k$--th $T$ hyperelliptic soliton is the zero divisor of the $k$--th component of the Toda Baker--Akhiezer function at times $\vec t=\vec 0$.
\end{enumerate}
We then use the third identity in (\ref{eq:idD}) to recursively compute the divisor of $T$--hyperelliptic solitons as $k$ varies from 1 to $n-1$.

\begin{theorem}\label{theo:KPToda}
Let $\mathcal K=\{ \kappa_1 < \cdots < \kappa_n \}$, $[a] \in Gr^{\mbox{\tiny TP}} (1,n)$. Let $\mu_0 (\vec t) = \sum\limits_{j=1}^n a_j E_j (\vec t)$, with the normalization $\mu_0(\vec 0)=\sum_{j=1}^n a_j =1$, and  $\mu_{s} (\vec t)= \partial_x^{s-1} \mu_0 (\vec t)$, $s\ge 1$. Let $\Delta_j(\zeta;\vec t)$, ${\hat \Delta}_j(\zeta;\vec t)$, $j\in [n-1]$, ${\mathfrak f}_{\mu}(\vec t)$ and $\Phi(\zeta; \vec t)$ be as in Propositions \ref{prop:mu0} and  \ref{prop:Toda}. For any $j\in [n-1]$, let $\tau^{(j)} (\vec t) = Wr \left(\mu_0 (\vec t), \dots, \mu_{j-1} (\vec t) \right)$ and $D^{(j)}$ be the Darboux transformation such that $D^{(j)} \mu_s (\vec t)\equiv 0$, $s\in [j-1]$. 
Then 
\begin{enumerate}
\item
The vacuum KP--wavefunction associated to the soliton data $({\mathcal K}, [a])$ as in (\ref{eq:1M}) satisfies
\begin{equation}\label{eq:TKP}
\Psi^{(-)} (\zeta; \vec t) =\frac{\Phi(\zeta;\vec t) }{\Phi(\zeta; \vec 0)} = \mu_0 (\vec t)\,\frac{\Delta_{n-1}(\zeta;\vec t)}{\Delta_{n-1}(\zeta;\vec 0)}= \mu_0 (\vec t)\,\frac{{\mathfrak f}(\zeta;\vec t)}{{\mathfrak f}(\zeta;\vec 0)} = \mu_0 (\vec t) \prod_{j=1}^{n-1}\frac{(\zeta - b_j(\vec t))}{(\zeta - b_j)},\quad\quad \forall\vec t,\; \zeta\in \Gamma_-;
\end{equation}
\item
For each given $k\in [n-1]$, the normalized KP-wavefunction of the $(n-k,k)$--soliton solution associated to the Darboux transformation $D^{(k)}$ is $\displaystyle {\tilde \Psi}^{(k)} (\zeta; \vec t) = \frac{\Psi^{(k)} (\zeta; \vec t)}{\Psi^{(k)} (\zeta; \vec 0)}$,
where, $\forall \vec t$,
\begin{equation}\label{eq:Psik}
\Psi^{(k)} (\zeta; \vec t) \equiv D^{(k)} \Psi (\zeta; \vec t) = \left\{ \begin{array}{ll}
{\hat \Delta}_k (\zeta; \vec t) e^{\theta(\zeta; \vec t)}, &\quad \zeta \in \Gamma_+,\\
\displaystyle\frac{D^{(k)} \Phi(\zeta; \vec t)}{\Phi(\zeta; \vec 0)} =\frac{\tau_{k+1}(\vec t)}{\tau_{k}(\vec t)} \frac{\Delta_{n-k-1} (\zeta; \vec t)}{\Delta_{n-1} (\zeta; \vec 0)}, &\quad \zeta \in \Gamma_-,
\end{array}\right.
\end{equation}
\item For any $k\in [n-1]$,  the divisor ${\mathcal D}^{(k)} (\vec t) = \{ \gamma^{(k)}_1(\vec t),\dots, \gamma_k^{(k)} (\vec t), \delta^{(k)}_1,\dots, \delta^{(k)}_{n-k-1}(\vec t)\}$ of ${\tilde \Psi}^{(k)} (\zeta; \vec t)$ satisfies 
\begin{equation}\label{eq:divToda}
\prod\limits_{j=1}^k (\zeta -\gamma^{(k)}_j (\vec t)) = {\hat\Delta}_{k} (\zeta; \vec t), \quad\quad
\prod\limits_{j=1}^{n-k-1} (\zeta -\delta^{(k)}_j (\vec t)) =\Delta_{n-k-1} (\zeta; \vec t).
\end{equation}
\end{enumerate}
\end{theorem}

\begin{proof}
Inserting $h_j (\vec t)= \frac{\mu_j (\vec t)}{\mu_j (\vec 0)}$ in (\ref{eq:gen}), and ${\mathfrak M}_l (\vec t)= \frac{a_l E_l(\vec t)}{\sum\limits_{j=1}^n a_j E_j(\vec t)}$ in (\ref{eq:Phi}), we get
\[
\Phi(\zeta;\vec t) = \mu_0 (\vec t) \left( \sum\limits_{j=1}^n {\mathfrak M}_j (\vec t)\prod\limits_{s\not = j} (\zeta-\kappa_s)\right) =\sum\limits_{j=1}^n a_j E_j (\vec t)\prod\limits_{s\not = j} (\zeta-\kappa_s).
\]
In particular, with our normalization of $a$ and by the definition of the vacuum divisor,
\[
\Phi(\zeta;\vec 0)= \Delta_{n-1} (\zeta; \vec 0)  = \prod\limits_{j=1}^n (\zeta-\kappa_j) \left( \sum\limits_{i=1}^n \frac{a_i}{\zeta-\kappa_i }\right) =\prod\limits_{r=1}^{n-1} (\zeta - b_r).
\]
Then, using Proposition \ref{prop:Toda}, the other assertions follow.
\end{proof}

Notice that in (\ref{eq:psiN}), $\displaystyle{\tilde A}^{(k)} (\vec t) = \frac{\tau^{(k+1)}(\vec t)  }{\tau^{(k)}(\vec t)}>0$, 
$\forall \vec t.$
 Moreover, for any fixed $k\in [n-1]$ and for all $\vec t$, the divisor $\mathcal D^{(k)}$ is $k$--compatible  according to the counting rule. 

\begin{corollary}\label{cor:TodaKP}
Let the Toda and KP soliton datum be $\mathcal K=\{ \kappa_1 < \cdots < \kappa_n \}$, $[a] \in Gr^{\mbox{\tiny TP}} (1,n)$. Let ${\tilde \Psi}^{(k)} (\zeta; \vec t)$ be the KP wavefunction for the $T$--hyperelliptic $(n-k,k)$--soliton as in Theorem \ref{theo:KPToda} and let $\Psi^{(T)} (\zeta;\vec t)$, $\Psi^{(T),\sigma} (\zeta;\vec t)$ be the Toda Baker--Akhiezer functions associated to such datum. Let $\{ b_1<\cdots <b_{n-1}\}$, $\{ b^{(T)}_1<\cdots <b^{(T)}_{n-1}\}$ respectively be the KP vacuum divisor and the Toda divisor for such data, ${\mathcal D}^{(k)} = \{ \gamma^{(k)}_1,\dots, \gamma_k^{(k)} , \delta^{(k)}_1,\dots, \delta^{(k)}_{n-k-1}\}$ be the pole divisor of ${\tilde \Psi}^{(k)} (\zeta; \vec t)$, $k\in [n-1]$. Then
\begin{equation}\label{eq:divTKP}
\begin{array}{l}
b_j = b^{(T)}_j, \quad j\in [n-1];\\
\Psi^{(T)}_k (\gamma^{(k)}_l, \vec 0) =0, \quad \forall l\in [k], \quad\quad 
\Psi^{(T),\sigma}_k (\delta^{(k)}_s, \vec 0) =0, \quad \forall s\in [n-k-1] .
\end{array}
\end{equation}
\end{corollary}

The corollary easily follows comparing (\ref{eq:BATodamio}) and (\ref{eq:Psik}) and the definition of the vacuum KP and of the Toda divisors in Sections \ref{sec:vacuumKP} and \ref{sec:TodaDarboux}.

\smallskip

In view of (\ref{eq:divToda}), we interpret the third identity in (\ref{eq:idD}) as a recursive relation to compute the divisor of $T$--hyperelliptic solitons as $k$ varies from 1 to $n-1$.

\begin{corollary}\label{cor:formulas}
Under the hypotheses of Theorem \ref{theo:KPToda}, for any fixed $k\in [n-1]$ the divisor of ${\tilde \Psi}^{(k)}(\zeta;\vec t)$ may be computed from the divisor of ${\tilde \Psi}^{(k-1)} (\zeta;\vec t)$, for all $\vec t$, using  (\ref{eq:idD})
\begin{equation}\label{eq:divpredual}
\prod\limits_{j=1}^n (\zeta -\kappa_j) = \prod\limits_{l=1}^{k} (\zeta - \gamma^{(k)}_l(\vec t)) \prod\limits_{s=1}^{n-k} (\zeta - \delta^{(k-1)}_s(\vec t))
-{\mathfrak a}_k (\vec t) \prod\limits_{i=1}^{k-1} (\zeta - \gamma^{(k-1)}_i(\vec t))\prod\limits_{r=1}^{n-k-1} (\zeta - \delta^{(k)}_r(\vec t)),
\end{equation}
where ${\mathfrak a}_k (\vec t)$ are as in (\ref{eq:abtau}).
Moreover, for any given $\vec t$ and $k\in [n-1]$, the divisor ${\mathcal D}^{(k)} (\vec t)$ may be computed from the vacuum divisor $(b_1(\vec t),\dots, b_{n-1}(\vec t))$ solving the system of equation
\begin{equation}\label{eq:divrel}
\prod\limits_{r=1}^{n-1} (\kappa_j -b_r(\vec t)) \prod\limits_{i=1}^{k} (\kappa_j - \gamma^{(k)}_i(\vec t)) - \left(\prod\limits_{s=1}^k {\mathfrak a}_s (\vec t) \right) \prod\limits_{r=1}^{n-k-1} (\kappa_j - \delta^{(k)}_r(\vec t)) = 0, \quad\quad j\in [n].
\end{equation}
\end{corollary}

\begin{proof}
If the divisors ${\mathcal D}^{(k)} (\vec t)$ are all generic for a given $\vec t$, {\sl i.e.} $\Delta_{k}(\kappa_j, \vec t),{\hat \Delta}_{k}(\kappa_j, \vec t)\not =0$, for all $k\in[n-1]$ and $j\in [n]$, the proof of
(\ref{eq:divrel}) is by induction in $k$ using the third identity in (\ref{eq:idD}).

Suppose now, that, for a given $\vec t$  the divisors are generic for $k\in [l-1]$ and ${\mathcal D}^{(l)} (\vec t)$ is not generic and contains the point $\kappa_{\hat  \jmath}$. Then (\ref{eq:divrel}) hold for $k\in [l]$.  By the intertwining properties of the zeros of the polynomials $\Delta_k$ and ${\hat\Delta}_k$, ${\hat \Delta}_{l} (\kappa_{\hat\jmath}, \vec t)= \Delta_{n-l-1} (\kappa_{\hat\jmath}, \vec t)=0$ implies that ${\hat \Delta}_{l+1} (\kappa_{\hat\jmath}, \vec t), {\hat \Delta}_{l-1} (\kappa_{\hat\jmath}, \vec t), \Delta_{n-l} (\kappa_{\hat \jmath}, \vec t),\Delta_{n-l-2} (\kappa_{\hat\jmath}, \vec t)\not =0$.
Let ${\hat \Delta}_{l} (\zeta; \vec t) = (\zeta -\kappa_{\hat\jmath}) {\hat \Delta}^{\prime}_l(\zeta; \vec t)$ and $\Delta_{n-l-1} (\zeta; \vec t) = (\zeta -\kappa_{\hat\jmath}) \Delta^{\prime}_{n
-l-1}(\zeta; \vec t)$. For $k=l+1$ and $j\not ={\hat\jmath}$, identities (\ref{eq:divrel}) still hold, while if $j={\hat  \jmath}$, using (\ref{eq:idD}), we get
\[
\begin{array}{l}
\prod\limits_{s=1}^{l-1} {\mathfrak a}_s (\vec t) \prod\limits_{r\not ={\hat\jmath}} (\kappa_{\hat\jmath}-\kappa_r) \left( {\hat\Delta}_{l+1} (\kappa_{\hat\jmath}, \vec t) +{\mathfrak a}_l (\vec t) {\hat \Delta}_{l-1} (\kappa_{\hat\jmath}, \vec t) \right) 
\\
\quad\quad= {\hat\Delta}^{\prime}_l(\kappa_{\hat\jmath}, \vec t){\hat \Delta}_{l-1} (\kappa_{\hat\jmath}, \vec t) \left( \Delta_{n-1}(\kappa_{\hat\jmath}, \vec t) {\hat \Delta}_{l+1} (\kappa_{\hat\jmath}, \vec t) - \big(\prod\limits_{s=1}^{l+1} {\mathfrak a}_l (\vec t) \big) \Delta_{n-l-2} (\kappa_{\hat\jmath} ,\vec t)
\right).
\end{array}
\]
Since ${\hat \Delta}_{l+1} (\kappa_{\hat\jmath}, \vec t) +{\mathfrak a}_l (\vec t) {\hat \Delta}_{l-1} (\kappa_{\hat\jmath}, \vec t)=0$, 
and ${\hat\Delta}^{\prime}_l(\kappa_{\hat\jmath}, \vec t),{\hat \Delta}_{l-1} (\kappa_{\hat\jmath}, \vec t)\not =0$, we conclude that (\ref{eq:divrel}) holds also for $j={\hat\jmath}$.
\end{proof}

\begin{remark}
For any fixed $k\in [n-1]$, the conditions
$D^{(k)} \Psi_+ (\kappa_j, \vec t) = D^{(k)} \Psi_- (\kappa_j, \vec t)$, for all $j\in [n]$,
are equivalent to (\ref{eq:divrel}), which may be rewritten as
\[
\Delta_{n-1} (\kappa_j, \vec t) {\hat \Delta}_{k} (\kappa_j, \vec t) -\left(\prod\limits_{s=1}^k {\mathfrak a}_s (\vec t) \right)\Delta_{n-k+1} (\kappa_j, \vec t)=0, \quad \forall j\in [n].
\]
Moreover, for any fixed $l\in [n]$,
\begin{equation}\label{eq:tau1}
\tau_1 (\vec t) = E_l(\vec t) \, \prod_{r=1}^{n-1} \frac{\kappa_l - b_r(\vec 0)}{\kappa_l - b_r(\vec t)}, \quad\quad \forall \vec t.
\end{equation}
\end{remark}

\section{Reconstruction of soliton data and Toda solutions from $k$--compatible divisors}\label{sec:invKP}

Let ${\mathcal K}$ be fixed. The relations found in the previous section, allow to reconstruct the soliton data associated to a $k$--compatible divisor and to express the solution of the Toda hierarchy in function of the Toda/KP zero--divisor dynamics.
Indeed, for any given $k\in [n-1]$, equations (\ref{eq:divrel}) allow to solve both the direct and the inverse problem. If we assign the 
soliton datum $[a]\in Gr^{\mbox{\tiny TP}} (1,n)$, we first compute the vacuum divisor
$(b_1,\dots,b_{n-1})$ using the identity 
\[
\prod_{r=1}^{n-1} (\zeta -b_r) = \prod_{j=1}^n (\zeta-\kappa_j) \left(\sum_{s=1}^n \frac{a_s}{\sum_{l=1}^n a_l } (\zeta-\kappa_s)^{-1}\right)
\]
and then the $k$--compatible divisor ${\mathcal D}^{(k)} = ( \gamma^{(k)}_1, \dots, \gamma^{(k)}_k, \delta^{(k)}_{1}, \dots, \delta^{(k)}_{n-k-1})$ from (\ref{eq:divrel}).

Viceversa, if we assign a $k$--compatible divisor ${\mathcal D}^{(k)}$ on ${\Gamma}$, we may reconstruct the soliton datum $[a]\in Gr^{\mbox{\tiny TP}} (1,n)$, by first computing the vacuum divisor from  (\ref{eq:divrel}) and then taking
$a_j = \frac{\prod_{r=1}^{n-1} (\kappa_j -b_r) }{\prod_{s\not =j}^{n} (\kappa_j -\kappa_s)}$.
Indeed we have the following

\begin{theorem}\label{theo:divfor}
Let ${\mathcal K} =\{ \kappa_1 < \cdots < \kappa_n\}$, $({\Gamma}, P_+, \zeta)$ as in (\ref{eq:hyprat}) and let ${\mathcal D}^{(k)} = ( \gamma^{(k)}_1, \dots, \gamma^{(k)}_k, \delta^{(k)}_{1}$, $\dots , \delta^{(k)}_{n-k-1})$
be a $k$--compatible divisor on ${\Gamma}\backslash\{P_+\}$. If ${\mathcal D}^{(k)}$ is generic, then the un-normalized  soliton datum is 
\begin{equation}\label{eq:inv}
 a_j = \frac{\prod_{s=1}^{n-1} (\kappa_j -\delta^{(k)}_s)}{\prod_{r=1}^{k} (\kappa_j -\gamma^{(k)}_r)\prod_{l\not =j}^{n} (\kappa_j -\kappa_l)}, \quad\quad j\in [n].
\end{equation}
In the non--generic case, (\ref{eq:inv}) holds for $j$ if $\kappa_j \not \in {\mathcal D}^{(k)}$. For any  ${\hat  \jmath}$ such that $\kappa_{\hat \jmath}\in {\mathcal D}^{(k)}$, let
$\kappa_{\hat \jmath} =\gamma^{(k)}_{\hat r} = \delta^{(k)}_{\hat  s}$. Then (\ref{eq:inv}) is substituted  by
\begin{equation}\label{eq:angen}
 a_{\hat \jmath} = -\frac{\prod_{s\not = {\hat s}}^{n-1} (\kappa_{\hat \jmath} -\delta^{(k)}_s)}{ \prod_{r\not = {\hat   r}}^{k} (\kappa_{\hat \jmath} -\gamma^{(k)}_r)\prod_{l\not ={\hat \jmath}}^{n} (\kappa_{\hat \jmath} -\kappa_l)}.
\end{equation}
\end{theorem}

\begin{proof} 
If the divisor is generic plugging (\ref{eq:divrel}) into $a_j = \frac{\prod_{r=1}^{n-1} (\kappa_j -b_r) }{\prod_{s\not =j}^{n} (\kappa_j -\kappa_s)}$, we get (\ref{eq:inv}) up to the constant normalization factor $\frac{\tau_{k+1}(\vec 0)}{\tau_{k}(\vec 0)}$. The non generic divisor containing  $\kappa_{\hat \jmath} =\gamma^{(k)}_{\hat r} = \delta^{(k)}_{\hat s}$ is the limit of the
generic divisor ${\mathcal D}^{(k)}_{\epsilon} = \left( {\mathcal D}^{(k)} \backslash \{  \gamma^{(k)}_{\hat r}, \delta^{(k)}_{\hat  s}   \} ) \right) \cup \{  \gamma^{(k)}_{\hat  r} +\epsilon, \delta^{(k)}_{\hat  s}  -\epsilon \}$, when $\epsilon\to 0$.
So $a_{\hat \jmath}$ satisfies (\ref{eq:angen}).
\end{proof}

Let ${\mathcal K} = \{ \kappa_1< \cdots < \kappa_n\}$ be fixed and let ${\mathcal D}^{(k)}$ be a $k$--compatible divisor on $\Gamma$. Then, using Theorem \ref{theo:divfor} we reconstruct the initial data of a solution to the Toda hierarchy (\ref{eq:Todafj}) and, using Corollary \ref{cor:formulas}, we may express the solution to the Toda hierarchy in function of the system of compatible divisors associated to such soliton data.

\begin{proposition}
Let $({\mathcal K}, [a])$ be soliton data with $[a]\in Gr^{\mbox{\tiny TP}}(1,n)$, $\sum_{j=1}^n a_j =1$. Let ${\mathcal D}^{(k)} (\vec t) = \{ \gamma^{(k)}_1 (\vec t), \dots, \gamma^{(k)}_k (\vec t), \delta^{(k)}_1(\vec t),\dots, \delta^{(k)}_{n-k-1} (\vec t)	\} $, $k\in [n-1]$, be the set of $k$--compatible divisors associated to such soliton data, with $k\in [n-1]$,  and 
let ${\mathcal B} (\vec t) = \{ b_1 (\vec t) < \cdots < b_{n-1} (\vec t) \}$ be the zero divisor of $\Psi(\zeta; \vec t)$ in (\ref{eq:TKP}). Let $j\in [n]$ be fixed. Then the solution to (\ref{eq:Todafj}) with initial datum $({\mathcal K}, [a])$ is, for any $\vec t$,
\begin{equation}\label{eq:abdiv}
\begin{array}{ll}
\displaystyle {\mathfrak a}_1 (\vec t) = \frac{ (\kappa_j - \gamma^{(1)}_1(\vec t)) \; \prod_{r=1}^{n-1} (\kappa_j - b_r(\vec t))}{\prod_{s=1}^{n-1} (\kappa_j - \delta^{(1)}_s(\vec t))},  \quad
 {\mathfrak b}_1 (\vec t)= \kappa_j+\sum_{r=1}^{n-1} \frac{\partial_x b_r (\vec t)}{\kappa_j - b_r (\vec t)},&\\
\displaystyle {\mathfrak a}_k (\vec t) = \frac{\prod_{i=1}^k (\kappa_j - \gamma^{(k)}_i(\vec t)) \; \prod_{l=1}^{n-k} (\kappa_j - \delta^{(k-1)}_l(\vec t))}{\prod_{r=1}^{k-1} (\kappa_j - \gamma^{(k-1)}_r(\vec t)) \; \prod_{s=1}^{n-k-1} (\kappa_j - \delta^{(k)}_s(\vec t))}, &\;\; k=2,\dots,n-1;\\
\displaystyle {\mathfrak b}_k (\vec t)= \kappa_j+\sum_{l=1}^{n-k} \frac{\partial_x \delta^{(k-1)}_l (\vec t)}{\kappa_j - \delta^{(k-1)}_l (\vec t)}-\sum_{i=1}^{k-1} \frac{\partial_x \gamma^{(k-1)}_i (\vec t)}{\kappa_j - \gamma^{(k-1)}_i (\vec t)} ,
&\;\; k=2,\dots,n,
\end{array}
\end{equation}
where, if for some $\vec t$ and ${\hat k}\in [n-1]$, $\gamma^{({\hat  k})}_{\hat \imath} (\vec t) =\delta^{({\hat k})}_{\hat s} (\vec t)=\kappa_{\hat \jmath}$, we substitute ${\mathfrak a}_{\hat k} (\vec t),{\mathfrak a}_{\hat k+1} (\vec t),{\mathfrak b}_{\hat  k+1} (\vec t) $in (\ref{eq:abdiv}) with
\[
\begin{array}{l}
\displaystyle {\mathfrak a}_{\hat k} (\vec t) = -\frac{\prod_{i\not ={\hat  \imath} }^{\hat  k} (\kappa_{\hat \jmath} - \gamma^{({\hat k})}_i(\vec t)) \; \prod_{l=1}^{n-{\hat k}} (\kappa_{\hat \jmath} - \delta^{({\hat k}-1)}_l(\vec t))}{\prod_{r=1}^{\hat  k-1} (\kappa_{\hat \jmath} - \gamma^{({\hat k}-1)}_r(\vec t)) \; \prod_{s\not = {\hat s}}^{n-{\hat k}-1} (\kappa_{\hat  \jmath} - \delta^{({\hat  k})}_s(\vec t))},\\
\displaystyle {\mathfrak a}_{{\hat k}+1} (\vec t) = -\frac{\prod_{r =1 }^{{\hat  k}+1} (\kappa_{\hat \jmath} - \gamma^{({\hat k}+1)}_r(\vec t)) \; \prod_{s\not = {\hat s}}^{n-{\hat k}-1} (\kappa_{\hat \jmath} - \delta^{({\hat k})}_s(\vec t))}{\prod_{i\not={\hat \imath}}^{{\hat k}} (\kappa_{\hat \jmath} - \gamma^{({\hat k})}_i(\vec t)) \; \prod_{l = 1}^{n-{\hat k}-2} (\kappa_{\hat \jmath} - \delta^{({\hat k}+1)}_l(\vec t))},\\
\displaystyle {\mathfrak b}_{{\hat k}+1} (\vec t) =\kappa_{\hat \jmath}+ \sum_{s\not = {\hat s}}^{n-{\hat k}-1} \frac{\partial_x \delta^{({\hat k})}_s (\vec t)}{\kappa_{\hat \jmath} - \delta^{({\hat k})}_s (\vec t)}-\sum_{i\not ={\hat \imath}}^{{\hat k}} \frac{\partial_x \gamma^{({\hat k})}_i (\vec t)}{\kappa_{\hat \jmath} - \gamma^{({\hat k})}_i (\vec t)}.
\end{array}
\]  
\end{proposition}

\begin{proof}
(\ref{eq:abdiv}) easily follow using (\ref{eq:abtau}), (\ref{eq:divrel}) and (\ref{eq:tau1}), since
\[
\prod_{s=1}^k {\mathfrak a}_s(\vec t) =\frac{\tau_{k+1} (\vec t)}{\tau_{k} (\vec t)\tau_{1} (\vec t)} = \frac{\prod_{i=1}^k (\kappa_j - \gamma^{(k)}_i(\vec t)) \; \prod_{r=1}^{n-1} (\kappa_j - b_r(\vec t))}{\prod_{s=1}^{n-k-1} (\kappa_j - \delta^{(k)}_s(\vec t))}.
\]
The case of the non--generic divisor as usual follows from the limit of the generic case. 
\end{proof}

\section{Duality of Grassmann cells, space--time inversion and divisors}\label{sec:dual}

Space--time inversion in KP soliton solutions induces a duality transformation of $Gr(k,n)$ to $Gr(n-k,n)$.
In this section, we investigate the effect of such transformation on the algebraic geometric description of $T$--hyperelliptic soliton solutions and we show that (\ref{eq:divpredual}) lead to a natural characterization of the dual divisor in $Gr(n-k,n)$  through hyperelliptic involution. For the Toda system such duality corresponds to the composition of the reflection w.r.t the antidiagonal defined in (\ref{eq:Astar}) with Toda times inversion, {\sl i.e.} to pass from Toda solutions for $\mathfrak A(\vec t)$ to Toda solutions for $\mathfrak A^{*} (-\vec t)$.

The KP equation is invariant under the space--time inversion $\vec t \to -\vec t$. In particular, if $u_{[A]}(\vec t)$\footnote{In the following the subscripts $[x]$, respectively $x$, mean that the value of the expression depends on the point in the Grassmannian, respectively on the representative matrix of the point in the Grassmannian.} is the KP solution for the soliton data $(\mathcal K, [A])$ with in $[A]\in Gr^{\mbox{\tiny TNN}} (k,n)$ then there exists  $[{\hat A}]\in Gr^{\mbox{\tiny TNN}} (n-k,n)$ such that $u_{[{\hat A}]} (\vec t) \equiv u_{[ A]}(-\vec t)$ is the solution associated to the dual soliton data $(\mathcal K, [{\hat A]})$.
The combinatorial interpretation of this transformation has been given in \cite{CK1} (see also \cite{Z}). 

Since the phases are invariant with respect to the space--time inversion, the curve $\Gamma$ is preserved. If $(\mathcal K, [A])$ are the data of a $T$--hyperelliptic soliton, also the dual data $(\mathcal K, [{\hat A}])$ are associated to a $T$--hyperelliptic soliton solution.
In this section we investigate the relations between the divisors of $T$--hyperelliptic dual solitons with $A$ and ${\hat A}$ as in (\ref{eq:Bmat}), 
\[
A^i_j = a_j \kappa_j^{i-1}, 
\quad\quad
{\hat A}^l_j = {\hat a}_j \kappa_j^{l-1}, \quad\quad \forall i\in [k], \; l\in [n-k], \; j \in [n],
\]
$[a_1,\dots, a_n],[{\hat a}_1,\dots, {\hat a}_n]  \in  Gr^{\mbox{\tiny TP}} (1,n)$, for some $k\in [n-1]$. Let us denote, respectively, the heat hierarchy solutions 
\[
\mu_{a, i} (\vec t) = \sum\limits_{j=1}^n a_j \kappa_j^{i} E_j (\vec t), \quad\quad 
\quad
\quad
\mu_{{\hat a}, i} (\vec t) = \sum\limits_{j=1}^n {\hat a}_j\kappa_j^{i} E_j (\vec t), \quad\quad i\ge 0,
\]
the $\tau$--functions 
\begin{equation}\label{eq:tauac}
\begin{array}{l}
\displaystyle
\tau^{(k)}_{a} (\vec t) = {\rm Wr } (\mu_{a,0}, \dots, \mu_{a,k-1} ) = \sum\limits_{1\le i_1 < \cdots < i_k\le n} \left(\prod\limits_{s=1}^k a_{i_s} E_{i_s} (\vec t) \right)\prod\limits_{1<r<s<k} (\kappa_{i_s} -\kappa_{i_r} )^2,
\\
\displaystyle
\tau^{(n-k)}_{{\hat a}} (\vec t) = {\rm Wr } (\mu_{{\hat a},0}, \dots, \mu_{{\hat a},n-k-1}) = \sum\limits_{1\le l_1 < \cdots < l_{n-k}\le n} \left(\prod\limits_{s=1}^{n-k} {\hat a}_{l_s} E_{l_s} (\vec t) \right)\prod\limits_{1<r<s<n-k} (\kappa_{l_s} -\kappa_{l_r} )^2,
\end{array}
\end{equation}
and the Darboux transformations
\begin{equation}\label{eq:dabac}
D^{(k)}_{[a]} = \partial_x^k - w^{(k)}_{[a],1}(\vec t)\partial_x^{k-1} -\cdots w^{(k)}_{[a],k} (	\vec t),\quad\quad
D^{(n-k)}_{[{\hat a}]} = \partial_x^{n-k} - w^{(n-k)}_{[{\hat a}],1}(\vec t)\partial_x^{k-1} -\cdots w^{(n-k)}_{[{\hat a}],n-k} (	\vec t),
\end{equation}
where 
$D^{(k)}_{[a]} \mu_{a,i} (\vec t) \equiv 0$, $\;\;D^{(n-k)}_{[{\hat a}]} \mu_{{\hat a},l} (\vec t) \equiv 0$, for all $i\in [0,k-1]$, $l\in [0,n-k-1]$, $\vec t$. The 
KP solutions 
\begin{equation}\label{eq:uac}
u_{[a], k} (\vec t) = 2\partial_x^2 \log \tau^{(k)}_{a} (\vec t),\quad\quad u_{[{\hat a}],n-k} (\vec t) = 2\partial_x^2 \log \tau^{(n-k)}_{{\hat a}} (\vec t)
\end{equation}
are related by the space--time transformation $u_{[{\hat a}]} (\vec t) =u_{[a]} (-\vec t) $ if and only if there exists a constant $C_k(a, {\hat a})>0$ such that
\[
\tau^{(n-k)}_{{\hat a}} (\vec t) = C_k(a, {\hat a})\,\tau^{(k)}_{ a} (-\vec t)\, \prod\limits_{j=1}^n E_j (\vec t), \quad\quad \forall \vec t, \quad k\in [0,n].
\]
To characterize the duality condition, it is convenient to use a different set of coordinates. Let $[{\alpha}], [\hat\alpha] \in Gr^{\mbox{\tiny TP}} (1,n)$ be related to $[a], [{\hat a}] \in Gr^{\mbox{\tiny TP}} (1,n)$ by 
\begin{equation}\label{eq:ac}
a_j = \frac{(-1)^{n-j}\alpha_j}{\prod\limits_{m\not = j}^n (\kappa_j-\kappa_m)}, \quad\quad
{\hat a}_j = \frac{(-1)^{n-j}\hat\alpha_j}{\prod\limits_{m\not = j}^n (\kappa_j-\kappa_m)},\quad\quad j\in [n].
\end{equation}
Then the duality condition is equivalent to the following relations between $[\alpha]$ and $[\hat\alpha]$.

\begin{lemma}\label{lemma:ac}
Let $k\in [n-1]$ be fixed and ${\mathcal K} = \{ \kappa_1 < \cdots <\kappa_n \}$.
Let $[a], [{\hat a}] ,[\alpha], [\hat\alpha]\in Gr^{\mbox{\tiny TP}} (1,n)$,  $\tau^{(k)}_{a}(\vec t)$, $\tau^{(n-k)}_{{\hat a}}(\vec t)$, $u_{[a], k} (\vec t)$, $u_{[{\hat a}], n-k} (\vec t)$ as in (\ref{eq:ac}),(\ref{eq:tauac}) and (\ref{eq:uac}).
Then the following statements are equivalent
\begin{enumerate}
\item\label{it:eq1} $u_{[{\hat a}],n-k} (\vec t) = u_{[a], k} (-\vec t)$ for all $\vec t$;
\item\label{it:eq2}
it is possible to normalize the representative vector of $[{\hat a}]$ so that
$ \tau^{(n-k)}_{{\hat a}} (\vec t) = \tau^{(k)}_{a} (-\vec t) \prod\limits_{j=1}^n \frac{E_j (\vec t)}{\alpha_j}$, $\forall \vec t$;
\item\label{it:eq3} $[\hat\alpha_1, \dots, \hat\alpha_n] = \left[ \alpha_1^{-1}, \dots, \alpha_n^{-1}\right]$.
\end{enumerate}
\end{lemma}  

\begin{proof}

(\ref{it:eq1}) and (\ref{it:eq2}) are equivalent since soliton solutions depend just on the point in $Gr^{\mbox{\tiny TP}} (1,n)$ and $\tau$--functions are defined up to multiplicative constants.
Condition $(\ref{it:eq2})$ on the $\tau$--functions is equivalent to  
$\prod_{i\in I} a_{i} \prod\limits_{r,s\in I, r<s} (\kappa_{s} -\kappa_{r})^2 = \prod_{j\in J} {\hat a}_{j} \prod_{l,k\in J, k<l} (\kappa_{l} -\kappa_{k})^2$,
for all $I\in \left( \mycom{[n]}{k} \right)$, $J=[n]-I$. Inserting (\ref{eq:ac}), in the above identities it is straightforward to show that
$\prod\limits_{i\in I} \alpha_i = \prod_{j\in [n]	\backslash I} {\hat \alpha}_j$, for all $I\in \left(\mycom{[n]}{k}\right)$, which is is equivalent to (\ref{it:eq3}). 
\end{proof}

\begin{corollary}\label{cor:ac}
Let ${\mathcal K}$ be given.
Let $[a], [{\hat a}] \in Gr^{\mbox{\tiny TP}} (1,n)$, $\tau^{(k)}_{a}(\vec t)$, $\tau^{(k)}_{{\hat a}}(\vec t)$, $u_{[a], k} (\vec t)$ , $u_{[{\hat a}],n- k} (\vec t)$, $\forall k\in [0,n]$, $\forall \vec t$, as in (\ref{eq:ac}), (\ref{eq:tauac}) and (\ref{eq:uac}).
Then the following statements are equivalent
\begin{enumerate}
\item\label{it:eq4} There exists ${\bar k} \in [n-1]$ such that, for any $\vec t$,  $u_{[{\hat a}],n- {\bar k}} (\vec t) = u_{[a], {\bar k}} (-\vec t)$;
\item\label{it:eq5} For any $k\in [1,n-1]$ and for any $\vec t$,  $u_{[{\hat a}],n- k} (\vec t) = u_{[a], k} (-\vec t)$;
\item\label{it:eq7} $[\hat\alpha_1,\dots,\hat\alpha_n]= [\alpha_1^{-1},\dots,\alpha_n^{-1}] $.
\end{enumerate}
\end{corollary}

Condition $(\ref{it:eq7})$ in Corollary \ref{cor:ac}  is equivalent to  (\ref{eq:adual}), that is the duality of Grassmann cells induces dual Toda hierarchy solutions and dual KP soliton solutions which are naturally linked.

\subsection{Dual Toda flows} 

The space--time inversion settles a duality condition in the space of KP line soliton solutions which is also a duality condition between Toda flows. In this subsection, we list the relevant relations between such dual Toda hierarchies using Proposition \ref{prop:Toda} and then in the next subsection we use them to determine the relations among the divisors associated to dual soliton data.

To the initial data $({\mathcal K}, [a])$ and $({\mathcal K},  [{\hat a}])$, with $[a]$, $[{\hat a}]$, $[\alpha]$, $[\hat \alpha] \in Gr^{\mbox{\tiny TP}} (1,n)$ satisfying condition (\ref{it:eq3}) in Lemma \ref{lemma:ac} and (\ref{eq:ac}), we
associate dual Toda hierarchies, $j\ge 1$,
\begin{equation}\label{eq:Todafdual}
\displaystyle\frac{d{\mathfrak A}_{[a]}}{dt_j}(\vec t) = [ {\mathfrak B}_{[a], j}(\vec t), {\mathfrak A}_{[ a]}(\vec t)], \quad\quad \frac{d{\mathfrak A}_{[{\hat a}]}}{dt_j} (\vec t)= [ {\mathfrak B}_{[{\hat a}], j}(\vec t), {\mathfrak A}_{[{\hat a}]}(\vec t)],
\end{equation}
with 
{\small
\begin{equation}\label{eq:Adual}
{\mathfrak A}_{[a]} (\vec t)=\left(\begin{array}{cccc}
\mathfrak{b}_{[a],1}  (\vec t)& \mathfrak{a}_{[a],1} (\vec t) & 0 & \cdots  \\
1 & \mathfrak{b}_{[a],2}  (\vec t)& \mathfrak{a}_{[a],2}  (\vec t)&\ddots  \\
0 & \ddots & \ddots                    & \ddots                  \\
0 & \cdots    & 1   & \mathfrak{b}_{[a],n} (\vec t)
\end{array}\right),\quad
{\mathfrak A}_{[{\hat a}]}(\vec t)=\left(\begin{array}{cccc}
\mathfrak{b}_{[{\hat a}],1}  (\vec t)& \mathfrak{a}_{[{\hat a}],1}  (\vec t)& 0 & \cdots  \\
1 & \mathfrak{b}_{[{\hat a}],2} (\vec t) & \mathfrak{a}_{[{\hat a}],2}  (\vec t)&\ddots \\
0 & \ddots &\ddots & \ddots \\
0 & \cdots & 1 & \mathfrak{b}_{[{\hat a}],n} (\vec t)
\end{array}\right),
\end{equation}
}
${\mathfrak B}_{[ a], j}(\vec t) = \left( {\mathfrak A}^j_{[a]}(\vec t) \right)_+$, ${\mathfrak B}_{[{\hat a}], j}(\vec t) = \left( {\mathfrak A}_{[{\hat a}]}^j (\vec t)\right)_+$, where $(\cdot)_+$ denotes the strictly upper triangular part of the matrix, via the generating functions
\begin{equation}\label{eq:gendual}
\begin{array}{l}
\displaystyle {\mathfrak f}_{[a]} (\zeta;\vec t) \equiv \prec e_1 , (\zeta {\mathfrak I}_n - {\mathfrak A}_{[a]}(\vec t))^{-1} e_1 \succ = \frac{\Delta_{ [a], n-1} (\zeta;\vec t)}{\Delta_{n} (\zeta)} = {\mu}_{a, 0}^{-1} (\vec t)\sum\limits_{j\ge 0} \frac{\mu_{ a,j} (\vec t)}{\zeta^{j+1}},\\
\displaystyle{\mathfrak f}_{[{\hat a}]} (\zeta;\vec t) \equiv \prec e_1 , (\zeta {\mathfrak I}_n - {\mathfrak A}_{[{\hat a}]}(\vec t))^{-1} e_1 \succ =\frac{\Delta_{[{\hat a}],n-1,} (\zeta;\vec t)}{\Delta_{n} (\zeta)} = \mu_{{\hat a},0}^{-1} (\vec t)\sum\limits_{j\ge 0} \frac{\mu_{{\hat a},j} (\vec t)}{\zeta^{j+1}},\\
\Delta_n (z) = \; {\rm det} \; \left( z I - {\mathfrak A}_{[a]} (\vec t) \right) = \; {\rm det} \; \left( z I - {\mathfrak A}_{[{\hat a}]} (\vec t) \right) =\prod\limits_{j=1}^n (z -\kappa_j),
\end{array}
\end{equation}
where $\mu_{a, 0}(\vec t) =\sum_{j=1}^n a_j E_j(\vec t) $, $\mu_{{\hat a}, 0}(\vec t) =\sum_{j=1}^n {\hat a}_j E_j(\vec t) $.
Moreover, let
\begin{equation}\label{eq:gendd}
\begin{array}{l}
\displaystyle {\hat{\mathfrak f}}_{[a]} (\zeta;\vec t) \equiv \prec e_n , (\zeta {\mathfrak I}_n - {\mathfrak A}_{[ a]}(\vec t))^{-1} e_n \succ = \frac{{\hat \Delta}_{[a], n-1} (\zeta;\vec t)}{\Delta_{n} (\zeta)} = {\hat\mu}_{a, 0}^{-1} (\vec t)\sum\limits_{j\ge 0} \frac{{\hat\mu}_{a,j} (\vec t)}{\zeta^{j+1}},\\
\displaystyle {\hat{\mathfrak f}}_{[\hat a]} (\zeta;\vec t) \equiv \prec e_n , (\zeta {\mathfrak I}_n - {\mathfrak A}_{[\hat a]}(\vec t))^{-1} e_n \succ = \frac{{\hat\Delta}_{[\hat a],n-1} (\zeta;\vec t)}{\Delta_{n} (\zeta)} = {\hat \mu}_{\hat a,0}^{-1} (\vec t)\sum\limits_{j\ge 0} \frac{\hat \mu_{\hat a, j} (\vec t)}{\zeta^{j+1}}.
\end{array}
\end{equation}
Then the following relations hold between such Toda hierarchies

\begin{proposition}\label{prop:dualToda}
Let ${\mathcal K} =\{\kappa_1 < \cdots < \kappa_n \}$, $[a], [\hat a], \alpha], [\hat\alpha] \in Gr^{\mbox{\tiny TP}}(1,n)$ such that (\ref{eq:ac}) and Lemma \ref{lemma:ac} holds. Let ${\mathfrak A}_{[a]}(\vec t)$, ${\mathfrak A}_{[\hat a]}(\vec t)$, be as in (\ref{eq:Adual}), with associated Toda flows and generating functions as in (\ref{eq:Todafdual}), (\ref{eq:gendual}) and (\ref{eq:gendd}) .
Then the following relations hold true for all $\vec t$,
\begin{equation}\label{eq:Todadual}
\begin{array}{ll}
{\hat \mu}_{\hat a, 0}(\vec t) =\mu_{a,0} (-\vec t)=\sum_{j=1}^n a_j E_j(-\vec t), &\quad {\hat \mu}_{a, 0}(\vec t) =\mu_{\hat a,0} (-\vec t)=\sum_{j=1}^n {\hat a}_j E_j(-\vec t), \\
{\hat \mu}_{\hat a,j} (\vec t) =\mu_{a,j} (-\vec t),&\quad\displaystyle 
{\hat \mu}_{ a,j} (\vec t) =\mu_{\hat a,j} (-\vec t),\quad\quad \forall j\ge 0,\\
{\hat {\mathfrak f}}_{[a]} (\zeta; \vec t) = {\mathfrak f}_{[\hat a]} (\zeta; -\vec t),&\quad
{\hat {\mathfrak f}}_{[\hat  a]} (\zeta; \vec t) ={\mathfrak f}_{[a]} (\zeta; -\vec t),\\
{\hat \Delta}_{[a],k} (\zeta; \vec t) =  \Delta_{[\hat a],k} (\zeta; -\vec t)&\quad
\Delta_{[a],k}(\zeta; \vec t) =  {\hat \Delta}_{[\hat a],k} (\zeta; -\vec t), \quad \forall k\in [1,n],\\
{\mathfrak a}_{[\hat a],n-k} (\vec t) = {\mathfrak a}_{[a],k} (-\vec t), \quad k\in [n-1],&\quad
{\mathfrak b}_{[\hat a],n-k} (\vec t) = {\mathfrak b}_{[a],k+1} (-\vec t), \quad k\in [0,n-1].
\end{array}
\end{equation}
\end{proposition}

The proof trivially follows from Propositions \ref{prop:aaa}, \ref{prop:Toda} and Corollary \ref{cor:ac}.
In conclusion the dual initial data $({\mathcal K}, [a])$ and $({\mathcal K}, [\hat a])$, with $[a]$ related to $[\hat a]$ by (\ref{eq:adual}) generate dual Toda hierarchy solutions which satisfy
{\small\[
\left(\begin{array}{cccc}
\mathfrak{b}_{[\hat a],1} (\vec t) & \mathfrak{a}_{[\hat a],1} (\vec t) & 0              & \cdots                \\
1              & \mathfrak{b}_{[\hat a],2} (\vec t) & \mathfrak{a}_{[\hat a],2} (\vec t) &\ddots                 \\
0              & \ddots         & \ddots       & \ddots                  \\
0              & \cdots         & 1            & \mathfrak{b}_{[\hat a],n} (\vec t)
\end{array}\right)
=
\left(\begin{array}{ccccc}
\mathfrak{b}_{[a],n} (-\vec t) & \mathfrak{a}_{[a],n-1} (-\vec t) & 0              & \cdots            \\
1              & \mathfrak{b}_{[a],n-1} (-\vec t) & \mathfrak{a}_{[a],n-2} (-\vec t) &\ddots           \\
0              & \ddots         &\ddots              & \ddots                  \\
0              & \cdots         & 1                  & \mathfrak{b}_{[a],1} (-\vec t)
\end{array}\right).
\]
}
\subsection{Duality and divisors of KP--soliton solutions}

Lemma \ref{lemma:ac} implies the following: for any given set of phases $\mathcal K = \left\{ \kappa_1 < \kappa_2 < \cdots < \kappa_n \right\}$ and
any given point $[a]  \in Gr^{\mbox{\tiny TP}} (1,n)$, there exists $[\hat a]\in Gr^{\mbox{\tiny TP}} (1,n)$, satisfying (\ref{eq:ac}) with $ \hat \alpha_j = 1/\alpha_j$, $j\in [n]$, such that the vacuum wavefunctions 
\begin{equation}\label{eq:psia}
\begin{array}{l}
\displaystyle
\Psi_{[a]} (\zeta; \vec t ) = \left\{ \begin{array}{ll}
\displaystyle e^{\theta (\zeta;\vec t)}, & \quad \mbox{ if } \zeta\in \Gamma_+,\\
\Psi^{(-)}_{[a]} (\zeta; \vec t) \equiv
\sum\limits_{l=1}^{n} \frac{a_j\prod\limits_{s\not =l} (\zeta -\kappa_s)}{\left(\sum\limits_{m=1}^n a_m \right)\prod\limits_{r=1}^{n-1} (\zeta-b_{[a],r })} E_j( \vec t), & \quad \mbox{ if }
\zeta\in\Gamma_-,
\end{array}
\right.
\\
\\
\displaystyle
\Psi_{[\hat a]} (\zeta; \vec t ) = \left\{ \begin{array}{ll}
\displaystyle e^{\theta (\zeta;\vec t)}, & \quad \mbox{ if } \zeta\in \Gamma_+,\\
\Psi^{(-)}_{[\hat a]} (\zeta; \vec t) \equiv
\sum\limits_{l=1}^{n} \frac{\hat a_j\prod\limits_{s\not =l} (\zeta -\kappa_s)}{\left(\sum\limits_{m=1}^n \hat a_m \right)\prod\limits_{r=1}^{n-1} (\zeta-b_{[\hat a],r})} E_j( \vec t), & \quad \mbox{ if }
\zeta\in\Gamma_-,
\end{array}
\right.
\end{array}
\end{equation}
generate dual $(k,n-k)$ and $(n-k,k)$--solitons respectively via the Darboux transformations $D^{(k)}_{[a]}$ and $D^{(n-k)}_{[\hat a]}$ for any $k\in [n-1]$.
From (\ref{eq:Psik}) and (\ref{eq:Todadual}), the vacuum divisors satisfy respectively
$\prod_{r=1}^{n-1} (\zeta - b_{[a],r} ) \equiv \Delta_{n-1, {[a]}} (\zeta; \vec 0) = {\hat \Delta}_{n-1, {[\hat a]}} (\zeta; \vec 0)$, $\prod_{r=1}^{n-1} (\zeta - b_{[\hat a],r} ) \equiv \Delta_{n-1, {[\hat a]}} (\zeta; \vec 0) = {\hat \Delta}_{n-1, {[a]}} (\zeta; \vec 0)$.

Choosing the representative elements of the dual soliton solutions as in (\ref{eq:ac}), with
$\hat\alpha_j = \alpha_j^{-1}$, $j\in [n]$, the dual vacuum divisor  $\left\{b_{[\hat a],1}<\cdots <b_{[\hat a],n-1}\right\}$ may be explicitly computed solving the following system of equations
\begin{equation}\label{eq:divac}
\prod\limits_{r=1}^{n-1} (\kappa_j -b_{[\hat a],r}) = \frac{(-1)^{n-j}\hat \alpha_j}{\sum\limits_{m=1}^n \hat a_m}=\frac{(-1)^{n-j}\prod\limits_{m=1}^n \alpha_m}{\alpha_j\tau^{(n-1)}_{a} (\vec 0)}, \quad\quad \forall j\in [n].
\end{equation}
In particular, the space--time inversion leaves the vacuum divisor invariant if and only if $b_{[\hat a],r} =b_{[ a],r}$, $r\in [n-1]$, which is equivalent to $[\hat  a] = [a]= [1/\hat  a]$, that is
\[
[\alpha_1,\dots,\alpha_n] = [\hat\alpha_1,\dots,\hat\alpha_n] = [1,\dots,1].
\]
We have thus proven

\begin{corollary}
Let ${\mathcal K} = \{ \kappa_1 < \cdots <\kappa_n \}$,
 $[a], [\hat a] \in Gr^{\mbox{\tiny TP}} (1,n)$,  $\Psi_{[a]} (\zeta; \vec t )$, $\Psi_{[\hat a]} (\zeta; \vec t )$,  as in (\ref{eq:ac}) and (\ref{eq:psia}), where
 $[\alpha],[\hat\alpha] \in Gr^{\mbox{\tiny TP}} (1,n)$ satisfy Lemma \ref{lemma:ac}.
Then the space--time inversion leaves the 0--divisor invariant, $b_{[a],r} =b_{[\hat a],r}$, $r\in [n]$, if and only if $[\alpha]=[\hat\alpha] = [1,\dots, 1]$.
\end{corollary}

The self--dual $(k,k)$--soliton solutions $u (\vec t)= u (-\vec t)$ are thus associated to $[\alpha]= [1,\dots,1]\in Gr^{\mbox{\tiny TP}} (1,2k)$. 

\begin{figure}
\includegraphics[scale=0.3]{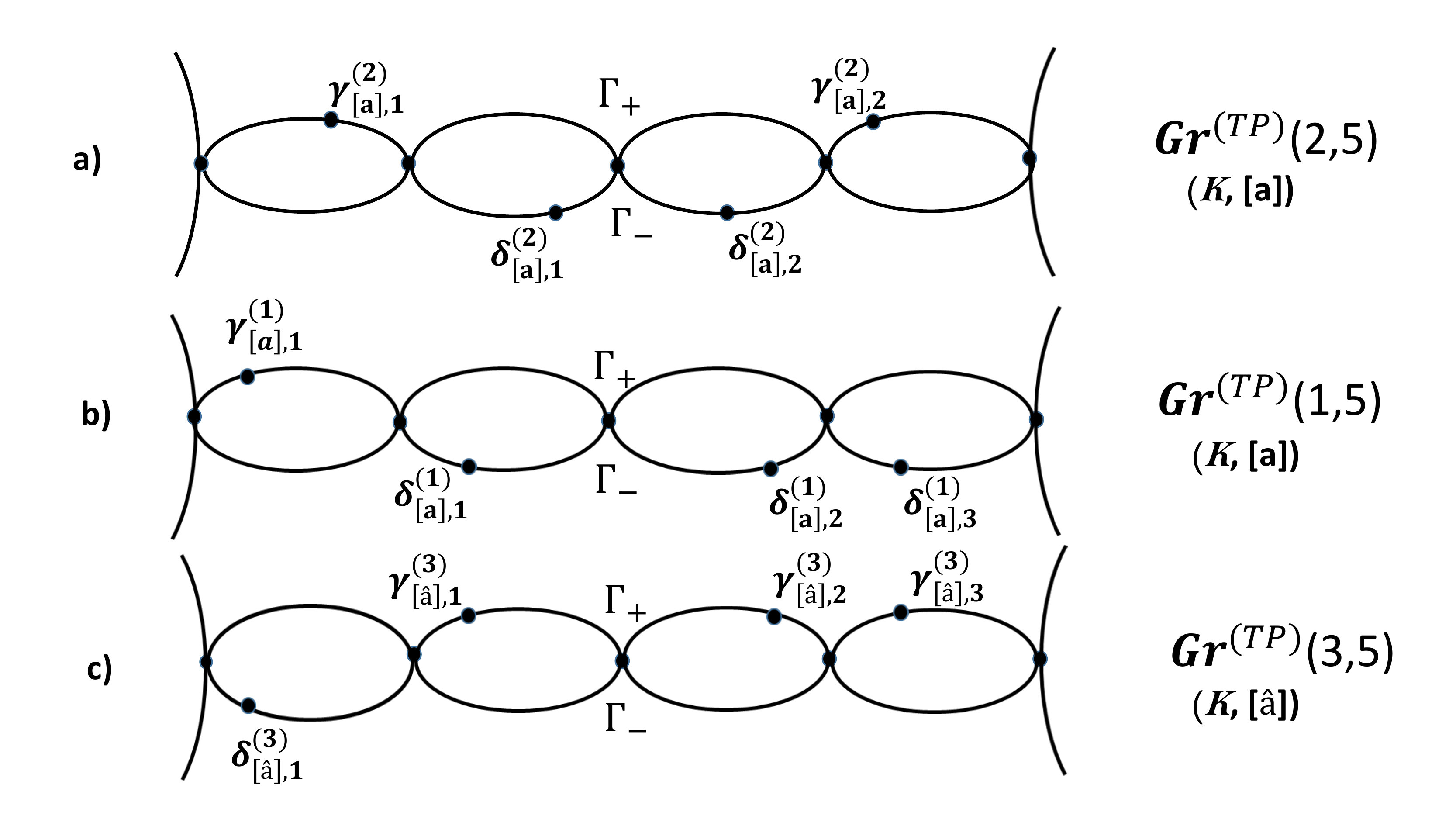}
\vspace{-.5 truecm}
\caption{ 
{\small{\sl We illustrate Theorem \ref{theo:dualdiv} on dual $T$--hyperelliptic divisors in $Gr^{\mbox{\tiny TP}} (2,5)$ and in $Gr^{\mbox{\tiny TP}} (3,5)$. a): ${\mathcal D}^{(2)}_{[a]} =\{ \gamma^{(2)}_{[a],1}, \gamma^{(2)}_{[a],2},\delta^{(2)}_{[a],1}, \delta^{(2)}_{[a],2}\}$ is the divisor of the (3,2)--line soliton in $Gr^{\mbox{\tiny TP}}(2,5)$ associated to the soliton data $({\mathcal K}, [a])$. 
b): ${\mathcal D}^{(1)}_{[a]}=\{ \gamma^{(1)}_{[a],1}, \delta^{(1)}_{[a],1},\delta^{(1)}_{[a],2}, \delta^{(1)}_{[a],3}\} $ is the divisor of the (4,1)--line soliton  in $Gr^{\mbox{\tiny TP}}(1,5)$ associated to the soliton data $({\mathcal K} ,[a])$.
c): ${\mathcal D}^{(3)}_{[\hat a]} =\{ \gamma^{(3)}_{[\hat a],1}, \gamma^{(3)}_{[\hat a],2},\gamma^{(3)}_{[\hat a],3}, \delta^{(3)}_{[\hat a],1}\}$  is the divisor for the (2,3)--line soliton solution in $Gr^{\mbox{\tiny TP}}(3,5)$ associated to the dual soliton data $({\mathcal K}, [\hat a])$ via space--time inversion and it is obtained applying the hyperelliptic involution $\sigma$ to ${\mathcal D}^{(1)}_{[a]}$: ${\mathcal D}^{(3)}_{[\hat a]} = \sigma ({\mathcal D}^{(1)}_{[ a]})$.}}}
\label{fig:gr1}
\end{figure}

Let us now return to the general case of dual $T$--hyperelliptic soliton data. 
The following theorem explains the relations between $k$--compatible divisors under the space--time inversion. Indeed, using (\ref{eq:divToda}) and (\ref{eq:Todadual}) we have the following.

\begin{theorem}\label{theo:dualdiv}
Let $n>1$, $\mathcal K =\left\{ \kappa_1 < \cdots < \kappa_n\right\}$ and $[\alpha_1,\dots,\alpha_n]\in Gr^{\mbox{\tiny TP}} (1,n)$
be given. Let $[{\hat\alpha}_1,\dots,{\hat\alpha}_n] =\left[ 1/\alpha_1,\dots,1/\alpha_n\right]$ and $[a],[\hat a]\in Gr^{\mbox{\tiny TP}} (1,n)$ as in (\ref{eq:ac}). Let $\Psi_{[a]}(\zeta; \vec t)$ and $\Psi_{[\hat a]}(\zeta; \vec t)$ be the dual vacuum wavefunctions as in (\ref{eq:psia}). 

For any given $k\in [n-1]$, let $D^{(k)}_{[a]}$, $D^{(n-k)}_{[\hat a]}$, respectively be the dual Darboux transformations as in (\ref{eq:dabac}). 
Let $\sigma$ be the hyperelliptic involution on $\Gamma$, {\sl i.e.} $\sigma (\Gamma_{\pm}) = \Gamma_{\mp}$. Let, for any fixed $k\in [n-1]$
\begin{equation}\label{eq:dualdiv}
{\mathcal D}^{(k)}_{[ a]} = {\mathcal D}^{(k)}_{[a], +} \cup {\mathcal D}^{(k)}_{{[a]}, -}, \quad\quad {\mathcal D}^{(k)}_{[\hat a]} = {\mathcal D}^{(k)}_{[\hat a], +} \cup {\mathcal D}^{(n-k)}_{[\hat a], -} ,
\end{equation}
be the pole divisors respectively of $\tilde \Psi_{[ a]}^{(k)} (\zeta; \vec t) \equiv \frac{D^{(k)}_{[a]} \Psi_{[ a]} (\zeta; \vec t)}{D^{(k)}_{[a]} \Psi_{[ a]} (\zeta; \vec 0)}$, $\tilde \Psi_{[\hat a]}^{(k)} (\zeta; \vec t) \equiv \frac{D^{(k)}_{[\hat a]} \Psi_{[\hat a]} (\zeta; \vec t)}{D^{(k)}_{[\hat a]} \Psi_{[\hat a]} (\zeta; \vec 0)}$, 
where
\[
\begin{array}{l}
{\mathcal D}^{(k)}_{[a], +} \equiv \left\{  \gamma_{[a],1}^{(k)}, \dots , \gamma_{[a],k}^{(k)}   \right\} = \{ (\zeta;\mu )\in \Gamma_+ \, : \, {\hat \Delta}_{k,[a]} (\zeta; \vec 0) =0\}, \\
{\mathcal D}^{(k)}_{[a], -} \equiv \left\{  \delta_{[a],1}^{(k)}, \dots , \delta_{[a],n-k+1}^{(k)}   \right\} = \{ (\zeta;\mu )\in \Gamma_- \, : \, \Delta_{n-k-1,[a]} (\zeta; \vec 0)=0 \},
\\
{\mathcal D}^{(k)}_{[\hat a], +} \equiv \left\{  \gamma_{[\hat a],1}^{(k)}, \dots , \gamma_{[\hat a],k}^{(k)}   \right\} = \{ (\zeta;\mu )\in \Gamma_+ \, : \, {\hat \Delta}_{k,[\hat a]} (\zeta; \vec 0) =0 \}, \\
{\mathcal D}^{(k)}_{[\hat a], -} \equiv \left\{  \delta_{[\hat a],1}^{(k)}, \dots , \delta_{[\hat a],n-k+1}^{(k)}   \right\} = \{ (\zeta;\mu )\in \Gamma_- \, : \, \Delta_{n-k-1,[\hat a]} (\zeta; \vec 0) =0\}.
\end{array}
\]
Then for any fixed $k\in [n-1]$, 
\begin{equation}\label{eq:invdiv}
{\mathcal D}^{(n-k)}_{[\hat a], +} = \sigma \left( {\mathcal D}^{(k-1)}_{[a], -} \right),\quad\quad {\mathcal D}^{(n-k)}_{[\hat a], -} = \sigma \left( {\mathcal D}^{(k-1)}_{[a], +} \right).
\end{equation}
In particular, if $k=1$,
${\mathcal D}^{(n-1)}_{[\hat a]} = {\mathcal D}^{(n-1)}_{[\hat a], +} \equiv   \{\sigma ( b_{[a],1}), \dots , \sigma( b_{[a],n-1})\}.$
\end{theorem}

\begin{corollary}
Under the hypotheses of the above theorem, for any fixed $k\in [n-1]$ the pole divisor of ${\tilde \Psi}^{(n-k)}_{[\hat a]} (\zeta;\vec t)$ may be computed from the pole divisor of ${\tilde \Psi}^{(k)}_{[a]} (\zeta;\vec t)$, 
\[
\prod\limits_{j=1}^n (\zeta -\kappa_j) = \prod\limits_{l=1}^{k} (\zeta - \gamma^{(k)}_{[a], l}) \prod\limits_{s=1}^{n-k} (\zeta - \gamma^{(n-k)}_{[\hat a],s})
-{\mathfrak a}_{[a],k} (\vec 0) \prod\limits_{i=1}^{k-1} (\zeta - \delta^{(n-k)}_{[\hat a],i})\prod\limits_{r=1}^{n-k-1} (\zeta - \delta^{(k)}_{[a],r}).
\]
In particular, the dual compatible divisors ${\mathcal D}^{(n-k)}_{[\hat a]}$ and  ${\mathcal D}^{(k)}_{[a]}$ satisfy
\begin{equation}\label{eq:const}
\frac{ \prod\limits_{l=1}^{k} \left( \kappa_j - \gamma_{[a],l}^{(k)} \right) \cdot \prod\limits_{i=1}^{n-k} \left( \kappa_j - \gamma_{[\hat a],i}^{(n-k)} \right) }{\prod\limits_{s=1}^{n-k-1} \left( \kappa_j - \delta_{[a],s}^{(k)} \right) \cdot \prod\limits_{r=1}^{k-1} \left( \kappa_j - \delta_{[\hat a],r}^{(n-k)} \right)} = {\mathfrak a}_{[a],k} (\vec 0), \quad\quad\forall j\in [n]
\end{equation}
Moreover, if $[\hat a] = [a] = [1,\dots,1]\in Gr^{\mbox{\tiny TP}}(1,n)$, then 
\[
{\mathcal D}^{(n-k)}_{[ a], +} = \sigma \left( {\mathcal D}^{(k-1)}_{[ a], -} \right),\quad\quad {\mathcal D}^{(n-k)}_{[ a], -} = \sigma \left( {\mathcal D}^{(k-1)}_{[a], +} \right), \quad \quad \forall k\in [n].
\]
\end{corollary}

\begin{proof}
It is sufficient to insert (\ref{eq:invdiv}) into the identities in Corollary \ref{cor:formulas}. 
for all $j\in [n]$, such that $\kappa_j \not \in {\mathcal D}^{(k)}_{\alpha} \cup {\mathcal D}^{(n-k)}_{\beta}$.
If $\kappa_j \in {\mathcal D}^{(k)}_{\alpha}$, that is $\kappa_j = \gamma_{\alpha,{\bar l}}^{(k)}=\delta_{\alpha,{\bar s}}^{(k)}$, in
(\ref{eq:const}) the factors corresponding to $s= \bar s$ and $l=\bar l$ are omitted and substituted by $(-1)$. Similarly if 
$\kappa_j \in {\mathcal D}^{(n-k)}_{\beta}$.
\end{proof}

In Figure \ref{fig:gr1}, we show the first non trivial example of such duality relation between $T$--hyperelliptic solitons in $Gr^{\mbox{\tiny TP}} (2,5)$ and $Gr^{\mbox{\tiny TP}} (3,5)$.

\section{Summary and concluding remarks}\label{sec:summary}

In this paper we have characterized the KP--soliton data $({\mathcal K}, [A])$, ${\mathcal K} =\{ \kappa_1<\cdots< \kappa_n\}$, $[A]\in Gr^{\mbox{\tiny TNN}} (k,n)$, which are compatible with a real divisor structure on a given rational degeneration of a real hyperelliptic curve. Indeed we start from the ansatz that
\[
\Gamma =\Gamma_+\sqcup\Gamma_- = \{ \eta^2 = \prod_{j=1}^N ( \zeta - \kappa_j)^2 \}, 
\]
is obtained from $\Gamma^{(\epsilon)} =\{ \eta^2 = \prod_{j=1}^N ( \zeta - \kappa_j)^2 -\epsilon^2 \}$ when $\epsilon^2\to 0$
and that the divisor  ${\mathcal D} =\{ P_1,\dots,P_{n-1} \}\subset\Gamma\backslash\{ P_+\}$ is the limit of a $(n-1)$--point non special divisor $(P^{(\epsilon)}_1,\dots, P^{(\epsilon)}_{n-1})$ on $\Gamma^{(\epsilon)}\backslash \{ P_+\}$ satisfying Dubrovin--Natanzon reality conditions.
For that reason we impose the finite oval condition  $\zeta(P_j)\in [\kappa_j, \kappa_{j+1} ]$, $\forall j\in [n-1]$ (see Figure \ref{fig:epsilon}). By construction such a KP soliton solution is regular and bounded for all times.

Here $P_+\in \Gamma_+$ is the marked point where the KP Baker--Akhiezer function has its 
essential singularity, $k$ divisor points, say $P_1,\dots,P_k$ belong to $\Gamma_+$, and the remaining $n-k-1$ are in $\Gamma_-$, for some $k\in [0,n-1]$. In agreement with vacuum dressing, there exist a finite inverse gauge operator $W = 1-w_1(\vec t) \partial_x^{-1}-\cdots w_k(\vec t) \partial_x^{-k}$ and $k$ linearly independent solutions to the heat hierarchy, $f^{(1)} (\vec t),\dots, f^{(k)} (\vec t)$, which form a basis of the linear operator (Darboux transformation) $D^{(k)} = W\partial_x^{k}$.

Such algebraic geometric structure is also compatible with the ansatz that the soliton solution corresponds to a point in $Gr(k,n)$ (the finite dimensional reduction of the Sato Grassmannian) and that the double points ${\mathcal K} = \{\kappa_1<\dots<\kappa_n\}$ are related to the phases of the KP soliton solutions, via the heat hierarchy solutions appearing in its $\tau$--function. 

Finally, according to \cite{KW2}, KP solitons corresponding to data in $Gr(k,n)$ (the so--called $(n-k,k)$--line solitons) are real regular bounded for all times $(x,y,t)$ if and only if the soliton data belong to the totally non--negative part of the Grassmannian $Gr^{\mbox{\tiny TNN}}(k,n)$.

For all of the above reasons we have investigated which soliton data $(\mathcal K, [A]  )$, with $[A] \in Gr^{\mbox{\tiny TNN}}(k,n)$ are compatible with the divisor structure on $\Gamma$ (Sections \ref{sec:vacuumKP}, \ref{sec:Tsol}). In particular we have proven that $\Gamma$ is a desingularization of $\Gamma_{\xi}$ as in \cite{AG} for points in $Gr^{\mbox{\tiny TP}} (n-1,n)$ (Section \ref{sec:desing}). 

More precisely, for any fixed $\mathcal K =\{\kappa_1 < \cdots < \kappa_n \}$ and $k\in [0, n-1]$, we have called ${\mathcal D} = \{ P_1, \dots ,P_{n-1} \}\subset {\Gamma}\backslash \{ P_+\}$, with $\Gamma$ as in (\ref{eq:G}),  a {\sl $k$--compatible divisor} if 
\begin{enumerate}
\item there is exactly one pole in each finite oval: $\zeta (P_j) \in [k_j, k_{j+1}]$, $\forall j\in [n-1]$;
\item exactly $k$ divisor points points belong to $\Gamma_+$.
\end{enumerate}
In figure \ref{fig:gen3}, we show possible $k$--compatible divisor configurations for $n=4$ and $k=3,2,1$. 

We have called {\sl $T$--hyperelliptic} the soliton data $(\mathcal K, [A])$, $[A]\in Gr^{\mbox{\tiny TNN}} (k,n)$, which produce $k$--compatible divisors on $\Gamma$. 
To identify  $T$--hyperelliptic soliton data in $Gr^{\mbox{\tiny TNN}} (k,n)$, we have proceeded in two steps:
\begin{enumerate}
\item On $\Gamma$ we have defined a vacuum wavefunction $\Psi(\zeta; \vec t)$, which coincides
with Sato vacuum wavefunction on $\Gamma_+$. Vacuum divisors are $0$--compatible by construction and are
in bi-jection with points $[a] \in Gr^{\mbox{\tiny TP}} (1,n)$ (see Lemmata \ref{prop:0} and \ref{lemma:inv}); 
\item  Then we have applied the Darboux transformation $D^{(k)}$ generated by the soliton data $(\mathcal K, [A])$ and 
we have required that there exists a vacuum divisor $[a] \in Gr^{\mbox{\tiny TP}} (1,n)$ such that the divisor of the normalized KP--wavefunction $\displaystyle \frac{D^{(k)}\Psi(\zeta; \vec t)}{D^{(k)}\Psi(\zeta; \vec 0)}$ is $k$--compatible. Then
\begin{enumerate}
\item If $k=1,n-1,$ we have obtained a bi--jection between vacuum divisors  $[a] \in Gr^{\mbox{\tiny TP}} (1,n)$ and soliton data $[A] \in Gr^{\mbox{\tiny TP}} (1,n), Gr^{\mbox{\tiny TP}} (n-1,n)$;
\item In Lemma \ref{lemma:pos} and Theorem \ref{theo:divisor}, for any fixed $k \in ]1,n-1[$,  we have proven that the soliton data $(\mathcal K, [A])$ are  $T$--hyperelliptic if and only if $[A]\in Gr^{\mbox{\tiny TP}} (n-1,n)$ with representative matrix
\begin{equation}\label{eq:AA}
A^i_j = \kappa_j^{i-1} {  a}_j, \quad\quad i\in [k], \; \; j\in [n].
\end{equation}
\end{enumerate}
\end{enumerate}
For any fixed $\mathcal K$ and $k\in ]1,n-1[$, such soliton data parametrize an $(n-1)$--dimensional real connected variety in the $k(n-k)$ dimensional $Gr^{\mbox{\tiny TP}} (k,n)$. For instance, varying the position of the $2$--compatible divisor
in Figure \ref{fig:epsilon}.b), we parametrize a 3--dimensional variety of KP soliton data in $Gr^{\mbox{\tiny TP}} (2,4)$. 

\medskip

Soliton data satisfying (\ref{eq:AA}) are naturally connected with the solutions to the finite non--periodic Toda 
lattice hierarchy (\ref{eq:Todafj}) since, for any fixed $k\in [n-1]$, as observed in \cite{BK}, the $\tau$--function generating such KP $(n-k,k)$--soliton solution $u_{KP}(\vec t)= 2\partial_x^2 \log \tau^{(k)} (\vec t) $
\[
\tau^{(k)} (\vec t) = \mbox{Wr}_x \left ( \mu_0(\vec t), \partial_x \mu_0(\vec t), \dots, \partial_x^{k-1} \mu_0(\vec t)\right),
\quad\quad \mu_0(\vec t) = \sum\limits_{j=1}^n a_j \exp(\kappa_j x +\kappa_j^2 y +\kappa_j^3 t +\cdots),
\]
is also the $\tau$--function associated to the solutions to the finite non--periodic Toda 
lattice hierarchy (\ref{eq:Todafj}) for the Toda datum $({\mathcal K}, [a])$, where ${\mathcal K}$ is the Toda spectrum and $[a]\in Gr^{\mbox{\tiny TP}}(1,n)$ uniquely identifies the initial value Toda problem.

We have then discussed the relation between the spectral problem associated to KP $T$-hyperelliptic $(n-k,k)$-line KP solitons and the open Toda lattice. The spectral curve proposed for the open Toda lattice in \cite{McK, Mar, BM} is determined by the equation
${\hat \eta} = \prod\limits_{j=1}^n (\zeta -\kappa_j)$,
considered as the limit $\epsilon\to 0$ of the hyperelliptic spectral curve 
${\hat \eta} + \frac{\epsilon^2}{4\hat \eta} = \prod\limits_{j=1}^n (\zeta -\kappa_j)$,
of the periodic Toda system. In \cite{KV}, the Baker-Akhiezer function approach is used to provide a solution to the Toda inverse spectral problem for the singular curve (\ref{eq:Todacurve}).

The rational curve\footnote{Notice that also the regular hyeperelliptic curves are the same, so it is natural to expect relations between Toda periodic hierarchy and regular real--periodic KP solutions on $\Gamma^{(\epsilon)}$.} is the same for both finite Toda and KP $T$--hyperelliptic solitons and the same $\tau$--functions govern the solutions in both cases. Then, it is natural to expect a relation between the divisor structure of the two problems. In the rational setting, the two systems may be distinguished from one another from the different asymptotics at the essential singularities, asymptotics which is modeled on the algebraic geometric regular periodic setting for the two systems. In sections \ref{sec:TodaKP} and \ref{sec:invKP} we have investigated the relations between the two problems. 

Here, to any KP soliton data $({\mathcal K}, [a])$, ${\mathcal K} = \{ \kappa_1< \cdots \kappa_n \}$ and $[a] \in Gr^{\mbox{\tiny TP}} (1,n)$: 
\begin{enumerate}
\item we associate the vacuum KP spectral data $({\mathcal K}, {\mathcal B})$, where ${\mathcal B} = \{ b_1 < \cdots < b_{n-1}\}$ is the vacuum divisor and satisfies $b_j \in ]\kappa_j, \kappa_{j+1}[$, for all $j\in [n-1]$ and the vacuum KP--wavefunction $\Psi(\zeta;\vec t)$ as in (\ref{eq:1M});
\item for any fixed $k\in [n-1]$, we associate a Darboux transformation $D^{(k)}$, the $T$-hyperelliptic $(n-k,k)$-line KP soliton solution $({\mathcal K}, [A])$, with $[A]\in Gr^{\mbox{\tiny TP}}(k,n)$ as in (\ref{eq:AA}), the $k$--compatible spectral data $({\mathcal K}, {\mathcal D}^{(k)})$ and the normalized KP--wavefunction $\frac{\Psi^{(k)}_{{\mathcal B}} (\zeta; \vec t)}{\Psi^{(k)}_{{\mathcal B}} (\zeta; \vec 0)}$, with $\Psi^{(k)}_{{\mathcal B}}$ as in (\ref{eq:psiN}).
\end{enumerate}

The same data $({\mathcal K}, [a])$ parametrize the solutions to the IVP of the finite non--periodic Toda hierarchy in the real configuration space (\ref{eq:confsp}).  
The results in \cite{KV} may be restated as follows:
\begin{enumerate}
\item $({\mathcal K}, [a])$ are in bi--jection with Toda spectral data  $({\mathcal K}, {\mathcal B}^{(T)})$, with
${\mathcal B}^{(T)} =\{ b_1^{(T)} < \cdots < b_{n-1}^{(T)} \}$, $b_j^{(T)}\in ]\kappa_j,\kappa_{j+1}[$. 
\item The time--dependent Baker--Akhiezer Toda vectors
$\Psi_k^{(T)} (\zeta; \vec t) = e^{\zeta/2} \left( \sum\limits_{i=0}^k c_i(\vec t, k) \zeta^i \right)$, $\Psi_k^{(T),\sigma} (\zeta; \vec t) = e^{-\zeta/2} \left( \frac{\sum_{l=0}^{n-1-k} c_l^{\sigma}(\vec t, k) \zeta^l  }{\prod_{s=1}^{n-1} (\zeta -b^{(T)}_s)}\right)$, $k\in [0,n-1]$, satisfy the gluing conditions $\Psi^{(T)}_j (\kappa^{(T)}_l, \vec t) = \Psi^{(T),\sigma}_j (\kappa^{(T)}_l, \vec t)$, and necessary normalization relations\cite{KV}.
\end{enumerate}

In Theorem \ref{theo:KPToda}), we have expressed the normalized KP wavefunction associated to the $T$--hyperelliptic $(n-k,k)$--solitons in function of the entries of the Toda resolvent ${\mathfrak R}=(\zeta {\mathfrak I}_n - {\mathfrak A})^{-1}$ for the 
data $(\mathcal K, [a])$. It then follows (Corollary \ref{cor:TodaKP}):
\begin{enumerate}
\item the KP vacuum divisor and the Toda divisor coincide for the datum $(\mathcal K, [a])$: ${\mathcal B}={\mathcal B}^{(T)}$;
\item the Darboux transformation associated to $({\mathcal K}, [A])$ is the differential operator associated to Toda $k$--minors ${\hat \Delta}^{(k)} (\zeta,\vec t)$;
\item the KP divisor of the $T$--hyperelliptic $(n-k,k)$--soliton, ${\mathcal D}^{(k)} =  \{ \gamma^{(k)}_1,\dots, \gamma_k^{(k)} , \delta^{(k)}_1,\dots, \delta^{(k)}_{n-k-1}\}$ is the zero divisor at times $\vec t\equiv \vec 0$, of the Toda Baker--Akhiezer components $\Psi^{(T)}_k$ and $\Psi^{(T),\sigma}_k $ (see (\ref{eq:divTKP})).
\end{enumerate}
Then we have used the identities among the entries of the Toda resolvent ${\mathfrak R}$ as recursive sets of equations which allow to compute both the $k$--compatible divisors of the KP wavefunction and the KP Darboux transformations $D^{(k)}$, as $k$ varies from 1 to $n-1$ (see Proposition \ref{prop:Toda} and Corollary \ref{cor:formulas}).
We have also solved the inverse problem: in Theorem \ref{theo:divfor}, given a $k$--compatible divisor on $\Gamma$, we explicitly reconstruct the soliton data $[a]$.

\medskip

The space--time inversion $\vec t\mapsto -\vec t$ allows to map
$(n-k,k)$-soliton KP solutions to $(k,n-k)$-soliton KP solutions and induces a duality relation between
the Grassmann cells in $Gr(k,n)$ and $Gr(n-k,n)$ (see \cite{CK1} and references therein). 

The space--time inversion leaves $\Gamma$ invariant since it preserves the spectrum of ${\mathfrak A}$. For the Toda hierarchy, it
corresponds to the composition of the space--time inversion with the reflection of the entries of ${\mathfrak A}$ with respect to the anti--diagonal. This correspondence for Toda is natural in view of the asymptotic behaviour of ${\mathfrak A}(\vec t)$
when all times go either to $+\infty$ or to $-\infty$, which generalizes a result in \cite{Mos} for the first flow to the Toda hierarchy.

As a consequence, the entries ${\mathfrak R}_{1,1}(\vec t)$ and ${\mathfrak R}_{n,n} (-\vec t)$ of the resolvent of $\mathfrak A (\vec t)$ are the generating functions for such dual Toda hierarchies (Proposition \ref{prop:dualToda}), and are also associated to dual families of $T$--hyperelliptic soliton solutions.

We have used such correspondence to obtain explicit relations among the dual divisors. 
In particular if the initial $k$--compatible divisor in $(Gr^{\mbox{\tiny TP}}(k,n)$ is associated to $[\alpha]\in Gr^{\mbox{\tiny TP}}(1,n)$, with $[ \alpha]$ as in (\ref{eq:ac}), then the dual divisor is $(n-k)$--compatible in $(Gr^{\mbox{\tiny TP}}(n-k,n)$ and associated to $[1/ \alpha]$ (see Corollary \ref{cor:ac}). We have shown that such $(n-k)$--compatible dual divisor  may be explicitly computed by applying the hyperelliptic involution to the $(k-1)$--compatible divisor associated to $[  \alpha]$ (see Theorem \ref{theo:dualdiv} and Figure \ref{fig:gr1}).

\medskip

$T$--hyperelliptic solitons are not the unique family of KP solitons which may be associated to rational degenerations of hyperelliptic curves. In \cite{AG1}, we associate the rational degeneration of a regular $\mathtt M$--curve to generic soliton data $(\mathcal K, [A])$, with $[A]\in Gr^{\mbox{\tiny TNN}} (k,n)$ and we plan to classify other families of KP solitons associated to rational degenerations of hyperelliptic or trigonal curves. We also plan to investigate other connection between KP theory and Toda type systems in future.

\bibliographystyle{alpha}

\end{document}